\documentclass[10pt]{amsart}
\usepackage{lineno}
 \usepackage{geometry}
\usepackage{fullpage}
\usepackage{tikz}
\usepackage{tikz-cd}
\usetikzlibrary{matrix,arrows,shapes,cd}
\usepackage{pgf}
\usepackage{url}
\usepackage{longtable}
\usepackage{enumerate}
\usepackage{amsthm}
\usepackage{amssymb}
\usepackage{amsmath}
\usepackage{amscd}
\usepackage{eucal}
\usepackage{mathrsfs}
\usepackage{upgreek}
\usepackage{mathtools}
\usepackage{dsfont}
\usepackage{yfonts}
\usepackage[T1]{fontenc}
\usepackage{bbm}
\usepackage{graphics}
\usepackage{array}
\usepackage[retainorgcmds]{IEEEtrantools}
\usepackage{booktabs}
\usepackage{enumerate}
\usepackage{feynmf}
\usepackage{amsaddr}

\usetikzlibrary{trees}
\usetikzlibrary{decorations.pathmorphing}
\usetikzlibrary{decorations.markings}
\tikzset{
	photon/.style={decorate, decoration={snake}, draw=red},
	electron/.style={draw=blue, postaction={decorate},
		decoration={markings,mark=at position .55 with {\arrow[draw=blue]{>}}}},
	gluon/.style={decorate, draw=magenta,
		decoration={coil,amplitude=4pt, segment length=5pt}},
	sderiv/.style={postaction={decorate},
		decoration={markings,mark=at position .3 with {\arrow{>}}}},
	tderiv/.style={postaction={decorate},
		decoration={markings,mark=at position .7 with {\arrow{<}}}},
	stderiv/.style={postaction={decorate},
		decoration={markings,mark=at position .7 with {\arrow{<}},mark=at position .3 with {\arrow{>}}}}
}
\definecolor{see}{RGB}{67,75,179}
\definecolor{darksee}{RGB}{42,44,148}
\definecolor{honey}{RGB}{232,180,129}
\definecolor{lighthoney}{RGB}{255,254,220}
\definecolor{citecol}{rgb}{0.5,0,0} 
\definecolor{blue1}{RGB}{130,150,209}
\newenvironment{narrow}[2]{%
	\begin{list}{}{%
			\setlength{\topsep}{0pt}%
			\setlength{\leftmargin}{#1}%
			\setlength{\rightmargin}{#2}%
			\setlength{\listparindent}{\parindent}%
			\setlength{\itemindent}{\parindent}%
			\setlength{\parsep}{\parskip}}%
		\item[]}{\end{list}}
\DeclareSymbolFont{bbold}{U}{bbold}{m}{n}
\DeclareSymbolFontAlphabet{\mathbbold}{bbold}
\usepackage{soul}
\usepackage[colorinlistoftodos,textsize=footnotesize,color=blue!30,textwidth=2cm]{todonotes}

\setstcolor{red}
\definecolor{see}{RGB}{67,75,179}
\usepackage{hyperref}
\hypersetup
{
	bookmarks=true,
	bookmarksopen=true,
	pdftitle=PAQFT,
	pdfauthor=Kasia Rejzner,
	pdfsubject=An introduction for mathematicians, %subject of the document
	pdfmenubar=true, %menubar shown
	pdfhighlight=/O, %effect of clicking on a link
	colorlinks=true, %couleurs sur les liens hypertextes
	pdfpagemode=None, %aucun mode de page
	pdfpagelayout=SinglePage, %ouverture en simple page
	pdffitwindow=true, %pages ouvertes entierement dans toute la fenetre
	linkcolor=see, %couleur des liens hypertextes internes
	citecolor=red!75!darkgray, %couleur des liens pour les citations
	urlcolor=see %couleur des liens pour les url
}
\newcommand{\E}{\mathcal{E}}
\newcommand{\fv}{\mathfrak{v}}

\newcommand{\fA}{\mathfrak{A}}

\newcommand{\fc}{\mathfrak{c}}
\newcommand{\F}{\mathfrak{F}}

\newcommand{\fP}{\mathfrak{P}}

\newcommand{\PV}{\mathfrak{PV}}

\newcommand{\Gcal}{\mathcal{G}}  % gauge group

\newcommand{\Ccal}{\mathcal{C}}
\newcommand{\Dcal}{\mathcal{D}}
\newcommand{\Ecal}{\mathcal{E}} 

\newcommand{\Fcal}{\mathcal{F}} 
\newcommand{\Acal}{\mathcal{A}} 

\newcommand{\Ncal}{\mathcal{N}}
\newcommand{\Mcal}{\mathcal{M}}
\newcommand{\Ocal}{\mathcal{O}}
\newcommand{\Scal}{\mathcal{S}}
\newcommand{\Pcal}{\mathcal{P}}

\newcommand{\Tcal}{\mathcal{T}}

\newcommand{\Ci}{C^\infty} % smooth functions
\newcommand{\Hom}{\mathrm{Hom}}

\newcommand{\Loc}{\mathrm{\mathbf{Loc}}}  
 
\newcommand{\bC}{\mathrm{\mathbf{C}}}  
\newcommand{\Sym}{\mathrm{Sym}}     % categorie delle categorie tensoriali
\newcommand{\TVec}{\mathrm{\mathbf{TVec}}}       % the category of locally convex topological vector
\newcommand{\Alg}{\mathrm{\mathbf{Alg}}}    %algebras 
\newcommand{\WF}{\mathrm{WF}}         % wave front set
\newcommand{\id}{\mathrm{id}}               % identity
\newcommand{\dvol}{\mathrm{dvol}_{\sst{M}}} % volume form
\DeclareMathOperator{\im}{\mathrm{Im}}             % image
\newcommand{\loc}{\mathrm{loc}}
\newcommand{\pol}{\mathrm{pol}}

\newcommand{\reg}{\mathrm{reg}}

\newcommand{\inj}{{\rm inj}}
\newcommand{\NN}{\mathbb{N}}          % natural naumbers
\newcommand{\ZZ}{\mathbbmss{Z}}     % Menge der ganzen Zahlen
\newcommand{\RR}{\mathbb{R}}           % real  numbers
\newcommand{\CC}{\mathbb{C}}           % complex numbers

\newcommand{\al}{\alpha}

\newcommand{\la}{\lambda}

\newcommand{\ph}{\phi}

\newcommand{\T}{\cdot_{{}^\Tcal}}

\newcommand{\TT}{\Tcal}

\newcommand{\Poi}[2]{\left\lfloor#1,#2\right\rfloor}
\newcommand{\Riem}{\textrm{Riem}}

\newcommand{\sst}[1]{\scriptscriptstyle{#1}}  % small font for the subscripts
\newcommand{\1}{\mathds{1}}                         % identity
\newcommand{\be}{\begin{equation}}
\newcommand{\ee}{\end{equation}}

\newcommand{\Lap}{\bigtriangleup}

\DeclareMathOperator{\supp}{\mathrm{supp}}      % support
\def\normOrd#1{\mathop{:}\nolimits\!#1\!\mathop{:}\nolimits}
\newcommand{\Pei}[2]{\lfloor #1, #2 \rfloor}

\theoremstyle{plain}
\newtheorem{thm}{Theorem}[section]
\newtheorem{df}[thm]{Definition}
\newtheorem{prop}[thm]{Proposition}
\newtheorem{cor}[thm]{Corollary}
\newtheorem{lemma}[thm]{Lemma}
\newtheorem{conj}{Conjecture}

\theoremstyle{definition}
\newtheorem{rem}[thm]{Remark}

\newtheorem{exa}{Example}

\def\d{{\rm d}}
\def\xto{\xrightarrow}

\def\potimes{\widehat{\otimes}}
\def\Pfrak{{\mathfrak P}}

\def\Alg{\mathbf{Alg}}
\def\CAlg{\mathbf{CAlg}}
\def\PAlg{\mathbf{PAlg}}
\def\Ch{\mathbf{Ch}}
\def\Caus{\mathbf{Caus}}
\def\Open{\mathbf{Open}}
\def\PFA{\mathbf{PFA}}
\def\FA{\mathbf{FA}}
\def\Emb{\mathbf{Emb}}
\def\Riem{\mathbf{Riem}}
\def\Nuc{\mathbf{Nuc}}
\def\Nuch{\mathbf{Nuc}_{\hbar}}
\def\Vec{\mathbf{Vec}}

\def\Oscr{\mathscr{O}}
\def\Orm{\mathrm{O}}
\def\cball{\mathbf{CBall}}

\title{Relating nets and factorization algebras of observables:\\ free field theories}

\author[1]{\small{Owen Gwilliam}}
\address{ University of Massachusetts, Amherst,\\ Department of Mathematics,\\
\normalfont{\texttt{gwilliam@math.umass.edu}}}
\author[2]{\small{Kasia Rejzner}}
\address{University of York,  \\
	Department of Mathematics,\\
\normalfont{\texttt{kasia.rejzner@york.ac.uk}}
}

\date{\today}

\begin{document}
 \sloppy

%\setpagewiselinenumbers
%\linenumbers

\maketitle
\begin{abstract}
In this paper we relate two mathematical frameworks that  make perturbative quantum field theory rigorous: perturbative algebraic quantum field theory (pAQFT) and the factorization algebras framework developed by  Costello and Gwilliam. To make the comparison as explicit as possible, we use the free scalar field as our running example, while giving proofs that apply to any field theory whose equations of motion are Green-hyperbolic (which includes, for instance, free fermions).
The main claim is that for such free theories, there is a natural transformation intertwining the two constructions. In fact, both approaches encode equivalent information if one assumes the time-slice axiom. The key technical ingredient is to use time-ordered products as an intermediate step between a net of associative algebras and a factorization algebra.
\end{abstract}
\tableofcontents

Recently there have appeared two, rather elaborate formalisms for constructing the observables of a quantum field theory via a combination of the Batalin-Vilkovisky framework with renormalization methods. One \cite{FR}, later referred to as FR, works on Lorentzian manifolds and weaves together (a modest modification of) algebraic quantum field theory (AQFT) with the Epstein-Glaser machinery for renormalization. 
The other \cite{CoGw,CG2}, later referred to as CG, works with elliptic complexes (i.e., ``with Euclidean theories'') and constructs factorization algebras using renormalization machinery developed in \cite{Cos}.
To practitioners of either formalism, the parallels are obvious, in motivation and techniques and goals.
It is thus compelling (and hopefully eventually useful!) to provide a systematic comparison of these formalisms, with hopes that a basic dictionary will lead in time to effortless translation.

The primary goal in this paper is to examine in detail the case of free field theories,
where renormalization plays no role and we can focus on comparing the local-to-global descriptions of observables. 
In other words, in the context of this free theory, 
we show how to relate the key structural features of AQFT and factorization algebras.
In the future we hope to compare interacting field theories, 
which demands an examination of renormalization's role
and deepens the comparison by touching on more technical features.

The key to our comparison result is that while the approach of \cite{Cos,CoGw,CG2} constructs the space of quantum observables by deforming the differential on the classical observables, 
one can equivalently leave the differential unchanged and deform the factorization product instead. This deformation of the factorization product corresponds, in the formalism of \cite{FR,FR3}, to the passage from the pointwise product to the time-ordered product by means of the time-ordering operator $\TT$. Hence one can either work with the pointwise product $\cdot$ and the differential $\hat{s}\doteq \TT^{-1}\circ s\circ \TT$, or with the product $\T$ and the differential $s$. In both cases the differential is a derivation with respect to the corresponding product, but only for arguments with disjoint supports, and in fact one obtains a prefactorization structure valued in  Beilinson-Drinfeld algebras.

In the Lorentzian framework of \cite{FR,FR3}, there is, in addition to the commutative time-ordered product, a non-commutative product $\star$, identified as the operator product of quantum observables. The products $\T$ and $\star$ are related by time-ordering and we show in section~\ref{assocstr} how to reconstruct $\star$ from $\T$ for the algebra of free fields.

A secondary goal of this paper is to facilitate communication between communities,
by providing a succinct treatment of this key example in each formalism.
We expect that interesting results---and questions!---can be translated back and forth.

Indeed, one consequence of this effort at comparison is that it spurred a modest enhancement of each formalism.
On the FR side, we introduce a differential graded (dg) version of the usual axioms for the net of algebras.
Prior work fits nicely into this definition, and in the future we hope to examine its utility in gauge theories.
On the CG side, we show that the free field construction applies to Lorentzian manifolds as well as Euclidean manifolds.
(The case of interacting theories in the CG formalism does not port over so simply, as it exploits features of elliptic complexes in its renormalization machinery.)

As an overview of the paper, we begin by raising key questions about how the formalisms agree and differ.
To sharpen these questions, we give precise descriptions of the \emph{outputs} generated by each formalism,
namely the kinds of structure possessed by observables.
On the FR side, one has a net of algebras; on the CG side, a factorization algebra of cochain complexes.
With these definitions in hand, we can state our main results precisely.
As a brief, imprecise gloss, our main result is that the FR and CG constructions agree where they overlap:
if one restricts the CG factorization algebra of observables to the opens on which the FR net is defined (and takes the zeroth cohomology),
then the factorization algebra and net determine the same functor to vector spaces.
We also explain how one can recover as well the algebraic structures on the nets (Poisson for the classical theory, associative for the quantum) from the constructions.
Next, we turn to carefully describing the \emph{constructions} in each formalism, so that we can prove the comparison results.
We recall in detail how each formalism constructs the observables for the free theory given by a Green-hyperbolic operator, producing on the one hand, a net of  algebras on a globally hyperbolic Lorentzian manifold, and on the other, a factorization algebra.
With the constructions in hand, the proof of the comparison results is straightforward.
Finally, we draw some lessons from the comparison and point out natural directions of future inquiry. 

\section{A preview of the key ideas}
\label{preview}

Before delving into the constructions, we discuss field theory from a very high altitude, ignoring all but the broadest features, and explain how each formalism approaches observables.
With this knowledge in hand, it is possible to raise natural questions about how the formalisms differ.
The rest of the paper can be seen as an attempt to answer these questions.

\subsection{Classical theories}

A classical field theory is specified, loosely speaking, by
\begin{enumerate}
\item a smooth manifold $M$ (the ``spacetime''), 
\item a smooth fiber bundle over the manifold $\pi: E \to M$ whose smooth sections $\Gamma(M,E)$ are the ``fields,'' 
\item and a system of partial differential equations on the fields (the ``equations of motion'' or ``Euler-Lagrange equations'') that are variational in nature.
\end{enumerate}

We will discuss issues of functional analysis later, 
but note that we equip the space $\Gamma(M,E)$ of smooth sections with its natural Fr\'echet topology and use the notation $\Ecal(M)$ for it.

In this paper, the focus is on free fields and
we will write the equations as $P(\phi) = 0$ where $\phi$ is a field and $P$ denotes the equations of motion operator.
(There are many variations and refinements on this loose description, of course, but most theories fit into this framework.)

Here the manifold $M$ is equipped with a metric $g$, 
and an important difference is that the FR formalism requires $g$ to have Lorentzian signature
while the CG formalism requires $g$ to be Riemannian. We use the notation $\Mcal\equiv(M,g)$.

In this paper we focus on the Lorentzian case and
we will assume that $P$ is a \textit{Green-hyperbolic} operator, i.e. it has unique retarded and advanced Green functions (see \cite{GreenBear} for a lucid and extensive discussion of this notion). Note that this class of operators allows one to treat the free scalar field and the free Dirac fermion as special cases.

The running example in this paper is the free scalar field, where the fiber bundle is the trivial rank one vector bundle $E = M \times \RR \to M$ 
so that the fields are simply $C^\infty(M,\RR)$, the smooth functions on $M$. The differential equations can be concisely given, since they play such a central role throughout physics and mathematics:
\[
\Box_g \phi + m^2 \phi = 0,
\]
where $\Box_g$ denotes the d'Alembertian (i.e. Laplace-Beltrami operator for a Lorentzian metric) and $m\in\RR_+$ is called the ``mass.''

A crucial feature of field theory is that it is local on the manifold $M$.
Note, to start, that the fields $\Ecal$ form a sheaf that assigns to an open set $U$,
the set
\[
\Ecal(U) = \Gamma(U, \pi_U: \pi^{-1}(U) \to U)
\]
of smooth sections of the bundle over $U$.
That is, $\Ecal$ defines a contravariant functor $\Ecal: \Open(M)^{op} \to Set$ 
from the poset category $\Open(M)$ of open sets in $M$ to the category of sets.
As global smooth sections are patched together from local smooth sections, $\Ecal$ forms a sheaf of sets on $M$.
(It also forms a sheaf of vector spaces and of topological vector space.)

Consider now $Sol(M)$, the set of solutions to the equations of motion, 
i.e., the configurations (or fields) that are allowed by the physical system described by the classical field theory.
(We ignore here, since we're speaking vaguely, whether we should consider solutions that are not smooth, such as distributional solutions and whether we ought to impose boundary conditions.)
Since differential equations are, by definition, local on $M$,
solutions to the equations of motion actually form a sheaf on $M$ .
That is, if we write
\[
Sol(U) = \{ \phi \in \Ecal(U) \, :\, P(\phi) = 0 \}
\]
for sections on $U$ that satisfy the equations of motion,
then $Sol$ also defines a contravariant functor $Sol: \Open(M)^{op} \to Set$.
As global solutions are patched together from local solutions, $Sol$ forms a sheaf of sets on~$M$.

Any measurement of the system should then be some function of $Sol(M)$, the set of global solutions.
In other words, the algebra of functions $\Ocal(Sol(M))$ constitutes an idealized description of all potential measuring devices for the system.
(An important issue later in the text will be what kind of functions we allow, but we postpone that challenge for now, simply remarking that solutions often form a kind of ``manifold,'' possibly singular and infinite-dimensional, so that $\Ocal$ is not merely set-theoretic.)
Even better, we obtain a covariant functor $\Ocal(Sol(-)): \Open(M) \to \CAlg$ to the category $\CAlg$ of \emph{commutative} algebras.
As $Sol$ is a sheaf, $\Ocal(Sol(-))$ should be a cosheaf, meaning that it satisfies a gluing axiom so that the global observables are assembled from the local observables.

Nothing about this general story depends on the signature of the metric,
and each formalism gives a detailed construction of a cosheaf of commutative algebras for a classical field theory (although some technical choices differ, e.g., with respect to functional analysis).
It is with \emph{quantum} field theories that the formalisms diverge.

\subsection{Quantization}

Loosely speaking, the formalisms describe the observables of a quantum field theory as follows.
\begin{itemize}
\item The CG formalism provides a functor $Obs^q: \Open(M) \to \Ch$, which assigns a cochain complex (or differential graded (dg) vector space) of observables to each open set. This cochain complex is a deformation of a commutative dg algebra $Obs^{cl}$, where $H^0(Obs^{cl}(U)) = \Ocal(Sol(U))$.
\item The FR formalism provides a functor $\fA: \Caus(\Mcal) \to \Alg^*$, which assigns a unital $*$-algebra to each ``causally convex'' open set (so that $\Caus(\Mcal)$ is a special subcategory of $\Open(M)$ depending on the global hyperbolic structure of $\Mcal$). The algebra $\fA(U)$ is, in practice, a deformation quantization of the Poisson algebra $\Ocal(Sol(U))$.
\end{itemize}
In brief, both formalisms deform the classical observables, but they deform it in different ways.
In Section \ref{formal definitions} we give precise descriptions of both formalisms.

Two questions jump out:
\begin{enumerate}\label{questions}
\item Why does the FR formalism (and AQFT more generally) restrict to a special class of opens but the CG formalism does not? And what should the FR formalism assign to a general open? 
\item Why does the FR formalism (and AQFT more generally) assign a $*$-algebra but the CG formalism assigns only a vector space? And can the CG approach recover the algebra structure as well?
\end{enumerate}
Both questions admit relatively simple answers, but those answers require discussion of the context (e.g., the differences between elliptic and hyperbolic PDE) and of the BV framework for field theory.
We will organize our treatment of the free scalar field toward addressing these questions.

\section{Nets versus factorization algebras}
\label{formal definitions}

This section sets the table for this paper.
We begin with some background notation (which is mostly self-explanatory, so we suggest the reader only refer to it if puzzled)
before reviewing quickly the key definitions about nets and factorization algebras.
We made an effort to make the definitions accessible to those from the complementary community.

\subsection{Notations}
\label{subsec:notation}

\subsubsection{Geometry}

We fix throughout a smooth vector bundle over the manifold $\pi: E \to M$.
As we work throughout with manifolds equipped with a metric,
we use the associated volume form of $(M,g)$ to identify smooth functions with densities. We also assume that $E$ is equipped with a nondegenerate bilinear pairing on the fibers, so as to identify sections of $E$ with sections of the dual bundle~$E^*$. 

\subsubsection{Functional analysis}\label{FA}

We will follow the conventions that began with Schwartz for various function spaces.
We denote
\begin{itemize}
\item $\Ecal(M) \doteq \Gamma(M,E)$, $\Ecal^*(M) \doteq \Gamma(M,E^*)$ with their natural Fr{\'e}chet topologies,
\item $\Ecal'(M)$ for the strong topological dual (i.e., the space of continuous linear $\RR$-valued functions on a given topological space), which consists of compactly supported distributions,
\item $\Dcal(M)\doteq \Gamma_c(M,E)$, $\Dcal^*(M)\doteq \Gamma_c(M,E^*)$ with their natural inductive limit topologies, and
\item $\Dcal'(M)$ for the strong  topological dual (i.e., the space of continuous linear $\RR$-valued functions on a given topological space), which consists of non-compactly supported distributions.
\end{itemize}
 The reason for using this notation is that it is quite standard in the literature (e.g. \cite{Hoer1}), where one makes a distinction between $\Gamma(M,E)$ (or $\Ci(M,\RR)$ in the special case of $E=M\times \RR$), which is understood just as a vector space and $\Ecal(M)$, which is  $\Gamma(M,E)$ with its natural Fr{\'e}chet topology.  
  
We will also work with certain natural completions of tensor products,
which arise by geometric constructions.

Given vector bundles $E \to M$ and $E' \to M'$, the exterior tensor product $E \boxtimes E'$ denotes the vector bundle on $M \times M'$ arising by the Whitney tensor product of the  bundles $\pi_1^* E \to M \times M'$ and $\pi_2^* E' \to M \times M'$ arising by pull back along the projections $\pi_1: M \times M' \to M$ and $\pi_2: M \times M' \to M'$.
We then introduce the following notations: 
\begin{itemize}
	\item $\Ecal_n(M)\doteq \Gamma(M^n,E^{\boxtimes n})$, which is equal to the completed projective tensor product $\Ecal(M)^{\widehat{\otimes}n}$,
	%\item $\Dcal^{\overline{\otimes}n}\doteq \Gamma(M^n,E^{\boxtimes n})'$,
	\item $\Ecal_n'(M)\doteq \Gamma(M^n,E^{\boxtimes n})'$, with the strong topology.
	\item $\Dcal_{n}(M)\doteq \Gamma_c(M^n,E^{\boxtimes n})$,
	%\item $\Ecal^{\overline{\otimes}n}\doteq (\Gamma_c(M^n,E^{\boxtimes n}))'$.
	\item $\Dcal_{n}'(M) \doteq \Gamma_c(M^n,E^{\boxtimes n})'$, with the strong topology.
\end{itemize}
%A peculiar feature of our notation is that the dual space to $\Dcal^{\widehat{\otimes} n}$ is denoted $\Ecal^{\overline{\otimes} n}$, because we wish to emphasize that
Note that  $\Dcal_n'(M)$ is the distributional completion of $\Ecal^*_n(M)$
and $\Ecal_n'(M)$ is the distributional completion of $\Dcal^*_n(M)$.
In particular,  $\Dcal'(M)$ is the distributional completion of $\Ecal^*(M)$ (smooth functions are densely embeded into the space of distributions)
and $\Ecal'(M)$ is the distributional completion of $\Dcal^*(M)$ (compactly supported functions are densely embedded into the space of compactly supported distributions).

Since we fixed an explicit isomorphism $E \cong E^*$ of vector bundles,
we have preferred inclusions $\Ecal(M) \cong \Ecal_n^*(M)\hookrightarrow \Dcal_n'(M)$ and $\Dcal_n(M)\cong \Dcal_n^*(M) \hookrightarrow \Ecal_n'(M)$.

\begin{rem}
We note that these conventions differ from those in \cite{CoGw},
where 
%$\Ecal(M)$ is used to denote the smooth sections of a vector bundle $F \to M$,
$\Ecal_c(M)$ denotes the compactly supported smooth sections, 
$\overline{\Ecal}(M)$ the distributional sections,
and $\overline{\Ecal}_c(M)$ the compactly supported distributional sections.
%As we stick to scalar fields (i.e., sections of the trivial rank one bundle) throughout this paper,
%we hope there is no confusion.
\end{rem}

We indicate the complexification of a real vector space $V$ by a superscript~$V^{\sst\CC}$. 

\subsubsection{Categories}

Myriad categories will appear throughout this work, 
and so we introduce some of the key ones, 
as well as establish notations for generating new ones.
Categories will be indicated in bold.

We start with a central player.
Let $\Nuc$ denote the category of nuclear, topological locally convex vector spaces,
which is a subcategory of the category of topological locally convex spaces $\TVec$. 
It is equipped with a natural symmetric monoidal structure via the completed projective tensor product $\widehat{\otimes}$
(although we could equally well say `injective' as the spaces are nuclear). 
A reason to work with nuclear spaces is the fact that injective and projective tensor products are isomorphic for such spaces, hence it is sufficient to work with just one monoidal structure. 
Nuclearity is preserved under taking strong duals, a direct sum, an inductive limit of a countable family of nuclear spaces, a product and a projective limit of any family of nuclear spaces. Moreover, spaces of smooth sections and their strong duals introduced in section~\ref{FA} are nuclear.
(See \cite{Tre67} for an accessible treatment of nuclear spaces.)

We emphasize that we make this choice as it suits our purposes.
Many of our ideas and constructions work with other categories, such as $\TVec$,
but require one to be more attentive to which monoidal structures are in play.

\begin{rem}
Given the spaces appearing in our construction, 
it is often worthwhile to work instead with convenient vector spaces \cite{Michor},
but we will not discuss that machinery here,
pointing the interested reader to~\cite{CoGw, Book}.
\end{rem}

If we wish to discuss the category of unital associative algebras of such vector spaces, 
we write $\Alg(\Nuc)$.
Here the morphisms are continuous linear maps that are also algebra morphims.
Similarly, we write $\CAlg(\Nuc)$ for unital commutative algebras in $\Nuc$
and $\PAlg(\Nuc)$ for unital Poisson algebras therein.
We will typically want $*$structures (i.e., an involution compatible with the multiplication),
and we use $\Alg^*(\Nuc)$, $\CAlg^*(\Nuc)$, and $\PAlg^*(\Nuc)$, respectively.

More generally, for $\bC$ a category with symmetric monoidal structure $\otimes$,
we write $\Alg(\bC, {\otimes})$ for the unital algebra objects in that category.
Often we will write simply $\Alg(\bC)$, 
if there is no potential confusion about which symmetric monoidal structure we mean. 

It is often useful to forget extra structure. 
We use  $\fv: \PAlg^*(\Nuc) \to \Nuc$ and $\fv: \Alg^*(\Nuc) \to \Nuc$ to denote forgetful functors to vector spaces.
We use $\fc: \PAlg^*(\Nuc) \to \CAlg^*(\Nuc)$ to denote the forgetful functor to commutative algebras. 

In a similar manner, if $\bC$ is an additive category,
we write $\Ch(\bC)$ to denote the category of cochain complexes and cochain maps in $\bC$.
Thus $\Ch(\Nuc)$ denotes the category of cochain complexes in $\Nuc$
(which, unfortunately, is not a particular nice place to do homological algebra).
We note that we allow unbounded complexes, 
but in practice our constructions here produce complexes bounded on one side.
(If we treated gauge theories, we would have complexes unbounded in both directions.)

This category admits a symmetric monoidal structure by the usual formula:
the degree $k$ component of the tensor product of two cochain complexes is
\[
(A^\bullet \otimes B^\bullet)^k = \bigoplus_{i+j=k} A^i \,\widehat{\otimes}\, B^j.
\]
Hence we write $\Alg(\Ch(\Nuc))$ for the category of algebra objects, 
also known as dg algebras.

\begin{rem}
This category $\Ch(\Nuc)$ admits a natural notion of weak equivalence:
a cochain map is a weak equivalence if it induces an isomorphism on cohomology.
Thus it is a relative category and presents an $(\infty,1)$-category,
although we will not need such notions here.
\end{rem}

There is another important variant to bear in mind.
In the original axiomatic framework of Haag and Kastler, the notion of subsystems is encoded in the injectivity requirement for algebra morphisms. 
We use the superscript ``$\inj$'', if we want to impose this condition on morphisms, for a given category. 
Hence $\Alg^*(\Nuc)^\inj$ consists of the category whose objects are nuclear, topological locally convex unital $*$-algebras but whose morphisms are \emph{injective} continuous algebra morphisms.

\subsubsection{Dealing with $\hbar$}

In perturbative field theory,
one works with $\hbar$ as a formal variable.
In our situation, since we restrict to free fields,
this is overkill: 
one can actually set $\hbar = 1$ throughout, and all the constructions are well-defined.
But $\hbar$ serves as a helpful mnemonic for what we are deforming 
and as preparation for the interacting case.

We thus introduce categories involving $\hbar$ that emphasize its algebraic role and minimize any topological issues.
As a gesture at the topological issues, note that the ring $\CC[[\hbar]]$ is equipped with an adic topology, 
and so one might want to work with topological vector spaces that are modules over $\CC[[\hbar]]$ in a continuous way.
There are then some nontrivial compatibilities to discuss.
Instead, we will restrict our attention to a special class of objects where we can avoid such discussions,
as follows.

Let $\Nuch$ denote the following category.
The objects are the same as those of $\Nuc$,
but given $V \in \Nuc$ we use $V[[\hbar]]$ to denote the corresponding object in $\Nuch$.
The reader should \emph{think} of this space as $\prod_{n \geq 0} \hbar^n V$ 
so that a vector $v$ would be a formal power series
\[
v = v_0 + \hbar v_1 + \cdots + \hbar^n v_n + \cdots
\]
with coefficients in~$V$.
We want a morphism to encode an $\hbar$-linear map of such modules, 
so it should be determined by where the $\hbar^0 V$ component of $V[[\hbar]]$ would go.
Hence, we define the space of morphisms to be
\[
\Hom_{\Nuch}(V[[\hbar]],W[[\hbar]]) = \prod_{n \geq 0}\hbar^n\Hom_\Nuc(V,W),
\]
where the $\hbar^n$ is just a formal bookkeeping device.
Composition is by precisely the rule one would use for $\hbar$-linear maps.
For instance, given $f = (\hbar^n f_n) \in \Hom_{\Nuch}(V[[\hbar]],W[[\hbar]])$ and $g = (\hbar^n g_n) \in \Hom_{\Nuch}(W[[\hbar]],X[[\hbar]])$, 
the composite $g \circ f$ has 
\begin{align*}
(g \circ f)_0 &= g_0 \circ f_0, \\
(g \circ f)_1 &= g_1 \circ f_0 + g_0 \circ f_1, \\
&\vdots
\end{align*}
since informally we want
\[
(\sum_{n \geq 0} \hbar^n g_n) \circ (\sum_{m \geq 0} \hbar^m f_m) = \sum_{p \geq0} \sum_{m +n = p} \hbar^p g_n \circ f_m.
\]
We equip this category with a symmetric monoidal structure borrowed from~$\Nuc$:
\[
V[[\hbar]] \potimes_\hbar W[[\hbar]] = (V \potimes W)[[\hbar]].
\]
Note that it agrees with the completed tensor product over $\CC[[\hbar]]$,
in the same sense that composition of morphisms does.

The category $\Alg(\Nuch)$ then consists of algebra objects in that symmetric monoidal category,
$\Ch(\Nuch)$ denotes cochain complexes therein, 
and $\Alg(\Ch(\Nuch))$ denotes dg algebras therein.
We use again the notations $\fv$ and $\fc$ as forgetful functors, hopefully without producing confusion.

\subsection{Overview of the pAQFT setting}

The framework of AQFT formalizes rigorously the core ideas of Lorentzian field theory, 
building on the lessons of rigorous quantum mechanics,
but the standard calculational toolkit for interacting QFT does not fit into the framework.
Perturbative AQFT is a natural modification of the framework within which one often can realize a version of the usual calculations,
while preserving the structural insights of AQFT.

\subsubsection{}

Let $\Mcal=(M,g)$ be an $n$-dimensional spacetime, i.e., a smooth $n$-dimensional manifold with the metric $g$ of signature $(+,-,\dots,-)$. We assume $\Mcal$ to be oriented, time-oriented and globally hyperbolic (i.e. it admits foliation with Cauchy hypersurfaces). To make this concept clear, let us recall a few important definitions in Lorentzian geometry.

\begin{df}
Let $\gamma: \RR\supset I\rightarrow M$ be a smooth curve in $M$, for $I$ an interval in $\RR$ and let  $\dot{\gamma}$ be the vector tangent to the curve. We say that $\gamma$ is
\begin{itemize}
	\item  \textbf{timelike}, if  $g(\dot{\gamma},\dot{\gamma})> 0$,
	\item  \textbf{spacelike}, if  $g(\dot{\gamma},\dot{\gamma})< 0$,
	\item  \textbf{lightlike} (or null), if  $g(\dot{\gamma},\dot{\gamma})= 0$,
	\item  \textbf{causal}, if  $g(\dot{\gamma},\dot{\gamma})\geq 0$.	
\end{itemize}
\end{df}

The classification of curves defined above is the \textit{causal structure} of~$\Mcal$.

\begin{df}
A set $\Ocal\subset\Mcal$ is \textbf{causally convex} 
if for any causal curve $\gamma:[a,b]\rightarrow \Mcal$ whose endpoints $\gamma(a),\gamma(b)$ lie in $\Ocal$, 
then every interior point $\gamma(t)$,  for $t\in [a,b]$, also lies in $\Ocal$ for every $t\in [a,b]$.
\end{df}

With these definitions in hand, we can define the category of open subsets on which we specify algebras of observables.

\begin{df}\label{causal}
Let $\Caus(\Mcal)$ be the collection of relatively compact, connected, contractible, causally convex subsets $\Ocal\subset \Mcal$. Note that the inclusion relation $\subset$ is a partial order on $\Caus(\Mcal)$, so $(\Caus(\Mcal),\subset~)$ is a poset (and hence a category).
\end{df}

\subsubsection{}

To formulate a classical theory, we start with making precise what we mean by the model for the space of classical fields.

\begin{df}\label{ClassFT}
	A \textbf{classical field theory model} on a spacetime $\Mcal$ is a functor $\fP : \Caus(\Mcal) \to \PAlg^*(\Nuc)^\inj$ that obeys \textbf{Einstein causality}, i.e.:
for $\Ocal_1,\Ocal_2\in \Caus(\Mcal)$ that are spacelike to each other, we have
		\[
		\Poi{\mathfrak{P}(\Ocal_1)}{\mathfrak{P}(\Ocal_2)}_{\Ocal}=\{0\}\,,
		\]
		where $\Poi{.}{.}_{\Ocal}$ is the Poisson bracket in  any $\fP(\Ocal)$ for an $\Ocal$ that contains both $\Ocal_1$ and $\Ocal_2$.
\end{df}

Note how this definition formalizes the sketch of classical field theory in Section \ref{preview}:
we have a category of open sets --- here, $\Caus(\Mcal)$ --- and a functor to a category of Poisson algebras,
since the observables of a classical system should form such a Poisson algebra. 

In the AQFT community, the underlying commutative algebra of $\Pfrak$ is known as the space of classical fields. In more formal language, we introduce a forgetful functor $\fc: \PAlg^*(\Nuc) \to \CAlg(\Nuc)$ and state the following. 

\begin{df}
The \textbf{space of classical fields} is the functor $\fc\circ \Pfrak$.
\end{df}

\begin{rem}
Note that there is a conflict here with terminology in the CG framework (and with some other communities working in physics),
where a field is an element of $\Ecal$, i.e., a configuration in the AQFT sense.
Thus the space of fields in the CG sense corresponds to the configuration space in the AQFT sense.
In the CG framework, $\fc \circ \Pfrak$ is the commutative algebra of observables on the classical fields, {\em aka} functions on the configuration space.
\end{rem}

It is useful to introduce a further axiom that articulates more precisely how the dynamics of a classical theory should behave.
Here, a time orientation plays an important role.

\begin{df}
Given the global timelike vector field $u$ (the time orientation) on $M$, 
a causal curve $\gamma$ is called \textbf{future-directed} if $g(u,\dot{\gamma}) > 0$ all along $\gamma$. 
It is \textbf{past-directed} if $g(u,\dot{\gamma}) < 0$.
\end{df}

\begin{df}
A causal curve $gamma:(a,b) \rightarrow M$ is \textbf{future inextendible} if $\lim_{t \rightarrow b} \gamma(t)$ does not exist in $M$.
\end{df}

\begin{df}
A \textbf{Cauchy hypersurface} in $\Mcal$ is a smooth subspace of $\Mcal$ such that every inextendible causal curve intersects it exactly once.
\end{df}

\begin{rem}
The significance of Cauchy hypersurfaces lies in the fact that one can use them to formulate the initial value problem for partial differential equations, and for normally hyperbolic equations this problem has a unique solution.
\end{rem}

With this notion in hand, we have a language for enforcing equations of motion at an algebraic level.

\begin{df}
A model is said to be \textbf{on-shell} if in addition it satisfies the \textbf{time-slice axiom}:
for any $\Ncal\in\Caus(\Mcal)$ a neighborhood of a Cauchy surface in the region $\Ocal\in\Caus(\Mcal)$, 
then $\mathfrak{P}$ sends the inclusion $\Ncal \subset \Ocal$ to an isomorphism $\mathfrak{P}(\Ncal)\cong \mathfrak{P}(\Ocal)$. Otherwise the model is called \textbf{off-shell}.
\end{df}

\begin{rem}
Note that being on-shell codifies the idea that the set of solutions is specified by the initial value problem on a Cauchy hypersurface.
\end{rem}

\subsubsection{}

We now turn to the quantum setting.

\begin{df}\label{AQFT}
A \textbf{QFT model} on a spacetime $\Mcal$ is a functor $\mathfrak{A} : \Caus(\Mcal) \to \Alg^*(\Nuc_\hbar)^{\inj}$ that satisfies \textbf{Einstein causality} (Spacelike-separated observables commute). That is, for $\Ocal_1,\Ocal_2\in \Caus(\Mcal)$ that are spacelike to each other, we have
		\[
		[\fA(\Ocal_1),\fA(\Ocal_2)]_\Ocal=\{0\}\,,
		\]
		where $[.,.]_{\Ocal}$ is the commutator in  any $\fA(\Ocal)$ for an $\Ocal$ that contains both $\Ocal_1$ and $\Ocal_2$. 
\end{df}

\begin{df}\label{timeslice}
	A QFT model is said to be \textbf{on-shell} if in addition it satisfies the \textbf{time-slice axiom} (where one simply replaces $\mathfrak{P}$ by $\fA$ in the definition above).
	Otherwise, it is \textbf{off-shell}. 
\end{df}

Often a quantum model arises from a classical one by means of quantization. In order to formalize this, we need some notation. Given a functor $\F$, let $\F[[\hbar]]$ denote the functor sending $\Ocal$ to $\F(\Ocal) \widehat{\otimes} \CC[[\hbar]]$.

\begin{df}\label{df:quantization}
A quantum model $\fA$ is said to be a \textbf{quantization} of a classical model $\Pfrak$, if:
\begin{enumerate}
	\item $\fv\circ \fA \cong \fv \circ \Pfrak[[\hbar]]$,
	\item $\fc\circ \Pfrak \cong \fA/(\hbar)$, and
	\item the brackets $-\frac{i}{\hbar}[.,.]_\Ocal$ coincides with $\{.,.\}_\Ocal$ modulo $\hbar$,
\end{enumerate}
where the isomorphism (2) is induced by the isomorphism~(1).
\end{df}

Later, it will be important to have a generalization of definitions that assigns a \emph{dg algebra} to each $\Ocal \in \Caus(\Mcal)$.
Recall that a dg algebra is a $\ZZ$-graded vector space $A = \bigoplus_n A^n$ equipped with 
\begin{itemize}
\item a grading-preserving associative multiplication $\star$ so that $a \star b \in A^{m+n}$ if $a \in A^m$ and $b \in A^n$, an
\item a differential $\d: A \to A$ that increases degree by one, satisfies $\d^2 = 0$, and is a derivation, so that
\[
\d(a \star b) = \d a \star b + (-1)^{|a|} a \star \d b 
\]
for homogeneous elements $a, b \in A$.
\end{itemize}
This generalization appears naturally when one adopts the BV framework for field theory,
as it uses homological algebra in a serious way.
We introduce these dg models in the next section.

\subsection{A dg version of pAQFT}
\label{sec:dgpaqft}

We articulate here a very minimal generalization of the usual AQFT axioms that allows dg algebras, rather than plain algebras, as the target category.
It will be apparent that free field theories fits these axioms, 
and we intend to show that the perturbative construction of gauge theories does as well.
We forewarn the reader that we do not impose certain conditions (notably isotony) because we do not yet know an appropriate dg generalization.

\begin{rem}
	Others have suggested modifications of AQFT in a dg direction, particularly \cite{BDHS13,BSS14,BS17}, who explore the case of abelian gauge theories in depth and even examine some nonperturbative facets. A generalization to non-abelian gauge theories has been obtained on the classical level in~\cite{BSS17}.
	We expect, based on explicit models constructed in \cite{FR3}, that our minimal, perturbative definitions apply verbatim to gauge theories like Yang-Mills theories and can be seen as the infinitesimal version of the axioms of homotopy AQFT proposed by Benini and Schenkel~\cite{BS17}.
\end{rem}

\subsubsection{}

Before we get to our definition, let us sketch the big picture.
The basic principle is to replace ordinary categories and functors by higher categorical analogues. 
Hence we want to articulate a version of a QFT model as a functor between $\infty$-categories, 
and so the first step is to determine the source and target $\infty$-categories.

Our source category $\Caus(\Mcal)$ does not require modification in any nontrivial way,
so we simply view it as an $\infty$-category.\footnote{If one wants to fix a particular model for $\infty$-categories, such as quasicategories,
	then there is always a standard way to promote an ordinary category to such a higher category.
	For instance, one can take the nerve ${\rm N}(\Caus(\Mcal))$ to obtain a quasicategory.}
On the other hand, the target category $\Alg^*(\Nuch)$ admits several analogues among $\infty$-categories.
Here are two:
\begin{itemize}
	\item Consider the category $\Ch(\Nuch)$ as a {\em relative} category with quasi-isomorphism as the notion of weak equivalence;
	this determines an $\infty$-category we momentarily denote ${\mathcal D}(\Nuch)$. 
	Then take the $\infty$-category of $\ast$-algebras ${\mathcal A}{\rm lg}^*({\mathcal D}(\Nuch))$.
	\item Consider the category $\Alg^*(\Ch(\Nuch))$ as a relative category with quasi-isomorphism as the notion of weak equivalence.
	This determines another $\infty$-category.
\end{itemize}
Note an important difference between these approaches:
whether we take algebras at a 1-categorical level or $\infty$-categorical level.
It is not manifest these constructions agree, nor are these the only ways to form a higher category of homotopy-coherent $\ast$-algebras in some class of topological vector spaces.\footnote{We are not even discussing here whether it would be better to work with some other class of functional-analytic spaces. It should be clear that one might reasonably replace $\Nuc$ by $\textbf{TVS}$, or some other category.}

Once one has fixed a target $\infty$-category $\Ccal$, 
then one can view a functor of $\infty$-categories $\fA: {\rm N}(\Caus(\Mcal)) \to \Ccal$ 
as a higher version of the data of a QFT model:
in a homotopy-coherent fashion, it assigns a $*$-algebra to each $\Ocal$ in~$\Caus(\Mcal)$.

\begin{rem}
	It is an interesting question --- particularly from the perspective of examples and applications --- to determine when such a functor of $\infty$-categories $\fA$ can be represented by a strict functor between explicit (relative) categories, such as $\widetilde{\fA}: \Caus(\Mcal) \to \Alg^*(\Ch(\Nuch))$.
	We do not address that question here, although we expect it admits a clean answer.
	We feel, however, that it is a question distinct from the issue of formulating a good, abstract definition.
\end{rem}

To obtain a higher version of a QFT model, however, 
we need to articulate versions of Einstein causality and the time-slice axiom.
Again there are several approaches.
As yet we do not feel it is clear which approach is most natural or compelling,
so we encourage interested readers to explore and advocate the approach that appeals to them.

Thankfully, as we will see, the constructions from the FR and CG formalisms yield ordinary functors that ought to determine functors of higher categories in almost any imaginable approach,
as will be manifest to those familiar with higher categories.
To flag the provisional nature of the definitions we provide below,
we include the adjective ``semistrict,'' since we mix lower and higher categorical approaches.

\subsubsection{}

Recall that $\Ch(\Nuc)$ denotes the category whose objects are cochain complexes in $\Nuc$  and whose morphisms are continuous cochain maps.
We equip it with the completed projective tensor product $\potimes$ to make it symmetric monoidal.
So far we have only specified an ordinary category, 
but we can view it as presenting an $\infty$-category by making it a relative category:
a morphism is a weak equivalence if it is a quasi-isomorphism.

\begin{df}\label{dgClassFT}
	A \textbf{semistrict dg classical field theory model} on a spacetime $\Mcal$ is a functor $\mathfrak{P} : \Caus(\Mcal) \to \PAlg^*(\Ch(\Nuc))$,  so that each $\mathfrak{P}(\Ocal)$ is a locally convex dg Poisson $*$-algebra satisfying \textbf{Einstein causality}: spacelike-separated observables Poisson-commute at the level of cohomology. That is, for $\Ocal_1,\Ocal_2\in \Caus(\Mcal)$ that are spacelike to each other, the bracket $\Poi{\mathfrak{P}(\Ocal_1)}{\mathfrak{P}(\Ocal_2)}$ is exact (and so vanishes at the level of cohomology)
	in $\Pfrak(\Ocal')$ for any $\Ocal' \in \Caus(\Mcal)$ that contains both $\Ocal_1$ and~$\Ocal_2$.
	
	It satisfies the \textbf{time-slice axiom} if for any $\Ncal\in\Caus(\Mcal)$ a neighborhood of a Cauchy surface in the region $\Ocal\in\Caus(\Mcal)$, 
	then the map $\mathfrak{P}(\Ncal) \to \mathfrak{P}(\Ocal)$ is a quasi-isomorphism.
\end{df}

Note that one can post-compose such a functor with the functor of cohomology.
One then obtains, for instance, a functor
\[
H^0 \Pfrak: \Caus(\Mcal) \to \PAlg^*(\Nuc).
\]
It is \emph{almost} a classical field theory model, as before.
By construction it satisfies Einstein causality, but it need not satisfy isotony.
Hence our definition imposes the usual axioms (excluding isotony) only at the level of cohomology.
This change is natural inasmuch as we view quasi-isomorphic cochain complexes as equivalent,
and so we should only impose conditions that are invariant under quasi-isomorphism.

\begin{rem}
	Isotony holds at the cochain level for the constructions and example with which we are familiar,
	but it may fail at the level of cohomology,
	as it does in the setting of gauge theory.
	(Consider, as a toy model, how ordinary cohomology can be viewed as arising from sheaf cohomology of a locally constant sheaf.
	Locally, the sheaf is simple but its cohomological behavior depends on the topology of each open.)
	One might guess that isotony holds at the level of cohomology for inclusions $\Ocal \to \Ocal'$ between contractible opens, but we hesitate to impose that condition until we have explored more examples. 
\end{rem}

One can further loosen the definition, if one wishes, by asking for associativity of morphisms only up to homotopy coherence.
This is a formal change to implement and not relevant to our focus in this paper.
We will introduce, however, the appropriate notion of weak equivalence of models,
so that we have a relative category implicitly presenting an $\infty$-category.

\begin{df}
	A natural transformation $\eta: \Pfrak \Rightarrow \Pfrak'$ between two semistrict dg classical field theory models is a \textbf{weak equivalence} if the map $\eta_{\Ocal}: \Pfrak(\Ocal) \to \Pfrak'(\Ocal)$ is a quasi-isomorphism for every $\Ocal \in \Caus(\Mcal)$.
\end{df}

We now turn to the quantum setting.

\begin{df}\label{LCQFT}
	A \textbf{semistrict dg QFT model} on a spacetime $\Mcal$ is a functor 
	$\fA : \Caus(\Mcal) \to \Alg^*(\Ch(\Nuch))$,  so that each $\fA(\Ocal)$ is a locally convex unital $*$-dg algebra satisfying \textbf{Einstein causality}: spacelike-separated observables commute at the level of cohomology. 
	That is, for $\Ocal_1,\Ocal_2\in \Caus(\Mcal)$ that are spacelike to each other, the bracket
	$[\fA(\Ocal_1),\fA(\Ocal_2)]$ is exact in $\fA(\Ocal')$ for any $\Ocal' \in \Caus(\Mcal)$ that contains both $\Ocal_1$ and~$\Ocal_2$.
	
	It satisfies the \textbf{time-slice axiom} if for any $\Ncal\in\Caus(\Mcal)$ a neighborhood of a Cauchy surface in the region $\Ocal\in\Caus(\Mcal)$, 
	then the map $\fA(\Ncal) \to \fA(\Ocal)$ is a quasi-isomorphism.
\end{df}

Again, we introduce a notion of weak equivalence.

\begin{df}
	A natural transformation $\eta: \fA \Rightarrow \fA'$ between two semistrict dg classical field theory models is a \textbf{weak equivalence} if the map $\eta_{\Ocal}: \fA(\Ocal) \to \fA'(\Ocal)$ is a quasi-isomorphism for every $\Ocal \in \Caus(\Mcal)$.
\end{df}

\subsection{Overview of factorization algebras}

In their work on chiral conformal field theory,
Beilinson and Drinfeld introduced factorization algebras in an algebro-geometric setting.
These definitions also encompass important objects in geometric representation theory, 
playing a key role in the geometric Langlands program.
Subsequently, Francis, Gaitsgory, and Lurie identified natural analogous definitions in the setting of  manifolds,
which provide novel approaches in, e.g., homotopical algebra and configuration spaces.
Below we describe a version of factorization algebras, developed in \cite{CoGw}, that is well-suited to field theory.

As this brief history indicates, factorization algebras do \emph{not} attempt to axiomatize the observables of a field theory.
Instead, they include examples from outside physics, such as from topology and representation theory, and permit the transport of intuitions and ideas among these fields.
We will explain below further structure on a factorization algebra that makes it behave like the observables of a field theory in the Batalin-Vilkovisky formalism.

\subsubsection{The core definitions}

Let $M$ be a smooth manifold. 
Let $\Open(M)$ denote the poset category whose objects are opens in $M$ and where a morphism is an inclusion.
A factorization algebra will be a functor from $\Open(M)$ to a symmetric monoidal category $\bC$ with tensor product $\otimes$ equipped with further data and satisfying further conditions.
We will explain this extra information in stages.
(Note that almost all the definitions below apply to an arbitrary topological space, or even site with an initial object, and not just smooth manifolds.)

\begin{df}
A \textbf{prefactorization algebra} $\Fcal$ on $M$ with values in a symmetric monoidal category $(\bC, \otimes)$ consists of the following data:
\begin{itemize}
\item for each open $U \subset M$, an object $\Fcal(U) \in \bC $,
\item for each finite collection of pairwise disjoint opens $U_1,\ldots,U_n$, with $n > 0$, and an open $V$ containing every $U_i$, a morphism
\[
\Fcal(\{U_i\}; V): \Fcal(U_1) \otimes \cdots \otimes \Fcal(U_n) \to \Fcal(V),
\]
\end{itemize}
and satisfying the following conditions:
\begin{itemize}
\item composition is associative, so that the triangle
\[
\begin{tikzcd}
\bigotimes_i \bigotimes_j \Fcal(T_{ij}) \arrow{rr} \arrow{rd} &&\bigotimes_i \Fcal(U_{i}) \arrow{ld} \\
&\Fcal(V)&
\end{tikzcd}
\]
commutes for any collection $\{U_i\}$, as above, contained in $V$ and for any collections $\{T_{ij}\}_j$ where for each $i$, the opens $\{T_{ij}\}_j$ are pairwise disjoint and each contained in $U_i$,
\item the morphisms $\Fcal(\{U_i\}; V)$ are equivariant under permutation of labels, so that the triangle
\[
\begin{tikzcd}
\Fcal(U_{1}) \otimes \cdots \otimes \Fcal(U_n) \arrow{rr}{\simeq} \arrow{rd} && \Fcal(U_{\sigma(1)}) \otimes \cdots \otimes \Fcal(U_{\sigma(n)}) \arrow{ld}\\
&\Fcal(V)&
\end{tikzcd}
\]
commutes for any $\sigma \in S_n$.
\end{itemize}
\end{df}

Note that if one restricts to collections that are singletons (i.e., some $U \subset V$), 
then one obtains simply a precosheaf $\Fcal: \Open(M) \to \bC $.
By working with collections, we are specifying a way to ``multiply'' elements living on disjoint opens to obtain an element on a bigger open.
In other words, the topology of $M$ determines the algebraic structure.
(One can use the language of colored operads to formalize this interpretation, but we refer the reader to \cite{CoGw} for a discussion of that perspective.
Moreover, one can loosen the conditions to be homotopy-coherent rather than on-the-nose.)

A factorization algebra is a prefactorization algebra for which the value on bigger opens is determined by the values on smaller opens,
just as a sheaf is a presheaf that is local-to-global in nature.
A key difference here is that we need to be able to reconstruct the ``multiplication maps'' from the local data, and so we need to modify our notion of cover accordingly.

\begin{df}
A \textbf{Weiss cover} $\{U_i\}_{\{i \in I\}}$ of an open subset $U \subset M$ is a collection of opens $U_i \subset U$ such that for any finite set of points $S= \{x_1,\ldots,x_n\} \subset U$, there is some $i \in I$ such that $S \subset U_i$.
\end{df}

\begin{rem}
Note that a Weiss cover is also a cover, simply by considering singletons.
Typically, however, an ordinary cover is not a Weiss cover. 
Consider, for instance, the case where $U = V \sqcup V'$, with $V,V'$ disjoint opens.
Then $\{V,V'\}$ is an ordinary cover by not a Weiss cover, since neither $V$ nor $V'$ contains any two element set $\{x,x'\}$ with $x \in V$ and $x' \in V'$.
Nonetheless, Weiss covers are easy to construct.
For instance, a Weiss cover of an $n$-manifold $M$ is given by the collection of open subsets that are each homeomorphic to a finite union of copies of~$\RR^n$.
\end{rem}

This notion of cover determines a Grothendieck topology on $M$;
concretely, this means it determines a notion of cover for each open of $M$ that behaves nicely with respect to intersection of opens and refinements of covers.
In particular, we can talk about (co)sheaves relative to this Weiss topology on~$M$.

\begin{df}
A \textbf{factorization algebra} $\Fcal$ is a prefactorization algebra on $M$ such that the underlying precosheaf is a cosheaf with respect to the Weiss topology.
That is, for any open $U$ and any Weiss cover $\{U_i\}_{i \in I}$ of $U$, the diagram
\[
\begin{tikzcd}
\coprod_{i,j} \Fcal(U_i \cap U_j) \arrow[r, shift left] \arrow[r, shift right] & \coprod_i \Fcal(U_i) \arrow[r] & \Fcal(U)
\end{tikzcd}
\]
is a coequalizer. 
\end{df}

Typically, our target category $(\bC, \otimes)$ is vector spaces of some kind (such as topological vector spaces),
in which case the coproducts $\coprod$ denote direct sums $\oplus$ and the coequalizer simply means that $\Fcal(U)$ is the cokernel of the difference of the maps for the inclusions $U_i \cap U_j \subset U_i$ and~$U_i \cap U_j\subset U_j$.
Note that we have implicitly assumed that $\bC $ possesses enough colimits,
and we will assume that henceforward. 

\begin{rem}
The prefactorization algebras we construct in this paper use spaces of smooth or distributional sections,
and hence live in nuclear spaces.
In Chapter 6, Section 5 of \cite{CoGw}, it is checked directly that the relevant colimits exist for these functors in the closely related category of differentiable vector spaces.
In short, it is proved there that our main constructions form factorization algebras.
(The arguments mimic the proofs that smooth functions form a sheaf --- partitions of unity play a role --- but exploit the Weiss condition at one key point.)
We do not examine here the colimit condition in nuclear spaces.
\end{rem}

\begin{rem}
In fact, our target category is usually cochain complexes of vector spaces, 
and we want to view cochain complexes as (weakly) equivalent if they are quasi-isomorphic.
Hence, we want to work in an $\infty$-categorical setting.
In such a setting, the cosheaf condition becomes higher categorical too: 
we replace the diagram above by a full simplicial diagram over the \v{C}ech nerve of the cover and we require $\Fcal(U)$ to be the homotopy colimit over this simplicial diagram.
For exposition of these issues, see~\cite{CoGw}.
\end{rem}

In practice, another condition often holds, and it's certainly natural from the perspective of field theory.

\begin{df}
A factorization algebra $\Fcal$ is \textbf{multiplicative} if the map
\[
\Fcal(V) \otimes \Fcal(V') \to \Fcal(V \sqcup V')
\]
is an isomorphism for every pair of disjoint opens~$V,V'$.
\end{df}

In brief, if $\Fcal$ is a multiplicative factorization algebra, one can reconstruct $\Fcal$ if one knows how it behaves on a collection of small opens.
For instance, suppose $M$ is a Riemannian manifold and one knows $\Fcal$ on all balls of radius $\leq 1$,
then one can reconstruct $\Fcal$ on every open of $M$.
(See Chapter 7 of \cite{CoGw} for how to reconstruct from a Weiss basis.)
Our examples are often multiplicative, or at least satisfy the weaker condition that the map is a dense inclusion.

Note that there is a category of prefactorization algebras $\PFA(M,(\bC,\otimes))$ 
where each object is a prefactorization algebra on $M$ and 
where a morphism $\eta: \Fcal \to \Gcal$ consists of a collection of morphisms in~$\bC $,
\[
\{\eta(U): \Fcal(U) \to \Gcal(U) \}_{U \in \Open(M)}
\] 
such that all the multiplication maps intertwine.
The factorization algebras form a full subcategory $\FA(M,(\bC,\otimes))$ of~$\PFA(M,(\bC,\otimes))$.

\begin{rem}
It is natural to wonder if there is a functor adjoint to the forgetful ({\em aka} inclusion) functor of factorization algebras into prefactorization algebras,
by analogy to the sheafification functor from presheaves to sheaves.
We do not know the answer to this question.
There exists a cosheafification functor from precosheaves to Weiss cosheaves,
but the underlying precosheaf of a prefactorization algebra does not know about structure maps involving multiple disjoint opens, 
so it seems unlikely that Weiss cosheafification is sufficient (by itself) to produce a factorization algebra.
\end{rem}

\subsubsection{Relationship with field theory}

By now, the reader may have noticed that there has been no discussion of fields or Poisson algebras or so on.
Indeed, the definitions here are more general and less involved than for the AQFT setting because they aim to apply outside the context of field theory (e.g., there are interesting examples of factorization algebras arising from geometric representation theory and algebraic topology) and because there is no causality structure to track.
By contrast, AQFT aims to formalize precisely the structure possessed by observables of a field theory on Lorentzian manifolds, and hence must take into account both causality and other characterizing features of field theories (e.g., Poisson structures at the classical level).

Let us briefly indicate how to articulate observables of field theory in this setting,
suppressing important issues of homological algebra and functional analysis,
which are discussed below in the context of the free scalar field and in \cite{CoGw,CG2} in a broader context.
The necessary extra ingredient is that on each open $U$, the object $\Fcal(U)$ has an algebraic structure.

\begin{df}
Given prefactorization algebras $\Fcal, \Gcal$ on $M$, let $\Fcal \otimes \Gcal$ denote the prefactorization algebra with 
\[
(\Fcal \otimes \Gcal)(U) = \Fcal(U) \otimes \Gcal(U)
\]
and the obvious tensor product of structure maps.
\end{df}

In other words, the category of prefactorization algebras $\PFA(M,(\bC,\otimes))$ is itself symmetric monoidal.
In many cases the full subcategory $\FA(M, (\bC,\otimes))$ is closed under this symmetric monoidal product.
In particular, if the tensor product $\otimes$ in $\bC $ preserves colimits separately in each variable (or at least geometric realizations), then $\Fcal \otimes \Gcal$ is a factorization algebra when $\Fcal, \Gcal$ are.

Thus, if $\bC$ is some category of vector spaces, 
one can talk about, e.g., a commutative algebra in $\PFA(M,(\bC,\otimes))$.
That means $\Fcal$ is equipped with a map of prefactorization algebras $\cdot: \Fcal \otimes \Fcal \to \Fcal$ satisfying all the conditions of a commutative algebra.
Similarly, one can talk about Poisson or $*$-algebras.

It is equivalent to say that $\Fcal$ is  in $\CAlg(\PFA(M,(\bC,\otimes)))$ 
or to say it is a prefactorization algebra with values in $\CAlg(\bC,\otimes)$, 
the category of commutative algebras in $(\bC,\otimes)$.
This equivalence does not apply, however, to factorization algebras, due to the local-to-global condition:
a colimit of commutative algebras does not typically agree with the underlying colimit of vector spaces.
For instance, in the category of ordinary commutative algebras $\CAlg(\Vec,\otimes)$, 
the coproduct  is $A \otimes B$, but in the category of vector spaces $\Vec$, it is the direct sum~$A \oplus B$.
(This issue is very general: for an operad $\Ocal$, the category $\Ocal\text{-alg}(\bC,\otimes)$ of $\Ocal$-algebras has a forgetful functor to $\bC $ that always preserves limits but rarely colimits.)
Thus, a commutative algebra in factorization algebras $\Fcal \in \CAlg(\PFA(M,(\bC,\otimes)))$ 
assigns a commutative algebra to every open $U$ and a commutative algebra map to every inclusion of disjoint opens $U_1,\ldots, U_n \subset V$, 
but it satisfies the coequalizer condition in $\bC $, not in~$\CAlg(\bC,\otimes)$.

This terminology lets us swiftly articulate a deformation-theoretic view of the Batalin-Vilkovisky framework.

\begin{df}
A \textbf{classical field theory model} is a 1-shifted Poisson (\emph{aka} $P_0$) algebra $\Pcal$ in factorization algebras $\FA(M, \Ch(\Nuc))$.
That is, to each open $U \subset M$, the cochain complex $\Pcal(U)$ is equipped with a commutative product $\cdot$ and a degree 1 Poisson bracket $\{-,-\}$;
moreover, each structure map is a map of shifted Poisson algebras.
\end{df}

Note that we always work with the completed projective tensor product~$\widehat{\otimes}$ with nuclear spaces, 
so we will suppress it from the notation.
In other words, we simply write $\Ch(\Nuc)$ instead of $(\Ch(\Nuc), \widehat{\otimes})$.

In parallel, we have the following.

\begin{df}
A \textbf{quantum field theory model} is a Beilinson-Drinfeld (BD) algebra $\Acal$ in factorization algebras $\FA(M, \Ch(\Nuch))$.
That is, to each open $U \subset M$, the cochain complex $\Acal(U)$ is flat over $\CC[[\hbar]]$ and equipped with
\begin{itemize} 
\item an $\hbar$-linear commutative product $\cdot$, 
\item an $\hbar$-linear, degree 1 Poisson bracket $\{-,-\}$, and
\item a differential such that
\[
\d(a \cdot b) = \d(a) \cdot b + (-1)^a a \cdot \d(b) + \hbar \{a,b\}.
\]
\end{itemize}
Moreover, each structure map is a map of BD algebras.
\end{df}

\begin{rem}
We include the condition of flatness to ensure that tensoring need not be derived. 
All of our examples will be free in the appropriate sense.
\end{rem}

Note that for any BD algebra $A$, there is a \textbf{dequantization} 
\[
A^{cl} = A \otimes_{\CC[[\hbar]]} \CC[[\hbar]]/(\hbar)
\]
that is automatically a 1-shifted Poisson algebra. 
Hence every quantum field theory model dequantizes to a classical field theory model.
Given a classical field theory model $\Pcal$, one can ask if it quantizes, i.e., if there exists a quantum field theory $\Acal$ whose dequantization is~$\Pcal$. 

\begin{rem}
The condition on the differential is an abstract version of a property possessed by the divergence operator for a volume form on a finite-dimensional manifold.
Thus, the differential of a BD algebra behaves like a ``divergence operator,'' 
as explained in Chapter 2 of \cite{CoGw},
and hence encodes (some of) the kind of information that a path integral would.
\end{rem}

\subsection{A variant definition: locally covariant field theories}

Above, we have worked on a fixed manifold, but most field theories are well-defined on some large class of manifolds. 
For instance, the free scalar field theory makes sense on any manifold equipped with a metric of some kind. 
Similarly, (classical) pure Yang-Mills theory makes sense on any 4-manifold equipped with a {\it conformal} class of metric and a principal $G$-bundle. 
One can thus replace $\Open(M)$ by a more sophisticated category whose objects are ``manifolds with some structure'' and whose maps are ``structure-preserving embeddings.'' 
(In the scalar field case, think of manifolds-with-metric and isometric embeddings.)  
In a field theory, the fields restrict along  embeddings and the equations of motion are local (but depend on the local structure), 
so that solutions to the equations $Sol$ forms a contravariant functor out of this category. 
Likewise, one can generalize the models of classical or quantum field theory to this kind of setting, as we now do.

\begin{rem}
This discussion is not necessary for what happens elsewhere in the paper,
so the reader primarily interested in our comparison results should feel free to skip ahead.
\end{rem}

\subsubsection{The Lorentzian case}

We begin by replacing the fixed spacetime $\Mcal$ by a coherent system of all such spacetimes.

\begin{df}
Let $\Loc_n$ be the category where an object is a connected, (time-)oriented globally hyperbolic spacetime of dimension $n$ and where a morphism $\chi:\Mcal\to \Ncal$ is an isometric embedding that preserves orientations and causal structure. The latter means that for any causal curve $\gamma : [a,b]\to N$, if $\gamma(a),\gamma(b)\in\chi(M)$, then for all $t	\in ]a,b[$, we have $\gamma(t)\in\chi(M)$. (That is, $\chi$ cannot create new causal links.)
\end{df}

We can extend $\Loc_n$ to a symmetric monoidal category $\Loc_n^{\otimes}$ by allowing for objects that are disjoint unions of objects in $\Loc_n$. The relevant symmetric monoidal structure is the disjoint union $\sqcup$. Note that a morphism in $\Loc_n^\otimes$ must send disjoint components to spacelike-separated regions. 

We are now ready to state what is meant by a locally covariant field theory in our setting, following the definition proposed in \cite{BFV}. We use here a very minimal version of the axioms for the locally covariant field theory functor. From the physical viewpoint, it might be necessary to require some further properties, e.g. dynamical locality (for more details see \cite{spass1,spass2}).

Note that isotony is implicit in the requirement that morphisms in $\Alg^*(\Nuc)^\inj$ are injective.
It is likewise implicit in the following definitions.

\begin{df}\label{LCClassFT}
	A \textbf{locally covariant classical field theory model} of dimension $n$ is a functor $\mathfrak{P} : \Loc_n \to \PAlg^*(\Nuc)^\inj$ 
	such that the \textbf{Einstein causality} holds: given two isometric embeddings $\chi_1:\Mcal_1\rightarrow\Mcal$ and $\chi_1:\Mcal_1\rightarrow\Mcal$ whose images $\chi_1(\Mcal_1)$ and $\chi_2(\Mcal_2)$ are spacelike-separated, the subalgebras 
		\[
		\mathfrak{P}\chi_1(\mathfrak{P}(\Mcal_1)) \subset \mathfrak{P}(\Mcal) \supset \mathfrak{P}\chi_2(\mathfrak{P}(\Mcal_2))
		\] 
		Poisson-commute, i.e., we have
		\[
		\Poi{\mathfrak{P}\chi_1(a_1)}{\mathfrak{P}\chi_2(a_2)}=\{0\}\,,
		\]
		for any $a_1 \in \mathfrak{P}(\Mcal_1)$ and $a_2 \in \mathfrak{P}(\Mcal_2)$.
\end{df}

\begin{df}
	A \textbf{locally covariant quantum field theory model} of dimension $n$ is a functor $\fA:\Loc_n\rightarrow\Alg^*(\Nuch)^\inj$ such that \textbf{Einstein causality} holds:
	\begin{quote}
Given two isometric embeddings $\chi_1:\Mcal_1\rightarrow\Mcal$ and $\chi_1:\Mcal_1\rightarrow\Mcal$ whose images $\chi_1(\Mcal_1)$ and $\chi_2(\Mcal_2)$ are spacelike-separated, the subalgebras 
		\[
		\fA\chi_1(\fA(\Mcal_1)) \subset \fA(\Mcal) \supset \fA\chi_2(\fA(\Mcal_2))
		\] 
		commute, i.e., we have
		\[
		[\fA\chi_1(a_1),\fA\chi_2(a_2)]=\{0\}\,,
		\]
		for any $a_1 \in \fA(\Mcal_1)$ and $a_2 \in \fA(\Mcal_2)$.
	\end{quote}
\end{df}

\begin{df}\label{OnShell}
A model $\Pfrak$ ($\fA$) is called \textbf{on-shell} if it satisfies in addition the \textbf{time-slice axiom}: If $\chi:\Mcal\rightarrow \Ncal$ contains a neighborhood of a Cauchy surface $\Sigma\subset\Ncal$, then the map $\Pfrak\chi: \Pfrak(\Mcal) \to \Pfrak(\Ncal)$ (respectively, $\fA\chi: \fA(\Mcal) \to \fA(\Ncal)$) is an isomorphism.
\end{df}

\begin{rem}
The category $\Alg^*(\Nuc)^\inj$ has a natural symmetric monoidal structure via the completed tensor product $\widehat{\otimes}$. 
Then Einstein causality can be rephrased as the condition that $\fA$ is a symmetric monoidal functor from $\Loc_n^{\otimes}$ to $\Alg^*(\Nuc)^{\inj,\widehat{\otimes}}$, as discussed in~\cite{BFIR}.
\end{rem}

\subsubsection{The factorization algebra version}

Let us begin with the simplest version.

\begin{df}
Let $\Emb_n$ denote the category whose objects are smooth $n$-manifolds and whose morphisms are open embeddings. It possesses a symmetric monoidal structure under disjoint union.
\end{df}

Then we introduce the following variant of the notion of a prefactorization algebra. Below, we  will explain the appropriate local-to-global axiom.

\begin{df}
A \textbf{prefactorization algebra on $n$-manifolds} with values in a symmetric monoidal category $(\bC, \otimes)$ is a symmetric monoidal functor from $\Emb_n$ to~$\bC$.
\end{df}

This kind of construction works very generally. For instance, if we want to focus on Riemannian manifolds, we could work in the following setting.

\begin{df}
Let $\Riem_n$ denote the category where an object is Riemannian $n$-manifold $(M,g)$ and a morphism is open isometric embedding. It possesses a symmetric monoidal structure under disjoint union.
\end{df}

\begin{df}
A \textbf{prefactorization algebra on Riemannian $n$-manifolds} with values in a symmetric monoidal category $(\bC, \otimes)$ is a symmetric monoidal functor from $\Riem_n$ to $\bC$.
\end{df}

\begin{rem}
In these definitions, the morphisms in $\Riem_n$ form a set, but one can also consider an enrichment so that the morphisms form a space, perhaps a topological space or even some kind of infinite-dimensional manifold. 
This kind of modification can be quite useful.
For instance, this would allow to view isometries (i.e., isometric isomorphisms) as a Lie group, 
rather than as a discrete group.
\end{rem}

In general, let $\Gcal$ denote some kind of local structure for $n$-manifolds, such as a Riemannian metric or complex structure or orientation. In other words, $\Gcal$ is a sheaf on $\Emb_n$. A {\it $\Gcal$-structure} on an $n$-manifold $M$ is then a section $G \in \Gcal(M)$. There is a category $\Emb_\Gcal$ whose objects are $n$-manifolds with $\Gcal$-structure $(M,G_M)$ and whose morphisms are $\Gcal$-structure-preserving embeddings, i.e., embeddings $f : M \hookrightarrow N$ such that $f^* G_N = G_M$. This category is fibered over $\Emb_\Gcal$. One can then talk about prefactorization algebras on $\Gcal$-manifolds.

We now turn to the local-to-global axiom in this context.

\begin{df}
A \textbf{Weiss cover} of a $\Gcal$-manifold $M$ is a collection of $\Gcal$-embeddings $\{ \phi_i : U_i \to M\}_{i \in I}$ such that for any finite set of points $x_1,\ldots,x_n \in M$, there is some $i$ such that $\{x_1,\ldots,x_n\} \subset \phi_i(U_i)$.
\end{df}

With this definition in hand, we can formulate the natural generalization of our earlier definition.

\begin{df}
A \textbf{factorization algebra on $\Gcal$-manifolds} is a symmetric monoidal functor $\Fcal: \Emb_{\Gcal} \to \bC$ that is a cosheaf in the Weiss topology.
\end{df}

One can mimic the definitions of models for field theories in this setting.

\begin{df}
A \textbf{$\Gcal$-covariant classical field theory} is a 1-shifted (\emph{aka} $P_0$) algebra $\Pcal$ in factorization algebras $\FA(\Emb_{\Gcal},\Ch(\Nuc))$.
\end{df}

\begin{df}
A \textbf{$\Gcal$-covariant quantum field theory} is a Beilinson-Drinfeld (BD) algebra $\Acal$ in factorization algebras $\FA(\Emb_{\Gcal}, \Ch(\Nuch))$.
\end{df}

\section{Comparing the definitions}

Now that we have the key definitions in hand, 
we can restate the questions (\ref{questions}) more sharply.
\begin{enumerate}
\item In the CG formalism a model for a field theory defines a functor on the poset $\Open(M)$ of all open subsets. By contrast, the FR formalism a model defines a functor on the subcategory $\Caus(\Mcal)$. Why this restriction? How should one extend an FR model to a functor on the larger category of all opens? Is it a factorization algebra?
\item In the FR formalism, a model assigns a Poisson algebra (or $*$-algebra) to each open in $\Caus(\Mcal)$, whereas in the CG formalism, a model assigns a shifted Poisson algebra (or BD algebra) to every open.  Are these rather different kinds of algebraic structures related?
\end{enumerate}
We will address these questions in the specific example of free scalar field theory.
In the conclusion, we draw some lessons and hints about the case of interacting theories and non-scalar theories. 

\subsection{Free field theory models}

We now turn to stating our main result, which is a comparison of the FR and CG procedures.
First, we need to state what each formalism accomplishes with the free field.
In the following sections, we spell out in detail how to construct the models asserted and prove the propositions.

We remark that these statements are likely hard to understand at this point;
the point we emphasize here is just that we get models in both the FR and CG senses.

\def\Pfrak{\mathfrak{P}}

\begin{prop}
\label{FRconstruction}
Let $\Mcal=(M,g)$ be a $d$-dimensional, oriented, time-oriented, and globally hyperbolic spacetime with the metric $g$ of signature $(+,-,\dots,-)$. 
Given a vector bundle $\pi:E\rightarrow M$ and a Green hyperbolic operator $P$, there is a classical field theory model $\Pfrak$ such that 
\begin{itemize}
\item The space of fields $\F(\Ocal)$ is the space generated (as a commutative algebra) by continuous linear functionals on distributional solutions of $P\phi  = 0$ on $\Ocal$;
\item the commutative product $\cdot$ is the obvious pointwise  product of the space of functionals on the solution space of $\Ocal$;
\item the Poisson bracket is the Peierls bracket $\Poi{.}{.}$ (see \cite{Pei} and the remark below).
\end{itemize}
There is a quantum field theory model $\fA$ on $\Mcal$ such that
for each $\Ocal \in \Caus(\Mcal)$, the associative $\CC[[\hbar]]$-algebra $\fA(\Ocal)$ is generated topologically by continuous linear functionals on distributional solutions to $P\phi= 0$ and the product $\star$ satisfies the relation
\[
[ F,G]_\star = i\hbar \Poi{F}{G}
\]
for linear functionals $F,G$.
\end{prop}

\begin{rem}
	In Proposition~\ref{FRconstruction} we mention the Peierls bracket, which is a Poisson bracket introduced by Peierls in \cite{Pei}. It is defined using the Lagrangian formalism (in contrast to the usual canonical bracket introduced in the Hamiltonian framework), in a fully covariant way, as a bracket on the algebra of functions on the space of solutions to the equations of motion. A key feature is that it has a well-defined off-shell extension to a Poisson bracket on the space of all functionals on the configuration space (see \cite{DF03}). We come back to this structure in Section~\ref{subsec: shifting}.   
\end{rem}

\begin{rem}
	Note that allowing for distributional solutions enforces a restriction on the dual, so that $\F$ is generated by functionals of the form $\ph\mapsto \int \ph f$, where $f$ is a compactly supported test density on $M$, modulo the ideal generated by functionals of the form $\ph\mapsto \int P\ph f$.
\end{rem}

Analogously, the CG approach to free theories applies to Lorentzian manifolds, as we show below, and we obtain the following.

\def\dvol{{\mathrm dvol}}

\begin{prop}
\label{CGconstruction}
Let $\Mcal=(M,g)$ be a $d$-dimensional, oriented, time-oriented, and globally hyperbolic spacetime with the metric $g$ of signature $(+,-,\dots,-)$. 
Given a vector bundle $\pi:E\rightarrow M$ and a Green hyperbolic operator $P$, there is a classical field theory  model $\Pcal$,
i.e., a $P_0$ algebra in factorization algebras $\Pcal$ on $M$ where 
for each open $U \subset M$, the commutative dg algebra $\Pcal(U)$ is generated topologically by the cochain complex
\[
\Dcal(U)[1] \xrightarrow{P} \Dcal(U).
\]
It is equipped with the degree 1 Poisson bracket by using the Leibniz rule to extend the pairing on generators
\[
\Poi{f_{-1}}{f_0} = \int_U f_{-1} f_0 \,\dvol_g,
\]
with $f_{-1}$ in degree -1 and $f_0$ in degree 0.
There is a quantum field theory model $\Acal$ for the free theory with operator $P$,
i.e., a BD algebra in factorization algebras on $M$ where $\Acal(U)$ is the BD quantization of $\Pcal(U)$ whose differential is $\d_{\Pcal} + \hbar \triangle$.
This BV Laplacian $\triangle$ is determined by the conditions that it is compatible with the shifted Poisson bracket on quadratic terms and that it vanishes on constants and on linear generators.
\end{prop}

To summarize, we have the following collection of models.
\begin{center}
	\begin{tabular}{|c|c|c|}
		\hline
		& \textbf{FR} & \textbf{CG}\\
		\hline
		classical &$\Pfrak$ &$\Pcal$\\
		\hline
		quantum & $\fA$ &$\Acal$\\
		\hline
	\end{tabular}
	\end{center}

We remark that these propositions might seem distinct on the surface, 
since the CG result involves cochain complexes while the FR result does not.
This distinction disappears when one examines the actual constructions:
both use a BV framework, and hence the FR construction actually builds a cochain-level functor as well.
We formalize a dg version of pAQFT in Section~\ref{sec:dgpaqft} below, which makes the comparison even more obvious.

\subsection{The comparison results}

With these models in hand, a clean comparison result can be stated.
Before making the formal statement, we first explain it loosely.

The basic idea is that we can restrict the factorization algebras to $\Caus(\Mcal)$, 
since every causally-convex open is manifestly an open subset and 
hence there is an inclusion functor $\Caus(\Mcal) \hookrightarrow \Open(M)$.
The restrictions $\Pcal|_{\Caus(\Mcal)}$ and $\Acal|_{\Caus(\Mcal)}$ can be further simplified by taking cohomology on each $\Ocal \in \Caus(\Mcal)$:
we define functors
\[
H^*(\Pcal)|_{\Caus(\Mcal)}(\Ocal) = H^*(\Pcal(\Ocal))
\]
and
\[
H^*(\Acal)|_{\Caus(\Mcal)}(\Ocal) = H^*(\Acal(\Ocal)).
\]
This cohomology is concentrated in degree zero, 
which we verify as we prove the comparison results.
(For gauge theories the cohomology is not necessarily concentrated in degree zero.)

We then want to compare the functors $H^0(\Pcal/\Acal)|_{\Caus(\Mcal)}$ 
to the corresponding FR functors.
The targets of these functors, however, are different.
For instance, $\Pcal$ takes values in 1-shifted Poisson algebras and hence so does $H^0 \Pcal$ (although the bracket must then be trivial for degree reason).
By contrast, $\Pfrak_\pol$ takes values in Poisson $*$-algebras.
(The subscript $\pol$ indicates that we will use polynomial algebras for the comparison.
See Remark~\ref{reg vs pol} for a discussion of natural variant constructions, notably with regular functions.)
Hence we apply forgetful functors to land in the same target category.
We now state our comparison result for the classical level.

\begin{thm}[Comparison of classical models]
\label{thm: classical cf}
There is
a natural transformation 
\[
\iota^{cl}: \fc \circ \Pcal|_{\Caus(\Mcal)} \Rightarrow \fc \circ \Pfrak_\pol
\]
of functors to commutative dg algebras $\CAlg(\Ch(\Nuc))$, and this natural transformation is an isomorphism of commutative dg algebras.
Thus, there is a natural isomorphism
\[
H^0(\iota^{cl}) : \fc \circ H^0({\Pcal})|_{\Caus(\Mcal)} \Rightarrow \fc\circ H^0(\Pfrak_\pol)
\]
of functors into commutative algebras~$\CAlg(\Nuc)$.
\end{thm}

This identification is not surprising, as both approaches end up looking at (a class of) functions on solutions to the equations of motion.

We can extend to the quantum level, 
but here we need the forgetful functor $\fv: \Alg^*(\Nuch) \to \Nuch$,
since $H^0 \Acal$ is \emph{a priori} just a vector space.

\begin{thm}[Comparison of quantum models]
\label{thm: quantum cf}
There is a natural transformation 
\[
\iota^q: \Acal|_{\Caus(\Mcal)} \Rightarrow \fv \circ \fA_\pol
\]
of functors to $\Ch(\Nuch)$, and this natural transformation is an isomorphism of cochain complexes.
Thus, there is a natural isomorphism
\[
H^0(\iota^q) : H^0({\Acal})|_{\Caus(\Mcal)} \overset{\cong}{\Rightarrow} \fv \circ H^0(\fA_\pol).
\]
Modulo $\hbar$, this isomorphism agrees with the isomorphism of classical models.

In fact, on each $\Ocal \in \Caus(\Mcal)$, the map $\iota^q$ is an isomorphism of cochain complexes
\[
\alpha_{\partial_{G^{\rm D}}}: \fv (\Pcal(\Ocal)[[\hbar]]) \xto{\cong} \fv ( \Acal(\Ocal))
\]
determined by the analytic structure of the equations of motion.
Under this identification, the factorization structure of $\Acal$ agrees with the time-ordered version of the product $\star_{G^{\mathrm C}}$ on~$\fA_\pol$.
\end{thm}

The second part of the quantum comparison theorem is likely cryptic at the moment,
as it involves the notations $\alpha_{\partial_{G^{\rm D}}}$ and $\star_{G^{\mathrm C}}$ and the terminology ``time-ordered products''
that we have not yet introduced.
We will explain these in the next section, 
as they are the key to understanding how the two approaches to QFT relate.
We wish to clarify now, however, the main thrust of the theorem.

To paraphrase the theorem, the factorization algebra $\Acal$ knows information equivalent to the QFT model $\fA$. 
Conversely, one can recover from $\fA$, the precosheaf structure of $\Acal$ restricted to~$\Caus(\Mcal)$.
(This assertion is true when one uses the cochain-level refinement of $\fA$, 
as we will see below when reviewing the explicit FR construction.)

What is even more important is that there is a natural way to identify the algebra structures on either side.
We will show that one can read off the FR deformation quantization $\fA$ from the CG factorization algebra $\Acal$ and conversely.

\subsection{Key ingredients of the argument}

In this section we recall the relevant background about quantum field theories,
notably the notions appearing in the theorems above.
We explain, in particular, how the associative algebra structure appears in $\fA$,
why it is important to the physics, 
and how it relates to constructions in the CG formalism.

\subsubsection{Time-ordered products and why they are important}
\label{sec: time ordering}

A free quantum theory is fully characterized by its net $\fA$, 
but in order to deal with interactions, 
we need one more structure, namely the time-ordered product. 
Constructing time-ordered products of free fields is an intermediate step towards building interacting fields. 
The idea is analogous to using the interacting picture in quantum mechanics. 
Namely, we would like to apply the Dyson formula to define the time evolution operator as a time-ordered exponential:
\begin{align*}
U(t,s)
&=e^{itH_0}e^{-i(t-s)(H_0+H_I)}e^{-isH_0}\\
&=1+\sum_{n=1}^\infty\frac{i^n}{n!}\int_ {([s,t]\times\mathbb R^3)^n}T(\normOrd{H_I(x_1)}\dots \normOrd{H_I(x_n)})\, \d^{4n}x.
\end{align*}
Here $\normOrd{-}$ denotes the normal-ordering, 
$T$ denotes time-ordering, $H_0$ denotes the free Hamiltonian,
and $H_I$ denotes the the interacting Hamiltonian,
which is the operator-valued function of spacetime 
\[
H_I(x)=e^{iH_0x^0}\normOrd{H_I(0,\mathbf x)}e^{-iH_0x^0},
\]
where $x = (x^0, \mathbf x)$ denotes a point in spacetime.
Heuristically, one could use the unitary map defined above to obtain interacting fields as
\be\label{phiI}
\ph_I(x)=U(x^0,s)^{-1}\ph(x)U(x^0,s)=U(t,s)^{-1}U(t,x^0)\ph(x)U(x^0,s)\,,
\ee
for $s<x^0<t$.

To put this approach on a rigorous footing,
the framework of pAQFT replaces the Dyson series by the \textit{formal S-matrix}:
\[
\Scal(\lambda V)=1+\sum_{n=1}^{\infty}\frac{1}{n!}\left(\frac{i\lambda}{\hbar}\right)^{n}\Tcal_n(V^{\otimes n})\,,
\]
where $V\in\Ci(\Ecal,\CC)$ is the interaction functional, 
each $\TT_n$ is a linear map from appropriate domain in $\Ci(\Ecal,\CC)^{\otimes n}$ to $\Ci(\Ecal,\CC)[[\hbar]]$, 
and the above expression is to be understood as a power series in the coupling constant $\lambda$ 
with coefficients in Laurent series in $\hbar$. 
Constructing $\Scal$ is then reduced to construction of $\TT_n$'s, 
which in turn is done using the Epstein-Glaser renormalization~\cite{EG}. 

In \cite{FR3} it was shown that the maps $\TT_n$ arise from a commutative, associative product $\T$ defined on a certain domain of  $\Ci(\Ecal,\CC)[[\hbar]]$. 
Here, to avoid problems related to renormalization, 
we will consider $\T$ on the subset $\F_{\reg}[[\hbar]]$ of  $\Ci(\Ecal,\CC)[[\hbar]]$.
(See Definition \ref{df: regular} for its description.)

More abstractly, we introduce the following notion.

\begin{df}\label{timeorderedprod}
Given the classical free off-shell theory $\mathfrak{P}$ and its quantization $\fA$, 
the \textbf{time-ordered product} is realized as a quadruple $(\fP_0,\fA_T,\xi,\Tcal)$ 
consisting of:
\begin{itemize}
	\item $\fP_0\subset \fP$, a subfunctor of the classical theory functor that characterizes the domain of definition of the time-ordered product,
	\item  a functor 	
	\[
	\fA_T \colon \Caus(\Mcal) \to \CAlg^*(\Nuc_\hbar)\,,
	\]
	which gives the time-ordered product as a commutative product,
	\item a natural embedding 
	\[
	\xi:\fv\circ\fA_T \Rightarrow \fv\circ\fA\,,
	\] 
which identifies $\fA_T$ as a subspace of $\fA$, but only as a vector space,
	\item and a natural isomorphism of commutative algebras
	\[
	\Tcal: \mathfrak{c}\circ\mathfrak{P}_0[[\hbar]] \Rightarrow \fA_T,
	\] 
\end{itemize}
such that for any pair of inclusions $\psi_i:\Ocal_i\rightarrow\Ocal$ in $\Caus(\Mcal)$,
 if $\psi_1(\Ocal_1)\prec \psi_2(\Ocal_2)$, then 
\[\xi_\Ocal\circ m_{\TT}\circ(\fA_T\psi_2\otimes \fA_T\psi_1)=m_{\star}\circ(\fA\psi_2\circ\xi_{\Ocal_2}\otimes\fA\psi_1\circ\xi_{\Ocal_1} )\,,
\]
where $m_{\TT}$/$m_\star$ is the multiplication with respect to the time-ordered/star product
and the relation ``$\prec$'' means ``not later than,'' i.e., there exists a Cauchy surface in $\Ocal$ that separates $\psi_1(\Ocal_1)$ and $\psi_2(\Ocal_2)$. 
\end{df}

The natural transformation $\TT$ provides an equivalence between the time-ordered product $\T$ of $\fA_T$ and the classical product $\cdot$ of $\mathfrak{c}\circ\mathfrak{P}[[\hbar]]$.
In formulas, we have
\[
F\T G\doteq \Tcal_{\Ocal}(\Tcal_{\Ocal}^{-1}F\cdot \Tcal_{\Ocal}^{-1}G)\,,
\]
where $F,G\in\fA_T(\Ocal)$.

\begin{rem}
Note that the existence of the natural isomorphism $\Tcal$ and the fact that $\fA$ is a quantization of $\fP$ imply that there is a natural embedding $\fv\circ\fA_T \Rightarrow \fv\circ\fA$, but $\xi$ does not have to coincide with this embedding. However, one can choose $\xi$ to be the identity map and choose the quantization map in definition \eqref{df:quantization} as $\xi\circ\Tcal$. Such choice has been used in \cite{HR16} and it greatly simplifies the construction of the interacting star product.
\end{rem}

\begin{rem}
	The way in which we phrased definition~\ref{timeorderedprod} is general enough to cover also the situation where renormalization is needed. For the purpose of this paper (where we work only with regular functionals, so no renormalization  is needed), we can take $\fP_0$ to be just $\fP$ and $\xi$ to be a natural isomorphism.
\end{rem}
This definition intertwines the product on classical and quantum observables in a nontrivial way,
and as mentioned in Theorem \ref{thm: quantum cf},
it is the key to relating the algebraic structures on $\Acal$ and $\fA$.
Hence our goal is to construct this time-ordered product on free fields and show how it appears in the comparison map $\iota^q$.
We explain that in the next few subsections, which are thus somewhat technical.
The main ingredient is various propagators, or Green's functions, for the equation of motion.\footnote{Indeed, this project began when we realized we were using the same tricks with propagators.}

\subsubsection{Propagators}\label{Propagators}

We introduce the four key propagators, which are linearly related.
%\begin{table}[t]
%{\renewcommand{\arraystretch}{1.2}
\begin{center}
	\begin{tabular}{|c|c|}
		\hline
		\textbf{Symbol} & \textbf{Meaning}\\
		\hline
		$G^{\mathrm{A}}$ &advanced propagator\\
		\hline
		$G^{\mathrm{R}}$ &retarded propagator\\
		\hline 
		$G^{\mathrm{C}}\doteq G^{\mathrm{R}}-G^{\mathrm{A}}$& causal propagator\\
		\hline
		$G^{\mathrm{D}}\doteq \frac{1}{2}\left(G^{\mathrm{R}}+G^{\mathrm{A}}\right)$& Dirac propagator\\
		\hline
	\end{tabular}
\end{center}%}
%\end{table}
Note that the causal propagator is \emph{not} a Green's function but rather a bisolution (i.e. a distributional solution in both arguments), so
\[
P\circ G^{\mathrm C} = 0
\]
whereas for the others 
\[
P\circ G^{\mathrm A/ \mathrm R/ \mathrm D} = \delta_{\Delta},
\]
where $\delta_{\Delta}$ denotes the delta function of the diagonal $M \hookrightarrow M \times M$.
The advanced (respectively, retarded) propagator $G^A(x,y)$ has the property that it vanishes when the first point $x$ is in the ``past'' (respectively, ``future'') of $y$.

The causal propagator $G^{\mathrm C} $ is related to another important type of  bi-solution of $P$, namely the Hadamard function.

\begin{df}
	A \textbf{Hadamard function} $G^+$ for a normally hyperbolic operator $P$ is a distribution in ${\Dcal_2'}^{\sst \CC}(M)$ satisfying:
	\begin{enumerate}
		\item $G^+$ is a distributional bi-solution for $P$. 
		\item \label{im} $2\,\mathrm{Im}\,G^+=G^{\mathrm C}$
		\item $G^+$ fulfills the microlocal spectrum condition: 
		its wavefront set\footnote{
The \textbf{wavefront set} of a distribution $u\in \Dcal'(\RR^n)$ is a subset of $\dot{T}^*\RR^n$ (the co-tangent bundle minus the zero section) 
characterizing singular points and singular directions of $u$ (i.e.,directions in the cotangent space in which the Fourier transform does not decay rapidly). 
More precisely, the complement of $\WF(u)$ in $\dot{T}^*\RR^n$ is the set of points $(x,{k}) \in \dot{T}^*\RR^n$ for which there exists a ``bump function'' $f \in\Dcal(\RR^n)$ with $f(x)=1$ and
an open conic neighborhood $C$ of ${k}$, with
\[
\sup_{{k}\in C}(1+|{k}|)^N|\widehat{f \cdot u}({k})|<\infty\qquad\forall N \in \NN_0\,.
\]
This notion easily generalizes to open subsets of $\RR^n$ and to manifolds \cite{Hoer1}. 
Note that if a $WF(u)=\varnothing$, then $u$ is a smooth function. 
} is
		\[
		\WF(G^+)=\{(x,k;x',-k')\in \dot{T}M^2|(x,k)\sim(x',k'), k\in (\overline{V}_+)_x\}\,,
		\]
		where $(x,k)\sim(x',k')$ means that there exists a null geodesic connecting $x$ and $x'$ and $k'$ is the parallel transport of $k$ along this geodesic, $\dot{T}$ denotes the tangent bundle with the zero section removed and $(\overline{V}_+)_x$ is the closure of the cone of positive, future-pointing vectors in $T_x^*M$.
		\item  $G^+$ is of positive type, i.e. $\left<G^+,f\otimes\bar{f}\right>\geq 0$, for all non-zero $f\in\Dcal(M)\otimes\CC$. The bracket denotes the dual pairing between distributions and test functions.
	\end{enumerate}
\end{df}

Note that any $G^+$  can be written as
\[
G^+=\frac{i}{2}  G^{\mathrm C}+H\,,
\]
where $H$ is a real, symmetric distributional bi-solution for $P$. 
The \textbf{Feynman propagator} associated with this Hadamard function $G^+$ is then defined as
\[
G^{\rm F}=iG^{\rm D}+H\,.
\] 
For notational convenience, we refer to both bi-solutions and Green's functions as \textit{propagators}. 

We extend our table of propagators with 
\begin{center}
	\begin{tabular}{|c|c|}
		\hline
		\textbf{Symbol} & \textbf{Meaning}\\
		\hline
		$G^+\doteq \frac{i}{2}G^{\mathrm{C}}+H$&Hadamard function\\
		\hline
		$G^{\mathrm{F}}\doteq iG^{\mathrm{D}}+H$& Feynman propagator for $G^+$\\
		\hline
	\end{tabular}
\end{center}

\medskip
\noindent The propagators listed above can be used to define 
\begin{enumerate}
	\item new products on the observables and 
	\item automorphisms of the (underlying vector spaces of) observables.
\end{enumerate}
In the following sections we will explain these constructions in detail.
\subsubsection{Smooth maps between locally convex vector spaces}
In this work, we model observables as $\CC[[\hbar]]$-valued functions on the space of solutions to some linear differential equations (elliptic or hyperbolic). On various stages of the comparison between the CG and FR approaches, we also consider functions between arbitrary locally convex topological vector spaces.  For such functions one can introduce the notion of \textit{smoothness}, which we are going to use later.
We start by introducing smooth functions on $\Ecal(M)$. For future convenience, we state here the general definition of a functional derivative of a function between two Hausdorff locally convex spaces.

\begin{df}
	Let $U$ be an open subset of a Hausdorff locally convex space $X$
	and let $F$ be a map from $U$ to a Hausdorff locally convex space
	$Y$. Then $F$ \textbf{has a derivative at $x\in U$
		in the direction of $v \in X$} if the following limit
	\[
	\left<F^{(1)}(x),v\right> :=  \lim_{t\to 0}\frac{F(x+tv)-F(x)}{t},
	\]
	exists. The function $F$ is said to have a 
	\textbf{G\^ateaux differential}
	at $x$ if $\left<F^{(1)}(x),v\right> $ exists for every $v\in X$. $F$ is $C^1$ or \textbf{Bastiani differentiable} \cite{Bas64,Mich38} on $U$
	if $F$ has a 
	G\^ateaux differential at every $x\in U$
	and the map $F^{(1)}: U\times X\to Y$
	defined by $(x,v)\mapsto\left<F^{(1)}(x),v\right> $ is continuous on $U\times X$.
\end{df}

This definition applies in particular to functions from $\Ecal(M)$ to $\CC$. Iterating it $n$ times we define $C^n$-functionals of $\Ecal(M)$. If a functional is  $C^n$ for all $n\in\NN$, we call it (Bastiani) \textit{smooth} and write $F\in\Ci(\Ecal(M),\CC)$. Detailed properties of such functionals have been investigated in \cite{BDGR}.

 Localization properties of smooth functionals on 
 $\Ecal(M)$ are characterized by the notion of \textit{spacetime support}:
\be\label{spsupp}
\supp F\doteq \{x\in M \,|\, \forall\Ocal\ni x\textrm{ open},\exists \ph,\psi\in\Ecal \textrm{ s.t. }\suppressfloats\supp_M(\psi)\subset\Ocal\ \textrm{and}\ F(\ph+\psi)\neq F(\ph)\}
\ee
where $\supp_M(\psi)$ denotes the support of $\psi$ as a function on $M$.
This definition satisfies the equality 
\be\label{support}
\supp F\doteq \overline{\bigcup_{\ph\in\Ecal(M)} \supp_M( F^{(1)}(\ph)})\,.
\ee
(See \cite[Lemma~3.3]{BDGR}.)

Among all smooth functionals, a special role is played by the regular ones. 
Regularity properties of a smooth functional are  formulated in terms of the wavefront (WF) sets of its derivatives, since $F^{(n)}(\ph)\in{\Ecal_n'}^{\!\sst \CC}(M)$. (Recall from Section \ref{subsec:notation} that this notation means compactly supported distributional sections on $M^n$, and the superscript $\CC$ denotes the complexification.) See \cite[section 3.4]{BDGR} for a proof.

\begin{df}
	\label{df: regular}
	A functional $F$ is \textbf{regular} if $\WF(F^{(n)}(\phi))$ is empty for all $\phi\in\Ecal$ and $n\in\NN$. This condition is equivalent to having
	\[
	F^{(n)}(\phi)\in\Dcal_n^{\sst \CC}(M)\,,
	\]
	where we implicitly use the pairing on the fiber of $F$ to identify sections of $F$ with sections of $F^*$.
	We denote the space of regular functionals by~$\F_\reg(M)$.
\end{df}

\begin{df}\label{df: poly}
	A functional $F$ is called \textbf{polynomial} if it can be written as a linear combination of functionals of the form
	\[
	\phi\mapsto \left<\ph^{\otimes n},f\right>\,,
	\]
	where $f\in \Dcal_n^{\sst \CC}(M)$ and the dual pairing is the pairing between $\Ecal_n(M)$ and $\Dcal_n^{\sst \CC}(M)\subset {\Ecal'}^{\sst \CC}_n(M)$. We denote the space of polynomial functionals by~$\F_\pol(M)$. 
\end{df}
Clearly, $\F_\pol(M)\subset \F_\reg(M)$.

\subsubsection{Exponential products}\label{sec:exp:prod}

A propagator $G$ is an element of  $\Dcal_2'(M)$ and as such, can be viewed as a bi-vector field on $\Ecal(M)$. 
To make this precise, we first need to make sense of the tangent bundle $T\Ecal(M)$. To this end, we have to equip $\Ecal(M)$ with an infinite dimensional manifold structure. One obvious choice is to use the Fr{\'e}chet topology of $\Ecal(M)$\footnote{This is not the most optimal choice, as discussed e.g. in \cite{Michor}. Another possibility is discussed in Section~\ref{sec:polyvect}}. With this choice,
 since $\Ecal(M)$ is a vector space, we find that the tangent bundle $T\Ecal(M)$ can be identified with $\Ecal(M)\times \Ecal(M)$. Similarly, the total space of the bundle arising from tensoring the tangent bundle with itself (as a vector bundle over the base manifold $\Ecal(M)$) can be identified with $\Ecal(M)\times \Ecal(M)^{\otimes 2}$ and hence admits a natural completion to the vector bundle $\Ecal(M)\times \Dcal'_2(M) \to \Ecal(M)$. The propagator $G$ is a constant section of this completed bundle. 

Let $F$ be a smooth functional on $\Ecal(M)$ with smooth derivatives, i.e. $F\in\F_{\reg}(M)$. We use the suggestive notation $\partial_G$ to denote the differential operator constructed from $G$ as follows:
\[
\partial_G F\doteq \iota_G(F^{(2)})\,,
\]
where $\iota_G(F^{(2)})(\ph)=\left<G,F^{(2)}(\ph)\right>$ and the pairing is induced by the duality between $\Dcal'_2(M)$, where $G$ lives, and $\Dcal_2^{\sst \CC}(M)$, where $F^{(2)}(\ph)$ lives.

The propagator $G$ can also be viewed as a section of $\Ci(\Ecal(M)\times\Ecal(M), \Dcal'_2(M))$. 
We write 
\[
\widetilde\partial_G: \Ci(\Ecal(M)\times\Ecal(M),\CC) \to \Ci(\Ecal(M)\times\Ecal(M),\CC)
\]
to denote the differential operator defined on a tensor product by
\[
\widetilde{\partial}_G(F_1\otimes F_2)\doteq \iota_G (dF_1\boxtimes dF_2)\,,
\]
where $dF_1\boxtimes dF_2$ is an element of $\Ci(\Ecal(M)\times\Ecal(M),\Dcal_2^{\sst\CC}(M))$, so the insertion makes sense. 
This operator determines an operator denoted $e^{\hbar \widetilde{\partial}_G}$,
which we understand as a formal power series in $\hbar$ whose $\hbar^n$ term is the $n$-fold power of~$\widetilde{\partial}_G$.

\begin{df}
Given a constant-coefficient second-order differential operator $\widetilde{\partial}_G$,
we define an \textbf{exponential product} by
\[
F_1 \star_G F_2 = m \circ e^{\hbar \widetilde{\partial}_G} (F_1 \otimes F_2),
\]
where $m$ denotes the usual commutative multiplication, i.e. pullback by the diagonal map $\ph\mapsto \ph\otimes \ph$.
\end{df}

Direct computation shows that $\star_G$ is associative.
As we will see later, the product structure of $\fA(\Ocal)$ comes from~$\star_{G^{\rm C}}$. 

One can also define automorphisms (of underlying vector spaces) by
\begin{equation}\label{def of alpha}
\alpha_G(F) = e^{\frac{\hbar}{2} \partial_G} F.
\end{equation}
(When $G_2 - G_1$ is symmetric, $\alpha_{G_2 - G_1}$ determines an isomorphism of algebras from the $\star_{G_1}$ product to the $\star_{G_2}$ product.)

\subsection{The time-slice axiom and the algebra structures}
\label{subsec: time slice}

A dissatisfying aspect of the comparison results is that they involve forgetful functors:
it seems like we ignore the crucial Poisson, respectively associative, algebra structures,
although the constructions (e.g., with propagators) certainly involved them.
As discussed, these algebraic structures play a crucial role in physics and hence appear in the axioms of AQFT,
but they are not built into the CG construction.
It is natural to ask how to resolve this tension.

We provide two perspectives that we feel clarify substantially this issue,
one rooted in a key maneuver of the FR work and another using results in higher algebra in conjunction with the CG perspective.
Both depend on a prominent and useful feature of these examples:
they satisfy the time-slice axiom.
That is, if $\Sigma$ is a Cauchy hypersurface for nested opens $\Ocal \subset \Ocal'$ in $\Caus(\Mcal)$, 
then $\fA(\Ocal) \to \fA(\Ocal')$ is an isomorphism.
The factorization algebra satisfies a cochain-level analog of this axiom: 
the map $\Acal(\Ocal) \to \Acal(\Ocal')$ is a quasi-isomorphism.

The time-slice property suggests formulating a version of $\fA$ and $\Acal$ living just on a Cauchy hypersurface itself.
We will state a natural comparison result before explaining the idea why one should exist from the CG perspective.

\subsubsection{The result on comparison of algebraic structures}

We now turn to formulating a precise framework for describing how the algebraic structures intertwine.

Let $\Sigma$ be a Cauchy hypersurface of $\Mcal$, 
which inherits a canonical Riemannian metric.
Consider the collection of open balls in $\Sigma$ such that each ball $B$ contains a point $x \in B$ such that the closure of $B$ is contained inside the injectivity radius of $x$. 
This definition plays nicely with inclusion, so let $\cball(\Sigma)$ denote the full subcategory of $\Open(\Sigma)$ for these opens.\footnote{Each ball is diffeomorphic to an ordinary ball in Euclidean space via the exponential map at $x$, 
and thus no funny topological business appears. 
The spirit of \cite{spass1} is at work here. 
Their slightly larger category of ``Cauchy balls'' would work equally well for our purposes.}
Each ball has an associated {\em diamond} $D_\Mcal(B) \in \Caus(\Mcal)$, 
which is the union $D^+_\Mcal(B) \cup  D^-_\Mcal(B)$, 
where $D^{\pm}_\Mcal(B)$ consists of every point $p$ in the future/past of $B$ such that every inextendible timelike past/future curve through $p$ passes through $B$.
This construction determines a functor $D_\Mcal: \cball(\Sigma) \to \Caus(\Mcal)$, 
and hence we obtain the following construction.

\begin{df}
Let $\fA|_\Sigma: \cball (\Sigma) \to \Alg(\Nuc)$ denote the composite functor $\fA \circ D_\Mcal$.
\end{df}

Instead of using the diamonds, we could set 
\[
\fA|_\Sigma(B) = \lim_{\substack{\Ocal \supset B \\ \text{is Cauchy}}} \fA(\Ocal),
\]
taking the limit over opens $\Ocal \in \Caus(\Mcal)$ for which $B$ is a Cauchy hypersurface.
The time-slice axiom then ensures that this limit construction agrees with the diamond construction above. 
It was shown in \cite{Chilian} that this algebra is naturally isomorphic to the algebra obtained by quantizing the Cauchy data.

Likewise, we provide a version of $\Acal|_\Sigma$.
One could use a limit construction, but we prefer the concrete approach. 

\begin{df}
Let $\Acal|_\Sigma : \cball(\Sigma) \to \Ch(\Nuch)$ denote the functor that assigns to $B$, 
the BD algebra~$\Acal(D_\Mcal(B))$.
\end{df}

We thus obtain a nice comparison statement.

\begin{thm}\label{Comparison:alg}
Let $\Sigma$ be a Cauchy hypersurface of $\Mcal$.
The functor $H^0(\Acal|_\Sigma)$ can be lifted to an algebra object $H^0(\Acal|_\Sigma)^{Alg}$ in the category $\FA(\Sigma, \Nuch)$.
Moreover, the functor $H^0( \iota^q|_\Sigma)$ of Theorem \ref{thm: quantum cf} lifts to a natural isomorphism
\[
H^0 (\iota^q|_\Sigma)^{Alg}: H^0(\Acal|_\Sigma)^{Alg} \overset{\cong}{\Rightarrow} \fA|_\Sigma
\]
of functors to algebras.
\end{thm}

\begin{rem}
Note that we use the superscript $Alg$ to indicate that we have a new functor that factors through algebras.
In particular, we have $\fv \circ H^0(\Acal|_\Sigma)^{Alg} = H^0(\Acal|_\Sigma)$.
\end{rem}

We prove this result in Section \ref{Proofs}, after we spell out the explicit constructions of our models.
The argument does something more refined: we show that the factorization product agrees with the star product up to exact terms.
In other words, we implicitly lift $\Acal|_\Sigma$ to a homotopy associative algebra object in $\FA(\Sigma, \Ch(\Nuch))$.
We refrain, however, from spelling out a full homotopy-coherent algebra structure (e.g., $A_\infty$ structure).

There is an obvious classical analogue to this result.
The associative algebra $H^0(\Acal|_\Sigma)^{Alg}$ is naturally filtered by powers of $\hbar$, 
and its associated graded algebra is isomorphic to the commutative algebra $H^0(\Pcal|_\Sigma)[[\hbar]]$.
Hence, the commutative algebra $\fc \circ H^0(\Pcal|_\Sigma)$ acquires an \emph{unshifted} Poisson bracket, 
by taking the $\hbar$-component of the commutator of the associative algebra.

\begin{cor}
There is a functor 
\[
H^0(\Pcal|_\Sigma)^{Pois}: \cball(\Sigma) \to \PAlg(\Nuc)
\]  
by using the Poisson bracket induced on $\fc \circ H^0(\Pcal|_\Sigma)$ since it is the associated graded of  $H^0(\Acal|_\Sigma)^{Alg}$.
The functor $H^0( \iota^{cl}|_\Sigma)$ of Theorem \ref{thm: classical cf} lifts to a natural isomorphism
\[
H^0 (\iota^{cl}|_\Sigma)^{Pois}: H^0({\Pcal}|_\Sigma)^{Pois} \overset{\cong}{\Rightarrow} \Pfrak|_\Sigma
\]
of functors to Poisson algebras.
\end{cor}

A version of this statement at the cochain-level, for $\Pcal$, would also be appealing.
We now turn to explaining a version that relies on homotopical algebra,
but in Section \ref{subsec: shifting} we use formulas to explain how the Peierls bracket follows from the BV bracket.

\subsubsection{The argument via higher abstract nonsense}

We wish to explain why $\Pcal$ and $\Acal$, when restricted to a Cauchy hypersurface, obtain Poisson and associative structures, respectively.
\emph{A priori} they have a shifted Poisson and BD structure.
How could this transmutation of algebraic structure occur?

The key is a pair of interesting results from higher algebra that will relate certain factorization algebras to associative and Poisson algebras. 
We state the results before extracting the consequence relevant to us.

Let $E_1$ denote the operad of little intervals.
Concretely, an $E_1$ algebra $A \in \Alg_{E_1}(\Ch)$ is a homotopy-associative algebra;
in particular, every $E_1$ algebra is weakly equivalent to some dg algebra.
The first result, due to Lurie \cite{LurieHA}, says that there is an equivalence of $\infty$-categories
\[
\Alg_{E_1}(\Ch) \simeq \FA^{lc}(\RR,\Ch),
\] 
where the superscript $lc$ means we restrict to \emph{locally constant} factorization algebras:
a factorization algebra $\Fcal$ on $\RR$ is locally constant if the map $\Fcal(I) \to \Fcal(I')$ is a quasi-isomorphism for every pair of nested intervals $I \hookrightarrow I'$.
Lurie's result says that a locally constant factorization algebra on $\RR$ encodes a homotopy-associative algebra and \emph{vice versa}.

The second result explains how to relate different kinds of shifted Poisson algebras.
Let $P_n$ denote the operad encoding $(1-n)$-shifted Poisson brackets,
so that $P_1$ algebras are the usual Poisson algebras (in a homotopy-coherent sense).

\begin{thm}[Poisson additivity, \cite{Saf}]
There is an equivalence of $\infty$-categories
\[
\Alg_{E_1}(\Alg_{Pois_n}(\Ch)) \simeq \Alg_{Pois_{n+1}}(\Ch).
\]
\end{thm}

For $n = 0$, these results combine to say that a locally constant factorization algebra with a 1-shifted Poisson structure determines a homotopy-coherent version of an 0-shifted Poisson algebra.
Now consider the map $q: \Mcal \to \RR$ by taking the leaf space with respect to the foliation by Cauchy surfaces.
The pushforward factorization algebra $q_* \Pcal$ has a 1-shifted Poisson structure but it is also locally constant, 
since the solutions to the equation of motion is a locally constant sheaf in terms of the ``time'' parameter $\RR$.
Hence, by general principles, we know that $q_* \Pcal$ determines a 0-shifted Poisson algebra.

In this case, the homotopy-Poisson algebra must be strict at the level of cohomology, 
since the cohomology $H^*\Pcal$ is concentrated in degree 0.
This strict Poisson structure agrees with the Poisson structure on $\Pfrak$,
as we will see.

At the quantum level, things are analogous but simpler.
The pushforward factorization algebra $q_* \Acal$ is also locally constant
and hence determines a homotopy-associative algebra.
Since the cohomology $H^*\Acal$ is concentrated in degree 0,
it equips $H^0 \Acal$ with a strict associative structure.
One can see it agrees with canonical quantization
by a modest modification of arguments from Section 4.4 of \cite{CoGw}.
Thus, it agrees the associative structure on $\fA$.
Hence, by keeping track of the $\hbar$-filtration, 
we deduce that we obtain a correspondence between the Poisson algebra structures.

Our proofs of the comparison theorems take a different tack.
Following Section 4.6 of \cite{CoGw}, 
we exhibit natural Poisson and associative algebra structures by explicit formulas involving the propagators.
These match on the nose with the time-ordered product,
which gives us a direct relation with the star product of $\fA$.
Hence, in the quantum case, we see directly that these agree with the associative algebra structures coming from the abstract machinery described above.
By keeping track of the $\hbar$-filtration, we deduce that we obtain a correspondence between the Poisson algebra structures.

\begin{rem}
At the core of these identifications is a relationship between the standard deformation quantization of symplectic vector spaces and the standard BV quantization of free theories,
which we exhibited here via explicit formulas.
Work-in-progress of the first author with Rune Haugseng suggests a general explanation via higher abstract nonsense.
In \cite{GH}, they constructed a functor of linear BV quantization on dg vector spaces with a 1-shifted, linear Poisson bracket.
Loosely speaking, one finds that additivity intertwines this linear BV quantization with the usual Weyl quantization of ordinary Poisson vector spaces:
namely, taking $E_1$ algebras on the domain and codomain of linear BV quantizations yields the dg version of standard deformation quantization.
\end{rem}

\section{Constructing the CG model for the free scalar field}\label{section:CG:construction}

After all that formalism, we turn in a concrete direction 
and sketch the construction of free field theories in the CG formalism. 
We give a brief treatment here as this example is treated at length in Sections 4.2 and 4.3 of~\cite{CoGw} for the case of a Riemannian manifolds.
As we shall see, the constructions apply verbatim to Lorentzian manifolds.

Let $\Mcal = (M,g)$ denote a Lorentzian manifold.
Lazily, we write $\d x$ for the associated volume form on $M$.
We will consider the case $(\RR, dx)$ as a running example.

\subsection{The classical model}

To start, consider the classical theory.
The equation of motion is $P\phi = 0$.
The running example is the free scalar field, with $\Box \phi + m^2 \phi = 0$ and $\phi$ a smooth function on $M$.
The space of distributional solutions $V\subset \Dcal'(M)$ consists of ``waves'', 
and let $V^*$ denote the continuous linear dual.
The natural algebra of observables---of a purely algebraic flavor---is $\Sym_{alg}(V^*)$,
the polynomial functions on $V$.
(Such functions should be contained in more sophisticated choices of observables, and indeed are often a dense subalgebra.)
In the BV framework, one replaces this commutative algebra by a commutative \emph{dg} algebra that resolves it and that also remembers the larger space of fields.

\begin{exa}
For the free scalar field on the real line, the space of solutions is a two-dimensional vector space $V$ spanned by $\{e^{\pm imx}\}$.
Here $\Sym(V^*) \cong \CC[p,q]$, a polynomial algebra with two generators.
These generators can be identified with ``position'' and ``momentum'' at $x = 0$,
since the value of a function and its derivative at one point determine a solution of the equation.
\end{exa}

In constructing this resolution, 
one eventually has to make some choices about functional analysis.
We will begin by avoiding any analysis and construct a purely algebraic version, in order to exhibit the structure of the BV approach,
but then we will turn to a functional-analytic completion convenient for free theories.
(See Section 3.5 or Appendix B of \cite{CoGw} for a seemingly interminable discussion of such functional analysis issues.)

For free theories, it is sufficient and convenient to work with smeared or smoothed observables.
Thus, for instance, each {\em linear} observable $\Orm_f$ is specified by a compactly supported smooth function $f \in \Dcal(M)$, where 
\[
\Orm_f(\phi) = \int_M \phi(x) f(x)  \d x.
\]
In other words, we will let $\Dcal(M)$ provide the linear observables, 
rather than the larger space of compactly supported distributions.
These then generate a commutative algebra of ``polynomial functions on the scalar fields'':
\[
\Sym_{alg}(\Dcal(M)) = \bigoplus_{n \geq 0} (\Dcal(M)^{\otimes_{alg} n})_{S_n}.
\]
Note here that $\otimes_{alg}$ simply means the algebraic tensor product;
we will introduce a convenient completion soon.

There is manifestly a surjection $\Sym_{alg}(\Dcal(M)) \to \Sym_{alg}(V^*)$ 
by restricting a function on all fields to a function on fields that satisfy the equation of motion.
We now extend this surjection to a resolution $\widetilde{PV}$.
(It might help some readers to know that we are going to write down the Koszul resolution for a linear equation, which in this case are the equations of motion.)

Some notation is helpful here: we use $V[1]$ to denote an ungraded vector space $V$ placed in cohomological degree -1 (i.e., we shift down by one), and
when we write $\Sym^n(V[1])$, we use the Koszul rule of signs, so that this vector space is naturally isomorphic to $(\Lambda^n V)[n]$.
Thus, we can write succinctly
\[
\widetilde{PV} = \Sym_{alg}(\Dcal(M) \oplus \Dcal(M)[1]),
\]
so that for $-k \leq 0$,
\[
\widetilde{PV}^{-k} \cong \Sym_{alg}(\Dcal(M)) \otimes_{alg} \Lambda_{alg}^{k} (\Dcal(M)),
\]
and $\widetilde{PV}^k = 0$ for $k > 0$.
This graded vector space is a version of ``polynomial polyvector fields on the space of scalar fields.''
By construction, $\widetilde{PV}$ is a graded commutative algebra.

We now describe its differential $\d$,
which encodes the equations of motion.
Given an element $f_1\cdots f_n \otimes g_1 \wedge \cdots \wedge g_m$ of 
$\Sym_{alg}^n(\Dcal(M)) \otimes_{alg} \Lambda_{alg}^m (\Dcal(M))$,
which has cohomological degree $-m$, 
we define
\[
\d( f_1\cdots f_n \otimes g_1 \wedge \cdots \wedge g_m ) = \sum_{i = 1}^m (-1)^{i-1} (P g_i) f_1\cdots f_n \otimes g_1 \wedge \cdots \widehat{g_i} \cdots \wedge g_m,
\]
where $\widehat{g_i}$ indicates that this term is removed from the wedge product. 
One can check that this differential $\d$ is a derivation,
so that we have constructed a commutative dg algebra.

\begin{rem}\label{rem:koszul}
It is an illuminating exercise to show that $(\widetilde{PV},\d)$ provides a cochain complex resolving the polynomial functions $\Sym(V^*)$ on the space of solutions~$V$. 
(It helps to bear in mind that we have written down a Koszul resolution for a linear equation,
albeit on an infinite-dimensional vector space.)
This resolution has the special property that polynomial functions on all scalar fields is given by the truncation consisting of the degree 0 component.
Hence, the commutative dg algebra also remembers, in this way, the ambient space of scalar fields.
\end{rem}

\begin{exa}
For the free scalar field on $M = \RR$, we know that we have a quasi-isomorphism
\[
(\Dcal(\RR)[1] \xto{\partial_x^2 + m^2} \Dcal(\RR)) \xto{\simeq} V^* \cong \CC^2
\]
sending a linear observable of degree 0 to its value on solutions to the equation of motion.
(This map is dual to the inclusion of solutions into all fields.)
Hence, taking the symmetric algebra on either side of the quasi-isomorphism, we again have a quasi-isomorphism.
The left hand side is precisely~$(\widetilde{PV},\d)$.
\end{exa}

Observe that these definitions make sense for any open subset $U \subset M$.
Thus, there is a kind of ``model of the classical fields'' given by the functor
\[
\begin{tikzcd}
\Fcal_{fields}: & \Open(M) \arrow[r] & \Ch(\Vec) \\
& U \arrow[|->,r] & \Sym_{alg}(\Dcal(U))
\end{tikzcd}
\]
which simply assigns to $U$, the polynomial functions on scalar fields on $U$.
The functoriality with respect to the open $U$ is simple: compactly supported functions extend by zero, and we apply this map to the symmetric algebra as well.
Likewise, there is a kind of ``model for the classical free theory'' given by the functor
\[
\begin{tikzcd}
\Fcal_{theory}: & \Open(M) \arrow[r] & \Ch(\Vec) \\
& U \arrow[|->,r] & (\widetilde{PV}(U),\d)
\end{tikzcd}
\]
which assigns to $U$, a commutative dg algebra resolving the polynomial functions on solutions on $U$ of the equation of motion.

Here is one way to ``complete'' these algebras and make them better behaved in a topological sense.
The key idea is simple: any  compactly supported smooth function $f \in \Dcal_n(U)$ determines an observable that is homogeneous of degree $n$ by the formula
\[
\Orm_f(\phi) = \int_{M^n} \phi(x_1) \cdots \phi(x_n) f(x_1,\ldots,x_n) \d x_1 \ldots \d x_n.
\]
Indeed, there is a dense inclusion
\[
\Sym_{alg}^n(\Dcal(U)) = (\Dcal(U)^{\otimes_{alg} n})_{S_n} \hookrightarrow \Dcal_n(U)_{S_n}.
\]
Note that we quotient out the action of permuting the coordinates because a function $f(x_1,\ldots,x_n)$ and its permutation $f(x_{\sigma(1)},\ldots,x_{\sigma(n)})$ define the same observable.
Similarly, there is a dense inclusion
\[
\Lambda^n_{alg}(\Dcal(U)) = (\Dcal(U)^{\otimes_{alg} n})/\{\text{sign action of }S_n\} \hookrightarrow \Dcal_n(U)/\{\text{sign action of }S_n\}.
\]
Hence, we replace $\widetilde{PV}$ by its completion along these lines:
for $-k \leq 0$, set
\[
PV^{-k}(U) = \bigoplus_{n\geq 0} \Dcal_{n+k}(U)_{S_n \times S_k}
\]
where the symmetric group $S_n$ acts on the first $n$ coordinates as before, but $S_k$ acts on functions depending on the last $k$ coordinates via the \emph{sign} representation,
and $PV^k(U) = 0$ for $k > 0$.
The multiplication on $\widetilde{PV}(U)$ extends naturally to $PV(U)$.
Concretely, one notes that given $f \in \Dcal_m(U)$ and $g \in \Dcal_n(U)$, 
there is a function $f \boxtimes g \in  \Dcal_{m+n}(U)$ with
\[
f \boxtimes g(x_1,\ldots,x_m,y_1,\ldots,y_n) = f(x_1,\ldots,x_m) g(y_1,\ldots,y_n).
\]
This extension is manifestly continuous.

The differential $\d$ on $\widetilde{PV}(U)$ also extends naturally to this completion.
Alternatively, one can note that there is a continuous map
\[
P_{y_1}: \Dcal_{n+k}(U) \to \Dcal_{n+k}(U),
\]
where $(x_1,\ldots,x_n,y_1,\ldots,y_k)$ give coordinates for $U^n \times U^k$,
and this map descends to the quotient by the action of $S_n \times S_k$.

This completion encodes a flavor of polynomial functions on distributional solutions to the equation of motion.

\begin{thm}
\label{painintheass}
Let $U \in \Caus(\Mcal)$ be an open subset of $\Mcal$ that is globally hyperbolic with respect to the induced Lorentzian metric.
For $k \neq 0$, the cohomology vanishes: $H^k(PV(U),\d) = 0$.
For $k=0$, there is a dense inclusion of commutative algebras
\[
\Sym_{alg}(Sol(U)') \to H^0(PV(U),\d)
\]
where
\[
Sol(U) = \{ \phi \in \Dcal'(U) \, : \, P\phi = 0 \}.
\]
\end{thm}

The proof involves a long detour into functional analysis, 
so we banish it to the appendix,
where we introduce some arcane terminology that leads to a sharper version of the theorem as well as a proof.
(Notably we improve the dense inclusion to an isomorphism, but using a symmetric algebra built by completing the bornological tensor product.)

\begin{exa}
In particular, for the free scalar field on $M = \RR$, we have $H^*(PV(\RR),\d) \cong \CC[p,q]$, where $p,q$ are two variables.
\end{exa}

Finally, $PV(U)$ can be equipped with a 1-shifted Poisson bracket. This bracket is straightforward to define.
Consider the natural bilinear pairing of cohomological degree~1,
\[
\{-,-\}: (\Dcal(U) \oplus \Dcal(U)[1])^{\widehat{\otimes} 2} \to \CC,
\]
where
\[
\{f,g\} = \begin{cases} \int_U f(x)g(x) \d x & \text{if } |f| \neq |g| \\ 0 & \text{else} \end{cases}.
\]
It is skew-symmetric in the graded sense, 
and hence we obtain a shifted Poisson bracket on $\widetilde{PV}(U)$ by extending this pairing from linear functions to polynomials via Leibniz's rule.
This construction extends continuously to $PV(U)$.
(It amounts to integrating out along diagonals.)

\begin{rem}
	Note that this bilinear pairing is ill-defined if one replaces compactly supported smooth functions by distributions.
	This issue is a key problem in setting up the BV formalism and begets many of the divergences in perturbation theory.
\end{rem}

To summarize, we give the following definition.

\begin{df}
The \textbf{classical model for the free theory with Green-hyperbolic operator $P$} is the prefactorization algebra $\Pcal$ on $M$ taking values in $\Ch(\Nuc)$ assigning to the open set $U$, the commutative dg algebra
\[
\Pcal(U) = (PV(U),\d),
\] 
equipped with a 1-shifted Poisson bracket~$\{-.-\}$.
\end{df}

It is simply a completion of the functor $\Fcal_{theory}$ defined earlier.
We believe it is actually a factorization algebra, 
but verifying this belief requires understanding homotopy colimits in $\Ch(\Nuc)$,
in particular whether the usual formulas for homotopy colimits of cochain complexes hold in this setting
(see Appendix C.5 of~\cite{CoGw}).

\begin{rem}
Two variations on this approach are needed when dealing with interacting theories.
First, one replaces polynomial functions by formal power series, i.e., $\oplus_{n \geq 0}$ becomes $\prod_{n \geq 0}$.
Second, one cannot restrict to smoothed observables but should allow distributional observables, i.e., $\Dcal$ is replaced by $\Ecal'$ (the space of compactly supported distributions).
In the setting of elliptic differential operators (or elliptic complexes, more generally),
the commutative dg algebras with smoothed or distributional algebras are (continuously) quasi-isomorphic.
Moreover, the differential is still determined by the equations of motion but is more complicated as it has terms changing the homogeneity of observables.
In particular, the smoothed and distributional algebras cease to be quasi-isomorphic in the interacting case.
\end{rem}

\subsection{The quantum model}

We now turn to BV quantization, which modifies the differential by adding the BV Laplacian.
This extra term is related to a shifted Poisson structure on $PV(U)$.

We now define the BV Laplacian $\triangle$ similarly.
We require it to satisfy the equation
\[
\triangle(a \cdot b) = \triangle(a) \cdot b + (-1)^a a \cdot \triangle(b) + \{a,b\}
\]
for any $a,b \in \widetilde{PV}(U)$.
Hence, once we assert that $\triangle$ annihilates any constant or linear terms,
we determine $\triangle$ iteratively.
For instance,
for a quadratic term $fg \in \Sym^2(\Dcal(U) \oplus \Dcal(U)[1])$,
we see
\[
\triangle(fg) = \{f,g\}.
\]
We then extend $\triangle$ to $PV(U)$ in the natural, continuous way.
For instance, given a quadratic term in $PV^{-1}(U)$, namely some
\[
F \in \Dcal(U\times U),
\]
we see
\[
\triangle(F) = \int_{x \in U} F(x,x) \d x.
\]
As with the bracket, the BV Laplacian amounts to integrating along diagonals.

\begin{df}
The \textbf{quantum model for the free theory with Green-hyperbolic operator $P$} is the prefactorization algebra $\Acal$ on $\RR$ taking values in $\Ch(\Nuch)$ that assigns  the BD algebra
\[
\Acal(U) = (PV(U)[[\hbar]],\d-i\hbar \triangle,\{-,-\})
\]
to the open set~$U$.
\end{df}

We now examine the following useful result.

\begin{lemma}
For each open $U$ in $M$, there is a natural isomorphism of graded vector spaces
\begin{align*}
H^* \Acal(U) 
&= H^*(PV(U)[[\hbar]],\d - i\hbar \triangle) \\
&\cong H^*(PV(U))[[\hbar]] \\
&= H^* \Pcal(U)[[\hbar]],
\end{align*}
and all nonzero cohomology groups vanish.
\end{lemma}

In other words, the cohomology of the classical and quantum observables agree up to adjoining $\hbar$ to the classical ones.
Note that this isomorphism does \emph{not} respect the commutative algebra structure on the cohomology of the classical observables.
Indeed, the differential of a BD algebra is not a derivation with respect to the commutative product, and hence the commutative product does not descend to the cohomology.

\begin{proof}
The filtration by powers of $\hbar$ determines a spectral sequence that computes the cohomology of the quantum observables.
The first page is just the cohomology of the classical observables. 
Since that is totally concentrated in degree 0, the spectral sequence collapses.
\end{proof}

\section{Constructing the pAQFT model for a free field theory}

In this section we describe the pAQFT construction for the classical and quantum models for a free  field theory and prove Proposition~\ref{FRconstruction}.
It is a succinct review of a more extensive treatment available~\cite{LesHouches,Book}.

The construction itself explicitly produces dg algebras;
to recover algebras, one takes the cohomology, which happens to be concentrated in degree zero.
Thus, before going into the details, we proffer a dg version of AQFT, as defined in section~\ref{sec:dgpaqft}.

\subsection{Constructing the dg models}

In this section we spell out the construction of a semistrict dg model of a free field theory. This is mainly a review of \cite{FR3}, but with more detail and recast in notation compatible with the CG framework.
The goal is to provide a kind of Koszul resolution of the algebra of functions on the space of solutions to the equations of motion.
We need to pin down some functional analytic choices, along with the homological algebra,
before we articulate the central construction, given in Definition~\ref{df:regclassobs}.

\subsubsection{Functionals}

Regular functionals on the configuration space $\Ecal$ were defined in Definition~\ref{df: regular}. We will use these to model classical observables.

Functionals that are both regular and linear are given as pairings with smooth compactly supported densities, i.e., are of the form
\[
\Orm_f(\phi) = \int_M \phi(x) f(x)\,,
\]
where $f\in\Dcal(M)$.

\begin{df}\label{def:tau}
Let $\tau$ be the locally convex topology on $\F_{\reg}$ generated by the following family of seminorms:
\[
q_{\phi,n,p}(F)\doteq p(F^{(n)}(\phi))
\]
where $B\subset \Ecal$ is bounded and $p$ runs over all the seminorms generating the locally convex topology of $\Dcal_n^{\sst \CC}(M)$.
\end{df}

We will always consider $\F_{\reg}$ together with this topology. It was shown in \cite{FR}[Appendix A] that this topology is nuclear. The idea behind the proof (following \cite{BDF}) is to use the fact that all the spaces $\Dcal_n^{\sst \CC}(M)$, $n\in \NN$ are nuclear and $\tau$ is the initial topology with respect to the evaluation maps $F\mapsto F^{(n)}(\ph)\in \Dcal_n^{\sst \CC}(M)$, $\ph\in\Ecal(M)$, $n\in\NN$. As nuclearity is preserved under projective limits, the result follows.

\subsubsection{Polyvector fields}\label{sec:polyvect}

The basic input for our field theory is $dS$, a 1-form on $\E$ that gives the equations of motion
\be\label{class}
dS(\ph)=0\,.
\ee
For free fields we have $dS(\ph)=P\ph$, so that the equations of motion are linear. 
In particular, for the free scalar field:
\[
P=-(\Box+m^2)\,.
\]
The operator $P$ extends to $\Dcal'(M)\supset \Ecal(M)$ and, as in the CG framework, we are interested in the space $V\subset \Dcal'(M)$ of distributional solutions.

Our goal is to construct a cohomological resolution of $\Sym(V')$,
and we will produce a Koszul-type resolution.
Since $dS$ is a kind of 1-form, 
this resolution is built using the algebra of regular polyvector fields on~$\Ecal(M)$, 
with the differential determined by the equation of motion.
Our focus at the moment is on the algebra; 
we postpone discussion of the differential until the paragraphs around equation~\eqref{K0}.

For a finite-dimensional manifold $M$, such a Koszul resolution can be understood as arising from the shifted cotangent manifold~$T^*[-1]M$. 
When $M$ is a vector space $V$, $T^*[-1]M$ corresponds to the graded vector space $V \oplus V^*[-1]$.
Its ring of functions then looks like a graded-symmetric algebra $\Sym(V^* \oplus V[1])$,
whose degree $-n$ component is then $\Sym(V^*) \otimes \Lambda^n(V)$.
In our setting we thus work with the following.

\begin{df}
	Let $X_1$, $X_2$ be in $\Nuc$. The \textbf{space of functions on $X_1\oplus X_2[-1]$} is a graded vector space with degree $-n$ component given by
	$$\Oscr^{-n}(X_1\oplus X_2[-1])=\Ci\Big(X_1, (X_2^{\widehat{\otimes} n})_{S_n}\Big)\,,$$
	where $\Ci$ means ``Bastiani smooth'', as explained at the beginning of this subsection and the sign action of $S_n$ means, effectively anti-symmetrization of the tensor product. Denote
	\[
	\Oscr(X_1\oplus X_2[-1])\doteq\prod_n\Ci\Big(X_1, (X_2^{\widehat{\otimes} n})_{S_n}\Big)[n]\,.
	\]
\end{df}

This definition is not the only approach, because in infinite-dimensional differential geometry, the precise definition of $T^*[1]\Ecal(M)$ depends on what differentiable manifold structure we put on $\Ecal(M)$. 
(Recall the discussion at the beginning of Section~\ref{sec:exp:prod}.) 
One option is to use the locally convex Fr{\'e}chet topology. Another option is to define open neighborhoods in $\Ecal(M)$ as 
\[
U_{\phi,V}\doteq\{\phi+\psi,\psi\in V\}\,,
\]
where $V$ is an open neighborhood in $\Dcal$, equipped with its natural topology. With this choice of topology $T\Ecal=\Ecal\times \Dcal$. 
Physically, this choice means that we allow for variations of field configurations only in the direction of compactly supported configurations. 
Since we want to allow for distributional solution to the equations of motion, however, 
it is useful to enlarge the tangent bundle to the distributional completion $\Ecal\times \Ecal'$. 
Consequently, the cotangent bundle gets restricted to $\Ecal\times \Ecal$. Hence the corresponding restricted odd cotangent bundle is $\underline{T}^*[1]\Ecal=\Ecal \times \Ecal[1]$, which we can view as a cochain complex in concentrated in degrees zero and one.
Hence our focus on field theory guides our choice.

\begin{rem}
	Polyvector fields are elements of $\Oscr(\Ecal(M) \oplus\Ecal(M)[-1])$. In order to define regular polyvector fields, we need to analyze the WF sets of derivatives of $F\in\Oscr(\Ecal(M) \oplus\Ecal(M)[-1])$.
\end{rem}
	
\begin{rem}
	Consider the special case $X_1=\Gamma(M,E_1)$, $X_2=\Gamma(M,E_2)$, where $E_1$, $E_2$ are vector bundles over $M$. It was shown in \cite{Book,BDGR} that $X_2^{\widehat{\otimes} n}\cong \Gamma'(M^n,E_2^{\boxtimes n})_{S_n}$ and that the $k$th functional derivative of $F\in \Oscr^n(E_1\oplus E_2[1])$ at a given point in $X_1$, is an element of $\Gamma'(M^{k+n},E_1^{\boxtimes k}\boxtimes E_2^{\boxtimes n})_{S_k\times S_n}$, symmetric in the first $k$ and antisymmetric in the last $n$ entries.
\end{rem}

\begin{df}
	Let $F\in\Oscr^n(\Ecal(M)\oplus\Ecal(M)[-1])$ be a polyvector field. 
	We say $F$ is \textbf{regular} if $F^{(k)}(\phi)$ has empty WF set (i.e. is smooth). We use $\mathfrak{PV}_{\reg}(\Ocal)$ to denote the space of all regular polyvector fields on $\Ecal(\Ocal)$ where $\Ocal\subset\Mcal$. 
\end{df}
In particular, among all regular  polyvector fields we can distinguish the polynomial ones (analogous to Definition~\ref{df: poly}), which we denote by $\mathfrak{PV}_{\pol}(\Ocal)$.

This construction gives a functor $\mathfrak{PV}_{\reg}$ from $\Caus(\Mcal)$ to $\CAlg(\Ch(\Nuc))$, where the action on morphisms is induced by the pullback.
(We currently have the zero differential on the polyvector fields, but we will introduce a differential depending on $dS$ below.)

Clearly, $\mathfrak{PV}^0_{\reg}=\F_{\reg}$ and $\mathfrak{PV}_\reg$ is a graded commutative algebra by the usual product on functions and the wedge product of polyvector fields.

\begin{rem}
	The topology $\tau$ from Definition \ref{def:tau} has a natural generalization to the locally convex topology on $\Ocal(X_1\oplus X_2[-1])$. We use the following family of seminorms:
	\[
	q_{\phi,n,p,B}(F)\doteq \sup_{\phi \in B \subset X_1} (p(F^{(n)}(\phi)))\,,
	\]
	where we run over bounded subsets $B\subset X_1$ and we run over the seminorms $p$ that generate the locally convex topology of~$\prod_n (X_2^{\widehat{\otimes} n})_{S_n}$. Again, we refer to \cite{FR} for the proof of nuclearity.
\end{rem}

\subsubsection{Poisson structure}

It is crucial that $P$ is a normally hyperbolic operator, so on a globally hyperbolic spacetime it has retarded/advanced Green's functions $G^{\mathrm{R}}$/$G^{\mathrm{A}}$, respectively and other propagators introduced in section~\ref{Propagators}.

Using the ideas of Peierls \cite{Pei} we introduce a Poisson bracket $\Poi{.}{.}_\Ocal$ on $\F_{\reg}(\Ocal)$ by
\be\label{Peierls}
\Poi{F_1}{F_2}(\ph)\doteq \left<F_1^{(1)}(\ph),G_\Ocal^{\mathrm{C}} \left(F_2^{(1)}(\ph)\right)\right>\,,
\ee
where $G_\Ocal^{\mathrm{C}}$ is the causal propagator constructed on $\Ocal\subset \Mcal$, treated as globally hyperbolic spacetime in its own right. Note that from the uniqueness of retarded and advanced Green functions follows that for the morphism $\psi:\Ocal_1\rightarrow \Ocal_2$,
\be\label{natural}
\chi_{\psi(\Ocal_1)}\, G_{\Ocal_2}^{\mathrm{C}}\Big|_{\Dcal(\psi(\Ocal_1))}=G_{\psi(\Ocal_1)}^{\mathrm{C}}\,,
\ee
where $\chi_{\psi(\Ocal_1)}$ is the characteristic function of $\Ocal_1$ and $G_{\Ocal_2}^{\mathrm{C}}$ is treated as a map $\Dcal(\Ocal_2)\rightarrow \Ecal(\Ocal_2)$.

We now extend to the dg setting. 
Note that the space of on-shell regular functionals is the zeroth cohomology of the two-term complex
\be\label{K0}
\begin{array}{c@{\hspace{0,2cm}}c@{\hspace{0,2cm}}c@{\hspace{0,2cm}}c@{\hspace{0,2cm}}c@{\hspace{0,2cm}}c@{\hspace{0,2cm}}c@{\hspace{0,2cm}}c@{\hspace{0,2cm}}c}
	&0&\rightarrow&\mathfrak{PV}^1_\reg&\xrightarrow{\delta_{S}}&\F_{\reg}&\rightarrow& 0
\end{array}\ ,
\ee
where
\[
\delta_{S}X\doteq \iota_{dS}X\,.
\]
The linear map $\delta_{S}$ easily extends to a differential on $\mathfrak{PV}_\reg$ by imposing the (graded) Leibniz rule; 
we can also extend the bracket trivially. Note that $\delta_S$ is compatible with the bracket, since $G^{\rm C}$ is a bi-solution for $P$.

Hence we can lift our notion to the cochain level. We assign   dg Poisson algebras to $\Ocal\in\Caus(\Mcal)$ by keeping track of support by means of (\ref{support}). 

\begin{df}
	The \textbf{dg Poisson algebra of regular classical observables} is
	\[
	\Pfrak_\reg(\Ocal)=(\mathfrak{PV}_\reg(\Ocal),\Poi{.}{.}_\Ocal,{\delta_{S}}_\Ocal)\,.
	\]
	Restricting to polynomial vector fields, we define $$\Pfrak_\pol(\Ocal)=(\mathfrak{PV}_\pol(\Ocal),\Poi{.}{.}_\Ocal,{\delta_{S}}_\Ocal)$$ in completely analogous way.
\end{df}
The following proposition shows that this indeed gives us a functor from $\Caus(\Mcal)$ to $\PAlg^*(\Ch(\Nuc))$.
\begin{prop}
	There is a functor from $\Caus(\Mcal)$ to $\PAlg^*(\Ch(\Nuc))$ that assigns to each $\Ocal\subset\Mcal$, the dg Poisson algebra $\mathfrak{P}_{\reg}(\Ocal)\doteq(\PV_{\reg}(\Ocal),\Poi{.}{.}_\Ocal,{\delta_S}_\Ocal)$,  
	and that assigns to each morphism $\psi:\Ocal_1\rightarrow \Ocal_2$, the Poisson map
	\[
	(\mathfrak{P}_{\reg}\psi F)(\ph)\doteq F(\psi^*\ph)\,,
	\]
	where $\ph\in\Ecal(\Ocal_2)$ and $F\in \PV_{\reg}(\Ocal_1)$.
	This functor $\mathfrak{P}_{\reg}$ is a classical field theory model in the sense of Definition~\ref{dgClassFT}.
\end{prop}

\begin{proof}
	By construction  $(\PV_{\reg}(\Ocal),\Poi{.}{.}_\Ocal,{\delta_S}_\Ocal)$ is a dg Poisson algebra, so it remains to check that $\PV_\reg$ has the correct behavior on the morphisms. Let $\psi:\Ocal_1\rightarrow \Ocal_2$. We have	
	\begin{multline*}
	\Poi{\Pfrak_\reg\psi F_1}{\Pfrak_\reg\psi F_2}_{\Ocal_2}(\ph)=\left<\psi(F_1^{(1)}(\psi^*\ph)),\chi_{\psi(\Ocal_1)} G_{\Ocal_2}^{\mathrm{C}}\circ \psi \left(F_2^{(1)}(\psi^*\ph)\right)\right>\\=\left<\psi (F_1^{(1)}(\psi^*\ph)),G_{\psi(\Ocal_1)}^{\mathrm{C}} \psi\left(F_2^{(1)}(\psi^*\ph)\right)\right>\\
	=\left< F_1^{(1)},G_{\Ocal_1}^{\mathrm{C}} \left(F_2^{(1)}\right)\right>(\psi^*\ph)=(\Pfrak_\reg\psi\Poi{F}{G}_{\Ocal_1})(\ph)\,.
	\end{multline*}
Hence we see that each structure map $\Pfrak_\reg\psi$ is a map of dg Poisson algebras, as desired.
It remains to verify Einstein causality, but this property is immediate from formula (\ref{Peierls}):
the support properties of $G^{\mathrm C}$ ensure that spacelike-separated observables bracket to zero strictly, and not just up to exact terms.
\end{proof}

The analogous result for $\PV_\pol$ follows by the same arguments.

\begin{rem}
	Note that the statement about the existence and uniqueness of retarded and advanced Green functions (needed in the proof of the proposition) is true only on opens that are themselves globally hyperbolic spacetimes (when equipped with the induced metric). Therefore, it is crucial to restrict to $\Caus(\Mcal)$, rather than to consider arbitrary opens.
\end{rem}

If we forget the Poisson algebra structure, we obtain the cohomological (derived) description of the space of classical observables as the functor $\fv\circ\Pfrak_\reg\equiv \fv\Pfrak_\reg$ (or $\fv\Pfrak_\pol$ if we restrict to polynomials).
\begin{df}\label{df:regclassobs}
	The \textbf{space of regular classical observables} is
	\[
	\fv\Pfrak_\reg=(\fv\mathfrak{PV}_\reg,\delta_{S}).
	\]
\end{df}
Going on-shell corresponds to taking the $H^0$ of $\Pfrak_\reg$ (or $\fv\Pfrak_\pol$). We obtain the following:
\begin{df}
The on-shell \textbf{Poisson algebra of regular classical observables} is the quotient
\[
H^0(\mathfrak{P}_{\reg})(\Ocal)=\mathfrak{P}_\reg(\Ocal)/\mathfrak{P}_{\reg,0}(\Ocal)\,,
\]
where $\mathfrak{P}_{\reg,0}(\Ocal)$ is the Poisson ideal generated by the elements of the form
\[
\left<dS(\ph),X(\ph)\right>\equiv \iota_{dS}X\,,
\]
where $X\in\mathfrak{PV}^1_{\reg}$.  
\end{df}

\begin{prop}\label{Class:mod:proof}
The assignment $\Ocal\mapsto H^0(\mathfrak{P}_{\reg})(\Ocal)$ defines a classical on-shell model in the sense of Definition~\ref{OnShell}.
\end{prop}

This result is part of Proposition~\ref{FRconstruction}, namely the classical piece of the model.

\begin{proof}
It remains only to verify the time-slice axiom, which was done in \cite{Dim} and also in \cite{Chilian,ChF08}. Here, for completeness, we provide an argument. 

Let $\Ncal$ be a causally convex neighborhood of a Cauchy surface $\Sigma$ of $\Ocal$. 
There is a natural inclusion map $i: \mathfrak{P}_{\reg}(\Ncal) \to \mathfrak{P}_{\reg}(\Ocal)$, as part of the net structure of $\mathfrak{P}_{\reg}$. Similarly, the net structure of $H^0(\mathfrak{P}_{\reg})$ includes a natural extension map $H^0(i): H^0(\mathfrak{P}_{\reg})(\Ncal) \to H^0(\mathfrak{P}_{\reg})(\Ocal)$.
By definition, there is a map of exact sequences
\[
\begin{tikzcd}
\mathfrak{P}_{\reg,0}(\Ncal) \arrow[r] \arrow[d, "i^0"] & \mathfrak{P}_\reg(\Ncal) \arrow[r] \arrow[d, "i"] & H^0(\mathfrak{P}_\reg)(\Ncal) \arrow[d,"H^0(i)"] \\
\mathfrak{P}_{\reg,0}(\Ocal) \arrow[r] & \mathfrak{P}_\reg(\Ocal) \arrow[r] & H^0(\mathfrak{P}_\reg)(\Ocal)
\end{tikzcd},
\]
where $i^0$ denotes the restriction of $i$ to the Poisson ideals.

To verify the time-slice axiom, we need to produce an inverse map $\widetilde{\beta}$ to $H^0(i)$ that explicitly demonstrates that $H^0(i)$ is an isomorphism.
We will do this by explicitly constructing a map 
\[
\beta: \mathfrak{P}_{\reg}(\Ocal) \to \mathfrak{P}_{\reg}(\Ncal)
\] 
compatible with the Poisson ideals,
and hence descending to a $\widetilde{\beta}$ that will produce the inverse.
Note that it is sufficient to produce this map $\beta$ just on generators, i.e., on elements induced by the linear functionals $\Orm_f$ for $f\in\Dcal(\Ocal)$.

In addition, pick two other Cauchy surfaces $\Sigma_\pm$ in $\Ncal$, such that $\Sigma_-$ is in the past $J^-(\Sigma)$ of $\Sigma$ and $\Sigma^+$ is in the future $J^+(\Sigma)$ of~$\Sigma$.
	
Finally, pick a smooth function $\chi$ that is equal to 1 on $ J^-(\Sigma_-)$, and vanishes on $ J^+(\Sigma_+)$. 
We use it to construct  a partition of unity subordinate to the cover by $J^+(\Sigma_-)$ and $J^-(\Sigma_+)$.
This partition leads us to decompose $f$ as the linear combination $\chi f+(1-\chi)f$. The first term is supported in the past $J^-(\Sigma_+)$ of $\Sigma_+$, and the second term is supported in the future $J^+(\Sigma_-)$ of~$\Sigma_-$. 

We define the map $\beta$ from $\Dcal(\Ocal)$ to $\Dcal(\Ncal)$~as
\[
\beta(\Orm_f)=\beta_+(\Orm_{\chi f})+\beta_-(\Orm_{(1-\chi)f})\,,
\]
where $\beta_+$ is defined on observables supported in the past of $\Sigma_+$ and $\beta_-$ is defined on observables supported in the future of $\Sigma_-$. We will study these two maps separately.

Assume first that $\supp f$ is in the past of $\Sigma_+$. We define
\be\label{beta:plus}
\beta_+(\Orm_f)=\Orm_{f-P\chi G^{\rm R} f}\,.
\ee
Note that the test function $f-P\chi G^{\rm R} f$ is supported within $\Ncal$, so $\beta_+$ maps $\Dcal(J^-(\Sigma_+))$ to $\Dcal(\Ncal)$. 
Moreover, by construction,
\[
\Orm_{f-P\chi G^{\rm R} f}=\Orm_{f}-\Orm_{P\chi G^{\rm R} f}\,,
\]
and $\Orm_{P\chi G^{\rm R} f}\in\Pfrak_{\reg,0}$.
Hence $\beta_+$ induces a map on the quotient algebras 
\[
\widetilde{\beta}_{+}: H^0(\Pfrak_{\reg})(J^-(\Sigma_+)) \to H^0(\Pfrak_{\reg})(\Ncal),
\]
and we have just shown that, postcomposing with the extension map $H^0(\Pfrak_{\reg})(\Ncal) \to H^0(\Pfrak_{\reg})(J^-(\Sigma_+))$, we obtain the identity map on $H^0(\Pfrak_{\reg})(J^-(\Sigma_+))$.
The construction of $\beta_+$ is illustrated on Figure~\ref{fig:BetaPlus}.

\begin{figure}[h!]
	\begin{center}
\begin{tikzpicture}[scale=0.1,x=1.00mm, y=1.00mm, inner xsep=0pt, inner ysep=0pt, outer xsep=0pt, outer ysep=0pt]
\path[line width=0mm] (-241.59,-206.40) rectangle +(596.05,549.60);
\definecolor{L}{rgb}{0,0,0}
\path[line width=0.30mm, draw=L] (-236.79,-204.40) [rotate around={360:(-236.79,-204.40)}] rectangle +(446.41,545.11);
\definecolor{F}{rgb}{0.753,0.753,0.753}
\path[line width=0.30mm, draw=L, fill=F] (-238.30,50.59) rectangle +(447.02,172.20);
\path[line width=0.30mm, draw=L] (-237.74,133.65) -- (208.21,132.75);
\path[line width=0.30mm, draw=L] (-239.07,203.60) -- (208.63,203.34);
\path[line width=0.30mm, draw=L] (-239.59,68.39) -- (210.02,69.24);
\definecolor{F}{rgb}{0.502,0.502,0.502}
\path[line width=0.30mm, draw=L, fill=F] (-100.51,-53.24) .. controls (-95.28,-16.30) and (-45.47,-28.27) .. (-14.31,-14.37) .. controls (6.34,-5.15) and (29.34,10.26) .. (39.78,-7.61) .. controls (45.65,-17.65) and (38.57,-29.14) .. (31.33,-38.88) .. controls (-24.63,-114.17) and (-107.16,-100.13) .. (-100.51,-53.24) -- cycle;
\draw(120,-163.63) node[anchor=base west]{\fontsize{12}{95.60}\selectfont $\Ocal$};
\path[line width=0.30mm, draw=L] (-98.82,-68.90) -- (-201.07,339.89);
\path[line width=0.30mm, draw=L] (40.29,-25.88) -- (130.09,340.58);
\draw(216.20,121.00) node[anchor=base west]{\fontsize{12}{95.60}\selectfont $\Sigma$};
\draw(215.32,189.33) node[anchor=base west]{\fontsize{12}{95.60}\selectfont $\Sigma_+$};
\draw(215.32,61.54) node[anchor=base west]{\fontsize{12}{95.60}\selectfont $\Sigma_-$};
\path[line width=0.3mm, draw=L] (208.47,68.79) .. controls (236.63,49.56) and (201.65,-47.03) .. (248.53,-61.28) .. controls (248.53,-61.28) and (248.53,-61.28) .. (248.53,-61.28) .. controls (202.09,-76.94) and (242.56,-188.90) .. (210.02,-203.97);
\draw(251.85,-65) node[anchor=base west]{\fontsize{12}{95.60}\selectfont $\chi=1$};
\draw(228.83,270) node[anchor=base west]{\fontsize{12}{95.60}\selectfont $\chi=0$};
\path[line width=0.3mm, draw=L] (210.52,341.20) .. controls (223.47,340.00) and (205.48,288.46) .. (224.06,272.13) .. controls (224.53,271.72) and (224.03,271.89) .. (223.60,271.45) .. controls (205.11,252.22) and (223.98,208.61) .. (209.56,203.06);
\draw(-92.17,-55.55) node[anchor=base west]{\fontsize{8}{85.36}\selectfont $\supp(f)$};
\path[line width=0.30mm, draw=L, fill=F] (-166.60,204.07) -- (-132.68,68.40) -- (64.03,68.40) -- (95.91,203.39) -- cycle;
\draw(-154,170) node[anchor=base west]{\fontsize{8}{81.94}\selectfont $\supp(f-P\chi G^{\rm R}f)$};
\draw(-300,125.54) node[anchor=base west]{\fontsize{12}{95.60}\selectfont $\Ncal$};
\path[line width=0.3mm, draw=L]  (-237.50,222.87) .. controls (-250.45,221.68) and (-233.19,151.22) .. (-251.93,133.34) .. controls (-252.38,132.91) and (-252.31,133.32) .. (-251.88,132.88) .. controls (-233.39,113.65) and (-252.43,56.82) .. (-238.00,51.27);
\end{tikzpicture}%
\caption{Supports of functions relevant in construction of $\beta_+$. %The lightly shaded region is the tubular neighborhood~$\Ncal$.
\label{fig:BetaPlus}}
	\end{center}
\end{figure}
	
A similar argument works when $\supp f$ is in the future of $\Sigma_-$, but then we need to use a function $1-\chi$ in place of $\chi$ and the propagator $G^{\rm A}$ in place of $G^{\rm R}$. 
We define
\be\label{beta:minus}
\beta_-(\Orm_f)=\Orm_{f-P(1-\chi) G^{\rm A} f}
\ee
and then mimic the preceding argument.

We combine now $\beta_+$ with $\beta_-$ using the partition of unity given by $\chi$ and define
\[
\beta(\Orm_f)=\Orm_{f}-\Orm_{P(\chi G^{\rm R}(\chi f)+(1-\chi)G^{\rm A}((1-\chi)f))}\,.
\] 
By construction, $\beta$ again only modifies $\Orm_f$ by a term in the Poisson ideal.
Hence it descends to a map 
\[
\widetilde{\beta}: \mathfrak{P}_{\reg}(\Ocal) \to \mathfrak{P}_{\reg}(\Ncal)
\]
on the quotient algebras, which is equal to the identity after postcomposition with the extension map.
\end{proof}

\subsubsection{Star product}
Next we want to quantize the theory, i.e., we wish to deform the Poisson algebra $\mathfrak{P}_\reg$ to an associative algebra.
Here we use the Moyal formula:
\be\label{star1}
F_1 \star F_2 = m \circ e^{\hbar \widetilde{\partial}_{G^{\mathrm C}}} (F_1 \otimes F_2),
\ee
where $F_1,F_2\in \F_{\reg}[[\hbar]]$. 
Thus we can define the quantum situation parallel to the classical.

Equip the regular functionals with this star product: 
\[
(\F_{\reg}[[\hbar]],\star).
\] 
First, lift the product $\star$ to $\mathfrak{PV}_\reg$ by postulating that it acts trivially on the odd generators. 
Next, since $G^{\mathrm{C}}$ is a distributional bisolution for the operator $P$,
we have 
\[
\Orm_{Pf}\star F=\Orm_{Pf}\cdot F+\left<Pf, G^{\mathrm{C}} F^{(1)}\right>
=\Orm_{Pf}\cdot  F\,,
\]
where $f\in\Dcal$, $F\in\F_{\reg}$. 
It follows that $\delta_{S}$ is a derivation with respect to~$\star$:
\[
\delta_{S}(X\star Y)=(\delta_{S}X)\star  Y+(-1)^{|X|}X\star (\delta_{S} Y)\,,
\] 
where $X,Y\in\mathfrak{PV}_{\reg}$. 
\begin{df}
Define the \textbf{dg QFT model} as 
\[
\fA_{\reg}=(\fv\PV_{\reg}[[\hbar]],\star,\delta_S).
\]
\end{df}
It is straightforward to check that $\fA_{\reg}$ is a QFT model in the sense of definition~\ref{LCQFT}. Hence we obtain an on-shell model as follows. 
\begin{df}
The \textbf{on-shell algebra of regular quantum operators} is the quotient 
\[
H^0(\fA_\reg)=\fA_{\reg}/\fA_{\reg,0}
\]
by the $\star$-ideal $\fA_{\reg,0}$ generated by the elements of the form
\[
\left<dS(\ph),X(\ph)\right>\equiv \iota_{dS}X\,,
\]
where $X\in\mathfrak{PV}^1_\reg$.

Analogously
the \textbf{on-shell algebra of polynomial quantum operators} is
\[
H^0(\fA_{\pol})\doteq \fA_{\pol}/\fA_{\pol,0}\,.
\]
\end{df}

With this definition, the functor $\Ocal\mapsto H^0(\fA_\reg)(\Ocal)$ (as well as $\Ocal\mapsto H^0(\fA_\pol)(\Ocal)$) is an on-shell QFT model in the sense of Definitions~\ref{AQFT} and \ref{timeslice}. Since causality holds by construction, the only non-trivial step is to prove the time-slice axiom. This is done exactly as in Proposition~\ref{Class:mod:proof}.

The space of regular quantum operators (as a vector space) is characterized by the cohomology of the same differential as in the classical case. 
What has changed is the product. More precisely, we have
\[
\fv\fA_\reg\cong \fv\Pfrak_\reg[[\hbar]]\,.
\]

%We can lift this observation to the cochain level as follows. 

%It is also clear that $\fA_{S}^\reg$ is given as the zeroth cohomology of $(\mathfrak{PV}_\reg[[\hbar]],\delta_{S})$, so\todo{$\mathrm{Obs}_\reg^{cl}$ has not been defined yet...}
%\[
%\fv\circ\fA_{S}^\reg\cong H^0(\mathrm{Obs}_\reg^{cl}[[\hbar]])
%\]
%as vector spaces. 

%Thus, we have obtained a dg quantum model.

%\begin{df}\label{df:regclassobs}
%The \textbf{dg associative algebra of regular classical observables} is
%\[
%\fA^\reg_{\rm dg}=(\mathfrak{PV}_\reg[[\hbar]],\star,\delta_{S}).
%\]
%\end{df}

As discussed in Section \ref{sec: time ordering}, it is important to construct a time-ordered product and not just the star product,
especially as a stepping stone towards interacting theories.

On regular functionals, the time-ordered product $\T$ is introduced by means of formula
\be\label{timeordered}
F_1 \T F_2 = \alpha_{iG^{\rm D}}(\alpha_{iG^{\rm D}}^{-1}(F_1) \cdot \alpha_{iG^{\rm D}}^{-1}(F_2))\,,
\ee
i.e. by twisting the pointwise product with the map $\TT$.
This definition has the crucial property that 
\[
F_1 \T F_2 = T(F_1 \star_{G^{\rm C}} F_2)
\]
when the observables $F_1$ and $F_2$ have disjoint supports
and $T$ denotes the time-ordering
\[
T(F_1 \star_{G^{\rm C}} F_2)=
\left\{
\begin{array}{ccc}
F_1 \star_{G^{\rm C}} F_2&\mathrm{if}& F_2 \prec F_1 \\
F_2 \star_{G^{\rm C}} F_1&\mathrm{if}& F_1 \prec F_2\,, 
\end{array}
\right.
\]
Note here the connection with the Dyson formula:
$\T$ agrees with the usual time-order prescription for $\star_{G^{\mathrm C}}$ and extends it to regular functionals with overlapping supports.
\begin{rem}
Those familiar with the CG approach, notably Section 4.6 of \cite{CoGw}, will recognize that this definition is precisely the factorization product on~$\Acal$.
\end{rem}
Note that the product $\T$ is not compatible with $\delta_S$ (since $G^{\rm D}$ is a Green function, rather than a bi-solution), we introduce the off-shell models $\Pfrak^{\rm off}_\reg=(\PV_\reg,\Poi{.}{.})$ and $\fA_\reg^{\rm off}=(\fv\PV_\reg[[\hbar]],\star)$.
\begin{df}\label{timeorderedprod:reg}
	For the classical free off-shell theory $\mathfrak{P}^{\rm off}_\reg$ and its (off-shell) quantization $\fA_\reg^{\rm off}$, 
	the \textbf{time-ordered product} is realized as a quadruple $(\mathfrak{P}^{\rm off}_\reg,\fA_T,\xi,\Tcal)$ 
	with:
	\begin{itemize}
		\item  a functor 	
		\[
		\fA_T \colon \Caus(M) \to \CAlg^*(\Nuc_\hbar)\,,
		\]
		defined by $\fA_T=(\PV_{\reg}[[\hbar]],\T)$, where $\T$ is given by \eqref{timeorderedprod}.
		\item the obvious embedding 
		\[
		\xi:\fv\circ\fA_T \Rightarrow \fv\circ\fA_\reg^{\rm off}\,.
		\] 
		\item and a natural isomorphism (see lemma~\ref{lem:Tnat} ) of commutative algebras
		\[
		\Tcal: \mathfrak{c}\circ\mathfrak{P}_\reg^{\rm off}[[\hbar]] \Rightarrow \fA_T\,,
		\] 
		where $\TT \doteq \al_{iG^{\rm D}}$. (See equation~\eqref{def of alpha}.) 
	\end{itemize}
\end{df}
\begin{lemma}\label{lem:Tnat}
	The time-ordering map $\TT$ determines a natural transformation $\Tcal: \mathfrak{c}\circ\mathfrak{P}_\reg^{\rm off}[[\hbar]] \Rightarrow \fA_T$.
\end{lemma}
\begin{proof}
On each $\Ocal \in  \Caus(\Mcal)$, 
we define $\TT^\Ocal\doteq e^{\frac{i\hbar}{2}\partial_{G^{\mathrm D}_\Ocal}}$. 
This map is well-defined, since $\TT^\Ocal$ is support-preserving. 

It remains to check that $\TT$ intertwines the morphisms.
Let $\psi:\Ocal_1\rightarrow\Ocal_2$ be a morphism in $\Caus(\Mcal)$,  
let $F\in \mathfrak{PV}_\reg(\Ocal_1)[[\hbar]]$, and let $\ph$ be a scalar function on $\Ocal_2$.
Then
\begin{align*}
(\TT^{\Ocal_2} \psi(F))(\ph)
&= (e^{\frac{i\hbar}{2}\partial_{G^{\mathrm D}_{\Ocal_2}}} \psi (F))(\ph)\\
&= \sum_{n=0}^{\infty}\tfrac{1}{n!}\left(\tfrac{i\hbar}{2}\right)^n\left<(\chi_{\psi(\Ocal_1)} G^{\mathrm D}_{\Ocal_2}\circ \psi)^{\otimes n},  F^{(2n)}(\psi^*\ph)\right>\\
&=\sum_{n=0}^{\infty}\tfrac{1}{n!}\left(\tfrac{i\hbar}{2}\right)^n\left<(G^{\mathrm D}_{\Ocal_1})^{\otimes n},  F^{(2n)}(\psi^*\ph)\right>\\
&=(\psi( \TT^{\Ocal_1}F))(\ph).
\end{align*}
\end{proof}

Now it is natural to ask how the time-ordered product intertwines with the story of BV quantization.
In \cite{FR3} the  deformed BV differential has been defined as 
\be\label{intertw}
\hat{s}=\TT^{-1}\circ \delta_S \circ \TT\,,
\ee
This motivates the following.

\begin{df}\label{Obsq}
The  \textbf{BV complex of quantum observables} is
\[
\fv\fA_\reg^{q}=(\fv\mathfrak{PV}_\reg[[\hbar]],\hat{s})\,.
\]
\end{df}

Explicit computation using the properties of Green's function $G^{\mathrm{D}}$ gives
\[
\hat{s}=\delta_{S}-i\hbar \Lap\,,
\]
where $\Lap$ is the graded (or BV) Laplacian on the space of regular polyvector fields~$\mathfrak{PV}_{\reg}$.
\begin{rem}
	The name ``quantum observables'' used to describe $\fv\fA_\reg^{q}$ is justified, as this corresponds to what one would describe as such in the physics literature (e.g. \cite{HT}). Even though $\fv\fA_\reg^{q}$ is quasi-isomorphic to $\fv\Pfrak_\reg[[\hbar]]$ (by means of $\TT$), it is conventional to work with the former. One reason is that it forms a BD algebra when equipped with the usual graded pointwise product and the Schouten bracket, while $\fv\Pfrak_\reg[[\hbar]]$ forms a BD algebra with $\T$ and the bracket twisted by $\TT$ (this is explained in more detail in section~\ref{sec:main lesson})
\end{rem}

\section{Proof of comparison theorems}\label{Proofs}

Let us build up the natural transformations $\iota^{cl}$ and $\iota^q$ in stages.

\subsection{The classical case}
Momentarily ignoring the differentials, we observe that 
\[
\Sym(\Dcal(U)) = \bigoplus_{n \geq 0} \Dcal_n(U) \cong \F_\pol(U) \hookrightarrow \F_{\reg}(U).
\]
Moreover, the last inclusion is dense when we use the topology $\tau$ defined at~(\ref{def:tau}).

This relationship extends to the polyvector fields as well.
We have 
dense (with respect to $\tau$ ) inclusions as graded nuclear vector spaces:
\[
\widetilde{PV}(U) \subset PV(U)\cong \mathfrak{PV}_\pol(U) \subset \mathfrak{PV}_\reg(U) \subset \Oscr(\Ecal(U)\oplus\Ecal(U)[-1]).
\]
These inclusions are manifestly functorial with respect to opens $U$, and in particular
we see that we have a natural isomorphism
\[
\iota^{\#}: PV \Rightarrow \mathfrak{PV}_\pol
\]
between the CG and FR constructions at the level of graded vector spaces.
Let us now examine the differentials on both sides.

In this classical case, the situation is  straightforward because $\delta_S$ coincides with~$d$.
Let $\iota^{cl}$ denote the isomorphism of cochain complexes
\[
\fc\circ \Pcal\big|_{\Caus(\Mcal)}\xto{\iota^{cl}} (\mathfrak{PV}_\pol,\delta_S)=\fc\circ\Pfrak_\pol,
\]
which induces an isomorphism $H^0(\iota^{cl})$ on the zeroth cohomology.

\begin{rem}\label{reg vs pol}
In the FR framework one usually works with $\PV_{\reg}$, rather than $\PV_\pol$, since it contains also infinite sums of polynomials (e.g. Weyl generators $e^{\Ocal_f}$). Here we stress that the precise comparison between FR and CG frameworks is most naturally done for $\PV_\pol$ and one can pass to $\PV_\reg$ by appropriate completion (on both sides).
\end{rem}

\subsection{The quantum case}

The quantum case is a bit subtler. 
The FR approach assigns a dg algebra $\fA_\pol=(\mathfrak{PV}_\pol[[\hbar]], \delta_S, \star)$ 
whereas the CG approach assigns merely a cochain complex $(PV[[\hbar]], \d + \hbar \triangle)$.
On the face of it, these look rather different.
In particular, the differentials are different, 
so the embedding that works for the classical case does not extend.

The key is to use the time-ordering machinery.
Following \cite{FR3} (and by Definition \ref{Obsq}), 
the time-ordering operator provides a cochain isomorphism
\[
\fv\fA_\pol^{q}=(\mathfrak{PV}_\pol[[\hbar]],\hat{s})\xto{\TT}\fv\Pfrak_\pol[[\hbar]]=(\mathfrak{PV}_\pol[[\hbar]],\delta_{S}).
\]
To construct $\iota^q$, 
we first note that $\TT$ is also a cochain isomorphism in the CG framework, i.e.,
\be\label{alpha:GD}
\Acal\big|_{\Caus(\Mcal)}\xto{\TT} \Pcal[[\hbar]]\big|_{\Caus(\Mcal)}.
\ee
Composing with $\iota^\#$, we obtain a map
\[
\Acal\big|_{\Caus(\Mcal)}\xto{\iota^\#\circ\TT} (\mathfrak{PV}_\pol[[\hbar]],\delta_{S})=\fv\circ \fA_\pol.
\]
%Next, we observe that 
%\[\fA^{\reg}_{S}\cong H^0(\fA_{\rm dg}^{\reg})=H^0(\mathfrak{PV}_\reg[[\hbar]],\star,\delta_{S})\,%,
%\]
%so $\mathfrak{v}\circ \fA^{\reg}_{S}\cong H^0(\mathfrak{PV}_\reg[[\hbar]],\delta_{S})$. 
As in the classical case, we use the fact that there exists a canonical map from $(\mathfrak{PV}_\reg[[\hbar]],\delta_{S})$ to its $H^0$ and define $\iota^q$ as composition of
 $\iota^\#\circ\TT$ with this map. 
 Note that modulo $\hbar$, the map $\TT$ is the identity 
 and hence $\iota^q$ recovers $\iota^{cl}$ modulo~$\hbar$.
 
\begin{rem} 
\label{rem: not a fact map}
As explained in chapter 4.6 of \cite{CoGw}, the map $\TT$ in \eqref{alpha:GD} is \emph{not} a morphism of factorization algebras.
The issue arises when considering structure maps involving disjoint opens containing into a larger open;
such maps do not arise when restricted to~$\Caus(\Mcal)$.
\end{rem}

\subsection{The associative structures}\label{assocstr}

Finally, we come to the comparison of algebra structures, i.e. we prove Theorem~\ref{Comparison:alg}.

In comparing the FR and CG frameworks,
a crucial role is played by the time-ordered product. 
To understand this, observe that in trying to pass from a net to a factorization algebra, 
we need to construct a commutative product that gives rise to the factorization product structure. A natural commutative product in the pAQFT framework is $\T$. 
But going back to to non-commutative product  $\star$ given the commutative one is also easy, 
as long as we keep track of the supports of observables. 

\subsubsection{}

To communicate the key idea, we present this conversion process in the 1-dimensional case,
where the situation is simple.

In $\RR$, any interval is a causally convex neighborhood of a Cauchy surface, which in this case is given by a point in the interval. 
Let $I_0=(-a,a)\subset \RR$ be an interval with $a>0$. 
For $I_0$, we fix the point $0$ as the Cauchy surface. 
We can also consider $I_t=(t-a,t+a)$, which is a translation of $I_0$ by~$t$. 

There is natural way to identify the observables in $I_0$ with the observables in $I_t$, 
using the techniques we developed in the proof of Proposition~\ref{Class:mod:proof}.
On a linear functional $\Orm_f$ for $f \in \Dcal(I_0)$, 
let
\begin{equation}
\label{eqn:translate}
\beta^t_+(\Orm_f)=\Orm_{f-P\chi^t G^{\rm R} f}\,,
\end{equation}
where $\chi^t$ is a smooth function with the property $\chi^t(s)=1$ for $s< t-a$ and $\chi^t(s)=0$ for $s> t+a$. 
(If we fix a $\chi^0$, we can simply translate it to obtain a $\chi^t$.)
The element $\beta^t_+(\Orm_f)$ is then a linear functional with support in $I_t$.
The map $\beta^t_+$ extends in a canonical way to the whole algebra~$\F_{\reg}(I_0)$.

Now consider two arbitrary elements $A,B$ in $\F_{\reg}(I_0)$.
Set
\[
A_t\doteq \beta_+^t (A)
\]
and likewise for $B_t$.
The $\star$-commutator of $A$ and $B$ has the following relationship with the $\T$-commutator:
\[
[A,B]_\star=A\star B-B\star A=\lim_{t\to 0} (A_t\star B-B_t\star A)=\lim_{t\to 0} (A_t\T B-B_t\T A).
\]
This identification is helpful, because we know there is a nice relationship between $\T$ and the factorization product.
Namely, they agree so long the elements have disjoint support.

For $|t| > 2a$, the factorization product allows us to compute the $\T$-commutator  
As $t$ gets smaller, however, the two intervals $I_0$ and $I_t$ begin to overlap,
so that we cannot invoke the factorization product.
It is possible to resolve this issue---to describe the $\T$-commutator in terms of the factorization product---at the level of cohomology.
The key point is that for any smaller interval $I'_0 \subset I_0$, 
the inclusion $\fA(I'_0) \to \fA(I_0)$ is a quasi-isomorphism.
Any cocycle $A \in \F_{\reg}(I_0)$ can be replaced by a cohomologous element with support in the smaller interval $I'_0 \subset~I_0$.
Hence, at the level of cohomology, we can make the width $a$ of the interval arbitrarily small,
and so the $\T$-commutator can always be computed using the factorization product.
In short, at the level of cohomology, we can recover the $\star$-commutator from the factorization product.
%
%By choosing $t>2a$ we guarantee that $A_t$ and $B$ (as well as $B_t$ and $A$ ) have non-overlapping supports and using $\alpha_{iG^{\mathrm D}}$ we can relate their time-ordered product to the factorization product. This way, the factorization algebra ``knows'' about the non-commutative product $\star$ and can be used to recover the $\star$-commutator. 

\subsubsection{}

The general case is also easy to understand, as there is already a factorization algebra structure on the Cauchy surface (i.e. spacelike separated regions are taken care of) and the relation between $\star$ and the factorization product for time-like separated observables works exactly the same as in the one-dimensional case. 
We will show, in fact, something slightly more refined by working at the cochain level: we will show that the factorization product agrees with $\star$ up to exact terms.
Let us spell this out in detail now.

As discussed in Remark~\ref{rem: not a fact map}, 
the map $\TT$ does not respect the factorization product, 
but this map is an isomorphism when restricted to each open.
Hence one can use it to transfer the factorization product on the quantum observables to a new factorization product on the underlying cochain complex of the classical observables $\Pcal[[\hbar]]$.
That is, one forgets the original structure maps and borrows them from the quantum side.
Denote this new factorization algebra by~$\Acal_T$. 

As in section \ref{subsec: time slice} 
we fix a small tubular neighborhood $\widetilde{\Sigma}$ of a Cauchy surface $\Sigma$ and construct $\Acal_T\big|_{\Sigma}$. 
Now we show how to obtain a homotopy-associative product on this restricted factorization algebra.

Consider a time-slice $\Ncal_+$ in the future of $\widetilde{\Sigma}$ and disjoint from it, so  $\Ncal_+\cap \widetilde{\Sigma}=\emptyset$.
Let $\Ncal$ be a larger time-slice that contains both $\Ncal_+$ and $\widetilde{\Sigma}$. 
By the time-slice axiom, we can make all these slices arbitrarily ``thin'' in the time direction.

Take $U\subset \widetilde{\Sigma}$ a causally convex set. 
Let $J^+(U)$ denote its future.
We then have $U_+\doteq J^+(U)\cap \Ncal_+$, the ``image of $U$ in the future time-slice $\Ncal_+$.
We also have $U_\Ncal \doteq J^+(U)\cap \Ncal$, which contains both $U$ and $U_+$.
As $U_+$ does not intersect $U$, 
we have the factorization product
\[
{m_T}: \Acal_T(U_+)\otimes \Acal_T(U)\rightarrow \Acal_T(U_\Ncal)
\] 
We will recover the $\star$-product up to exact terms from this map.

As in the one-dimensional case, the formula~\eqref{eqn:translate} determines a map $\beta_+$ that transports observables to the future.
(One has to pick a partition of unity, following the proof of Proposition~\ref{Class:mod:proof}.)
Hence, given $F,G\in\Acal_T(U)$, we define
\[
F_+\doteq \beta_+(F)\in \Acal_T(U_+)
\]
and a product
\[
{m_T} \circ \beta_+ \otimes {\rm id}: \Acal_T(U)\otimes \Acal_T(U)\rightarrow \Acal_T(U_\Ncal),
\] 
which sends $F \otimes G$ to~$m_T(F_+, G)$.

We want to compare this map to $\star$.
%As $\Acal_T$ is locally constant in the time direction %and $U$ contains the Cauchy surface of $J^+(U)$, 
%we can identify  $\Acal_T(\widetilde{\Sigma})$, $\Acal_T(\Ncal_+)$ and $\Acal_T(\Ncal)$,
%%$\Acal_T(J^+(U))$ with $\Acal_T(U)$ again, 
%yielding a product map
%\[
%F \otimes G \mapsto m_T(F_+,G).
%\]
%To see that this map is the same as the $\star$ product of $\fA^\reg(\Mcal)$, 
It suffices to perform the explicit computation for $F=\Orm_f$ and $G=\Orm_g$. 
Since
\[
{\Orm_{f}}_+=\Orm_{f-P\chi G^{\rm R} f}\,,
\]
we see that
\begin{align*}
m_T({\Orm_{f}}_+,\Orm_{g})
&={\Orm_{f}}_+\cdot\Orm_{g}+\frac{i}{2}\int (f-P\chi G^{\rm R} f)(x)G^{\rm D}(x,y)g(y) \\
&= {\Orm_{f}}_+\cdot\Orm_{g}+\frac{i}{2}\int (f(x)G^{\rm D}(x,y) g(y)-
\chi(x) (G^{\rm R} f)(x)\delta(x-y)g(y))\\
&={\Orm_{f}}_+\cdot\Orm_{g}+\frac{i}{2}\int (f(x)G^{\rm D}(x,y) g(y)-
 g(x)G^{\rm R}(x,y) f(y))\\
&={\Orm_{f}}_+\cdot\Orm_{g}+\frac{i}{2}\int (f\tfrac{1}{2} (G^{\rm R}+ G^{\rm A}) g-
 fG^{\rm A} g)\\
&={\Orm_{f}}_+\cdot\Orm_{g}+\int f G^{\rm C} g\\
&=\Orm_f\star \Orm_g+\delta_S(\Orm^\dagger_{\chi G^{\rm R} f}\cdot \Orm_g)\,,
\end{align*}
where in the last step we made use of the fact that $\Orm_g$ is a cocycle.

This equation implies that $m_T({\Orm_{f}}_+,\Orm_{g})$ and $\Orm_f\star \Orm_g$ are cohomologous.
Thus, at the level of cohomology, the product $m_T \circ \beta_+ \otimes {\rm id}$ agrees with $\star$, the product on~$\fA_{\reg}(U)$.

\subsection{Shifted vs. unshifted Poisson structures}
\label{subsec: shifting}
\subsubsection{}
We discuss here a classical analogue of Theorem~\ref{Comparison:alg},
asking whether we can see a cochain-level version of the corollary.
We need a more subtle argument, 
since we need to relate the 1-shifted Poisson bracket $\{.,.\}$ with the 0-shifted bracket $\Poi{.}{.}$ on $PV(U)$, $U\in \Caus(\Mcal)$. 
As we are working with free theories,
we can exploit the fact that  $\{.,.\}$ is uniquely defined by its action on generators 
\be\label{pairing}
\{\Orm^{\ddagger}_g,\Orm_f \}\doteq \int gf \,,
\ee
where $\Orm_{g}^\ddagger\doteq \int g(x) \frac{\delta}{\delta \ph(x)}$ is a vector field in $PV^1$ and  $f,g$ live in~$\Dcal(U)$. 

The 0-shifted bracket $\Poi{.}{.}$ must live on $PV^0$, 
so in order to obtain it, we need a map from $PV^0$ to $PV^1$. 
Fortunately, the field theory naturally provides such a map, 
induced by the bisolution $G^{\rm C}$ treated as a constant bivector field on $\Ecal$. 
We denote it by $\sigma$ and write explicitly 
\[
\sigma(\Orm_f)=\iota_{d\Orm_f} G^{\rm C}=\int  G^{\rm C}(y,x) f(x) \frac{\delta}{\delta \ph(y)}=\Orm^\ddagger_{G^{\rm C}f} \,.
\]
This map does not land in $PV^1$, 
but rather in its completion (since $\int f(x) G^{\rm C}(x,y)dx$ is not compactly supported).
Nonetheless, the pairing \eqref{pairing} is still well defined on its image, 
so that we can write the new bracket on $PV^0$ by the formula
\[
\Poi{\Orm_g}{\Orm_f}=\{\sigma(\Orm_f),\Orm_g\}=\int g(y) G^{\rm C}(y,x) f(x)\,,
\] 
which is exactly the bracket of~$\Pfrak_{\rm pol}(U)$.

There is a nice interpretation of the map $\sigma$ in terms of the hyperbolic complex. 

\begin{thm}[\cite{Baer} Thm. 3.4.7]
Let $\Mcal=(M,g)$ be a connected time-oriented Lorentzian manifold with compact Cauchy surfaces. Let $P$ be a
normally hyperbolic operator acting on $\Ecal(\Mcal)$. Then the sequence of linear maps
\[0\rightarrow \Dcal(\Mcal) \xto{P}  \Dcal(\Mcal) \xto{G^{\rm C}}\Ecal(\Mcal) \xto{P} \Ecal(\Mcal) 
\]
is an exact sequence.
\end{thm}	

Clearly, the map $\sigma$ is induced by the second to last mapping in this sequence, whose image is exactly the kernel of the equations of motion operator. 

\subsubsection{}

There is yet another perspective on the Peierls bracket, related to the one presented above, but placing more emphasis on the BD algebra structure.\footnote{We thank Kevin Costello for suggesting this formulation.} 
We now proceed to formulating a precise statement by introducing some assumptions and notation.

Suppose that $\Mcal$ is foliated with compact Cauchy surfaces. 
We will call any small open neighborhood of some Cauchy surface a {\em Cauchy slice}.
Now fix such a Cauchy slice $\Ncal\subset \Mcal$.
Since solutions to the equations of motion are locally constant in the time direction (with respect to the foliation),
we can translate observables forward or backward in time.
In particular,
for a linear observable $\Orm_f\in\Pcal(\Ncal)$ localized in this slice,
we can produce an observable $\beta_+(\Orm_f)$, which is $\Orm_f$ shifted to the future, and an observable $\beta_-(\Orm_f)$, which is $\Orm_f$ shifted to the past.
(We will give an explicit formula for $\beta_{\pm}$ in the proof below.)

\begin{lemma}
For any two linear observables $\Orm_f,\Orm_g\in\Pcal(\Ncal)$ localized in $\Ncal$ and of cohomological degree zero, 
the Peierls bracket of these observables satisfies
\[
i\hbar \Pei{\Orm_g}{\Orm_f}=
\Orm_g\cdot (\beta_+(\Orm_{f})-\beta_-(\Orm_f))\, \mod\im(\hat{s}), \hbar^2.
\]
That is, we recover the Peierls bracket from the factorization product by working up to homotopy and  modulo~$\hbar^2$.
\end{lemma}

From local constancy of $\Pcal$ in the time direction,
it follows that there exists an element $\Psi$ such that
\[
\beta_+(\Orm_{f})-\beta_-(\Orm_f)=\delta_{S}\Psi.
\]
(In other words, the translated observables are cohomologous.)
Hence
\[
\Orm_g\cdot (\beta_+(\Orm_{f})-\beta_-(\Orm_f)) 
= 
\Orm_g\cdot \delta_S \Psi
=\delta_S(\Orm_g\cdot  \Psi)
=\hat{s}(\Orm_g\cdot  \Psi)+i\hbar \{\Orm_g,\Psi\}
\,.
\]
In conjunction with the lemma, we thus see that the Peierls bracket $\Pei{\Orm_g}{\Orm_f}$ is identified with $\{\Orm_g,\Psi\}$, modulo the image of $\hat{s}$ and modulo~$\hbar^2$.
This result gives a direct relationship between the BV and Peierls bracket.

\begin{proof}
To verify the relation to the usual formula for the Peierls bracket, 
we use a particular form of maps $\beta_\pm$ and the element $\Psi$. 
Consider two Cauchy slices $\Ncal_{\pm}$, disjoint away from $\Ncal$ and with the property $\Ncal_{\pm}\subset J^\pm(\Ncal)$ respectively. The arrangement of the Cauchy slices is presented on figure~\ref{betaplusminus}. 

\begin{figure}[ht]
\begin{tikzpicture}[scale=0.12,x=1.00mm, y=1.00mm, inner xsep=0pt, inner ysep=0pt, outer xsep=0pt, outer ysep=0pt]
\path[line width=0mm] (-815.17,-244.12) rectangle +(1167.03,587.43);
\definecolor{L}{rgb}{0,0,0}
\definecolor{F}{rgb}{0.753,0.753,0.753}
\path[line width=0.30mm, draw=L, fill=F] (-499.60,-26.91) [rotate around={360:(-499.60,-26.91)}] rectangle +(709.31,140.19);
\path[line width=0.30mm, draw=L] (-499.56,-241.40) rectangle +(709.19,582.72);
\path[line width=0.30mm, draw=L, fill=F] (-499.67,167.65) [rotate around={360:(-499.67,167.65)}] rectangle +(709.42,48.14);
\definecolor{F}{rgb}{0.502,0.502,0.502}
\path[line width=0.30mm, draw=L, fill=F] (-59.58,45.92) .. controls (-54.34,82.86) and (-4.53,70.89) .. (26.63,84.80) .. controls (47.28,94.01) and (70.27,109.43) .. (80.72,91.56) .. controls (86.59,81.51) and (79.51,70.03) .. (72.27,60.29) .. controls (16.30,-15.01) and (-66.22,-0.97) .. (-59.58,45.92) -- cycle;
\path[line width=0.30mm, draw=L] (-58.62,35.10) -- (-153.74,338.71);
\path[line width=0.30mm, draw=L] (80.79,73.43) -- (177.41,339.40);
\draw(231.42,184.59) node[anchor=base west]{\fontsize{12}{95.60}\selectfont $\Ncal_+$};
\draw(230.43,35.03) node[anchor=base west]{\fontsize{12}{95.60}\selectfont $\Ncal$};
\path[line width=0.3mm, draw=L] (209.60,-100.07) .. controls (214.67,-103.53) and (208.38,-120.91) .. (216.81,-123.47) .. controls (216.81,-123.47) and (216.81,-123.47) .. (216.81,-123.47) .. controls (208.45,-126.28) and (215.73,-146.42) .. (209.88,-149.13);
\draw(-616.00,269.47) node[anchor=base west]{\fontsize{12}{95.60}\selectfont $\chi=0$};
\path[line width=0.3mm, draw=L] (209.40,113.31) .. controls (222.35,112.11) and (204.36,60.57) .. (222.94,44.24) .. controls (223.41,43.83) and (222.91,44.00) .. (222.49,43.56) .. controls (203.99,24.32) and (224.56,-21.22) .. (210.13,-26.77);
\draw(-51.24,43.61) node[anchor=base west]{\fontsize{10}{85.36}\selectfont $\supp(f)$};
\path[line width=0.30mm, draw=L, fill=F] (-114.82,215.36) -- (-99.93,167.81) -- (115.23,168.10) -- (132.79,215.36) -- cycle;
\definecolor{F}{rgb}{0.753,0.753,0.753}
\path[line width=0.30mm, draw=L, fill=F] (-499.81,-149.92) [rotate around={360:(-499.81,-149.92)}] rectangle +(709.01,49.77);
\path[line width=0.30mm, draw=L] (-59.69,43.18) -- (-163.32,-242.11);
\path[line width=0.30mm, draw=L] (81.71,88.80) -- (167.72,-242.12);
\definecolor{F}{rgb}{0.502,0.502,0.502}
\path[line width=0.30mm, draw=L, fill=F] (-111.89,-100.15) -- (-129.88,-149.92) -- (143.25,-149.50) -- (131.12,-99.73) -- cycle;
\path[line width=0.30mm, draw=L, fill=F] (-351.04,-13.38) .. controls (-377.40,-2.91) and (-369.34,36.90) .. (-338.03,62.33) .. controls (-229.65,150.37) and (-138.41,42.42) .. (-187.78,-1.55) .. controls (-213.23,-24.21) and (-246.95,4.16) .. (-278.88,5.55) .. controls (-289.56,6.01) and (-300.08,3.25) .. (-309.63,-1.55) .. controls (-323.23,-8.38) and (-336.90,-18.99) .. (-351.04,-13.38) -- cycle;
\draw(-311.50,41.25) node[anchor=base west]{\fontsize{10}{85.36}\selectfont $\supp(g)$};
\draw(227.63,-132.94) node[anchor=base west]{\fontsize{12}{95.60}\selectfont $\Ncal_-$};
\draw(-616.00,-207.47) node[anchor=base west]{\fontsize{12}{95.60}\selectfont $\chi=0$};
\path[line width=0.3mm, draw=L] (209.60,216.13) .. controls (214.67,212.68) and (208.38,195.30) .. (216.81,192.74) .. controls (216.81,192.74) and (216.81,192.74) .. (216.81,192.74) .. controls (208.45,189.92) and (215.73,169.79) .. (209.88,167.08);
\path[line width=0.3mm, draw=L] (-498.70,340.73) .. controls (-509.40,339.74) and (-494.77,292.38) .. (-510.13,278.88) .. controls (-510.51,278.54) and (-510.10,278.68) .. (-509.75,278.32) .. controls (-494.47,262.42) and (-509.92,219.40) .. (-498.00,214.81);
\path[line width=0.3mm, draw=L] (-499.44,-149.94) .. controls (-507.48,-150.68) and (-497.19,-186.51) .. (-508.72,-196.65) .. controls (-509.02,-196.91) and (-508.71,-196.80) .. (-508.44,-197.08) .. controls (-496.96,-209.02) and (-508.56,-237.52) .. (-499.60,-240.96);
\path[line width=0.3mm, draw=L] (-498.57,168.05) .. controls (-523.80,165.72) and (-488.75,65.28) .. (-524.96,33.46) .. controls (-525.87,32.66) and (-524.90,32.99) .. (-524.07,32.13) .. controls (-488.03,-5.35) and (-527.47,-88.75) .. (-499.36,-99.56);
\draw(-616.00,25) node[anchor=base west]{\fontsize{12}{95.60}\selectfont $\chi=1$};
\draw(-770,215.71) node[anchor=base west]{\fontsize{12}{95.60}\selectfont $\alpha_+=1$};
\draw(-770,-140) node[anchor=base west]{\fontsize{12}{95.60}\selectfont $\alpha_-=1$};
\path[line width=0.3mm, draw=L] (-627.28,-28.09) .. controls (-646.24,-29.85) and (-619.58,-106.85) .. (-646.78,-130.76) .. controls (-647.47,-131.37) and (-646.74,-131.12) .. (-646.12,-131.76) .. controls (-619.04,-159.92) and (-651.47,-233.22) .. (-630.35,-241.34);
\path[line width=0.3mm, draw=L] (-626.51,340.17) .. controls (-647.37,338.24) and (-616.58,251.83) .. (-646.51,225.53) .. controls (-647.26,224.87) and (-646.46,225.15) .. (-645.77,224.43) .. controls (-615.99,193.46) and (-650.51,121.28) .. (-627.28,112.34);
\definecolor{L}{rgb}{0.502,0.502,0.502}
\path[line width=0.30mm, draw=L] (209.34,341.01) -- (-628.25,341.01);
\path[line width=0.30mm, draw=L] (206.98,-241.05) -- (-630.61,-241.05);
\definecolor{L}{rgb}{0,0,0}
\path[line width=0.30mm, draw=L] (-499.56,-241.40) [rotate around={0:(-499.56,-241.40)}] rectangle +(709.19,582.72);
\definecolor{L}{rgb}{0.502,0.502,0.502}
\path[line width=0.30mm, draw=L] (176.22,113.86) -- (-630.35,113.06);
\path[line width=0.30mm, draw=L] (208.16,-25.73) -- (-630.35,-28.86);
\end{tikzpicture}%
\caption{Arrangement of Cauchy slices and function supports.\label{betaplusminus}}
\end{figure}

Now let $\alpha_++\alpha_-=1$ be a partition of unity with the property that $\alpha_\pm\equiv 1$ on to the future/past of $\Ncal$. Let $\chi$ be a test function with the properties that $\chi\equiv 1$ on $\Ncal$ and $\chi\equiv 0$ to the future of $\Ncal_+$ and to the past of $\Ncal_-$. Define 
\[
1-\chi_\pm\doteq (1-\chi)\alpha_\pm
\]
and choose the maps $\beta_\pm$ as in the proof of Proposition~\ref{Class:mod:proof}, i.e.
\[ 
\beta_\pm(\Orm_f)=\Orm_{f-P\chi_\pm G^{\rm R/A} f}\,,
\]
With these choices we have
\begin{align*}
\beta_+(\Orm_{f})-\beta_-(\Orm_f)&= \int \phi P(1-\chi_+) G^{\rm R} f-\int \phi P(1-\chi_-) G^{\rm A} f\\
&=\int \phi P(1-\chi_+) G^{\rm C} f+\int \phi P(1-\chi_-) G^{\rm C} f\\
&=-\int \phi P\chi G^{\rm C} f\\
&=-\int P\phi \chi G^{\rm C} f\,,
\end{align*}
where in the last step we used the fact the $\chi$ is compactly supported, so we could integrate by parts. We now define
\[
\Psi =-\int \phi^\ddagger \chi G^{\rm C} f\,,
\]
so that
\[
\{\Orm_g,\Psi\}=\int g \chi G^{\rm C} f= \int g G^{\rm C} f\,,
\]
where we used the fact that $\supp g\subset \Ncal$ and $\chi\equiv 1$ on $\Ncal$.
Equation \eqref{Peierls} agrees with this expression, after specializing to linear observables.
\end{proof}

%\subsection{Comparison of the one-dimensional free scalar field}
%In the case $M = \RR$, many features of the comparison simplify.
%For instance, the classical field theory is described by an ordinary differential equation,
%which is much simpler than PDE.
%Similarly, the poset $\Caus(\Mcal)$ of causally-convex opens consists of intervals,
%while the poset $\Open(M)$ consists of disjoint unions of intervals,
%so that we only need to understand how disjoint union relates observables.
%Hence one might hope that a direct, concrete comparison is possible, 
%as we will show is indeed the case.
%
%The equation of motion operator for the free scalar field in 1D is $P=-\frac{d^2}{dt^2}$ and the configuration space is $\Ecal=\Ci(\RR,\RR)$. The retarded and advanced Green functions are:
%\[
%G^{\rm R}(t,s)=(t-s)\theta(t-s)\,,\qquad G^{\rm A}(t,s)=(s-t)\theta(s-t)\,,
%\]
%Hence the Dirac propagator is
%\[
%G^{\rm D}=\frac{1}{2}|t-s|\,,
%\]
%and it coincides with what would be the Euclidean Green functions (in 1D there is no distinction between the Lorenzian and Euclidean settings),
%while 
%\[
%G^{\rm C}=t-s\,.
%\]
%\kasia{What else do we want to show explicitly?}
\section{Interpretation of the results}

Now that we have precise statements and arguments in place,
it may be useful to step back and articulate what they mean.
Here we explain how our dialogue has modified our own perspective on these formalisms.

\subsection{The main lesson}\label{sec:main lesson}

The map $\TT$ used in the comparison theorems plays a double role: 
it is both a cochain isomorphism between classical and quantum observables and also an intertwiner between two products $\T$ and $\cdot$.
The take-home message is that
\medskip
\begin{quote}
	\textit{Quantum observables are described either by deforming the product (from $\cdot$ to $\T$) and keeping the differential as $\delta_S$ or, equivalently, by deforming the differential (from $s$ to $\hat{s}$) and keeping the product.}
\end{quote}
\medskip
We will now make this statement more precise. 

The approach to quantization taken in pAQFT relies on deformation of the product, while the observables are left unchanged. 
According to this philosophy, the free quantum theory is obtained by deforming $\cdot$ to the non-commutative star product $\star$.
Since $\delta_{S}$ is a derivation with respect to $\star$, 
the vector space of observables is just $\fv\Pfrak_\pol[[\hbar]]$. 
Now let's check if this is compatible with the time-ordered product $\T$. 
%We  define new prefactorization maps as follows.
%For
%$\Ocal_1,\dots,\Ocal_n\subset \Ocal$ disjoint elements of $\Caus(\Mcal) $, 
%we set
%\[
%\mathscr{F}^\TT_{\Ocal_1,\dots,\Ocal_m;\Ocal}:\fv\Pfrak_\pol(\Ocal_1)[[\hbar]]\times\dots\times\fv\Pfrak_\pol(\Ocal_n)[[\hbar]] \rightarrow \fv\Pfrak_\pol(\Ocal)[[\hbar]]
%\]
%where
%\[
%\mathscr{F}^\TT_{\Ocal_1,\dots,\Ocal_m;\Ocal}(X_1,\dots,X_n)= X_1\T\dots\T X_n\,,
%\]
%and $\mathscr{F}^\TT(\emptyset)=\CC[[\hbar]]$.
This structure does not form a differential graded commutative algebra,
since $\delta_S$  is not a derivation with respect to $\T$. 
In fact the following identity holds:
\[
\delta_{S}(X\T Y)=(-1)^{|X|}\delta_SX\T Y+X\T\delta_SY-i\hbar\{X,Y\}_\TT,
\]
where $X,Y\in\mathfrak{PV}_{\pol}[[\hbar]]$ and $\{.,.\}_\TT$ is the Schouten bracket on polyvector fields twisted by the $\T$ product, i.e.,
\[
\{X,Y\}_\TT\doteq \TT\{\TT^{-1}X,\TT^{-1}Y\}\,,
\]
with the usual Schouten bracket $\{.,.\}$. 

Since the Schouten bracket vanishes for arguments with disjoint supports, we have
\[
\delta_{S}(X\T Y)=(-1)^{|X|}\delta_{S}X\T Y+X\T\delta_{S}Y\,,
\]
for $X\in \fv\Pfrak_\pol(\Ocal_1)[[\hbar]]$,  $Y\in \fv\Pfrak_\pol(\Ocal_2)[[\hbar]]$ if $\Ocal_1,\Ocal_2\in\Caus(\Mcal) $ and $\Ocal_1\cap\Ocal_2=\emptyset$.

%Since $\{.,.\}$ is local, $\delta_S$ behaves like a derivation, if applied of products of local functionals localized in disjoint regions. This is consistent with the fact that if $F$, $G$ are local with disjoint supports, then their time-ordered product can be expressed in terms of star products of functionals with smaller supports (consequence of the additivity property).

Equivalently to deforming the product, one can deform the differential instead. 
This point of view guides the CG approach \cite{Cos,CoGw,CG2}. 
In this way of looking at things, we leave the product to be $\cdot$,  
but we deform $\delta_{S}$ to $\hat{s}$ (see \eqref{intertw}). Again we have
\[
\hat{s}(X\wedge Y)=(-1)^{|X|} \hat{s}X\wedge Y+X\wedge \hat{s}Y-i\hbar\{X,Y\}\,,
\]
so $\hat{s}$ acts like a derivation for arguments with disjoint support. 

To summarize, we identify the space of quantum observables in the FR framework with the BD algebra
\[
(\mathfrak{PV}_\pol(\Ocal)[[\hbar]],\T,\delta_S,\{.,.\}_\TT)\,,
\]
which by means of $\TT$ is quasi-isomorphic to
\[
 (\mathfrak{PV}_\pol(\Ocal)[[\hbar]],\cdot,\hat{s},\{.,.\})\,.
\]
\subsection{Yet another perspective}

Another important fact about the time-ordered product is that it essentially encodes the same combinatorics as the path integral.
In section \ref{sec: time ordering}, for instance, we discussed the Dyson series,
which displays this encoding. 

Hence, as the BV formalism was originally formulated in the path-integral approach, 
it is no surprise that in pAQFT, the BV formalism naturally appears alongside the time-ordered product. 
Formally, we can identify $\TT^H=\TT\circ \alpha_{H} $ with the convolution with the oscillating Gaussian measure of covariance $i\hbar G^{\rm F}$ (recall from Section~\ref{Propagators} that $G^{\rm F}=iG^{\rm D}+H$), i.e.
\[
\mathcal{T}^HF(\varphi)\stackrel{\mathrm{formal}}{=}\int F(\varphi-\phi)d\mu_{i\hbar\Delta_S^F}(\phi)\,.\] 
Again, formally, we would like the quantum BV operator $\hat{s}$ to fulfill 
\[
\int \hat{s}F(\varphi-\phi)d\mu_{i\hbar\Delta_S^F}(\phi)=\int \delta_S(F(\varphi-\phi)d\mu_{i\hbar\Delta_S^F})(\phi), 
\]
so by analogy
\[
\mathcal{T}^H(\hat{s}F)=\delta_S(\mathcal{T}^HF)\,.
\]
This formula suggests 
\[
\hat{s}=(\mathcal{T}^H)^{-1}\circ \delta_S \circ \mathcal{T}^H=\mathcal{T}^{-1}\circ \delta_S \circ \mathcal{T}\,,
 \]
where the last step follows from the fact that $H$ is a bisolution for the equation of motion operator $P$, so $\alpha_{H} $ commutes with $\delta_S$. 
Here we have yet another way to heuristically motivate the pAQFT definition of the quantum BV operator and its relation to the traditional BV formalism.
 
\subsection{A summary by way of a dictionary}

The following dictionary (spelled out for the free scalar field) encodes the relationships we have unraveled, 
hopefully making it easier to transfer results obtained in one approach to results in the other. 
Note that here (but not elsewhere in the paper) we use the notation of \cite{CoGw} on the CG side.

\begin{narrow}{0in}{0pt}
	\begin{center}
		\begin{longtable}{|l|l|}
			\caption[Dictionary]{Dictionary between the FR and the CG approaches for the free scalar field.\label{dictionary}}\\
			
			\hline \multicolumn{1}{|c|}{\textsc{Fredenhagen-Rejzner}} & \multicolumn{1}{c|}{\textsc{Costello-Gwilliam}} \\ \hline 
			\endfirsthead
			\multicolumn{2}{c}%
			{\emph{\tablename\ \thetable{} -- continued from previous page}} \\
			\hline \multicolumn{1}{|c|}{\textsc{Fredenhagen-Rejzner}} &
			\multicolumn{1}{c|}{\textsc{Costello-Gwilliam}} \\ \hline 
			\endhead
			
			\hline \multicolumn{2}{|r|}{{Continued on next page}} \\ \hline
			\endfoot
			
			\hline \hline
			\endlastfoot	
			\hline
			$M=(\RR^4,\eta)$, $\eta=\mathrm{diag}(1,-1,-1,-1)$&$M=(\RR^4,\1)$\\
			\hline
			\multicolumn{2}{|c|}{\textbf{The space of field configurations}}\\
			\hline
			\multicolumn{2}{|c|}{$\Ecal=\Ci(M,\RR)$}\\
			\hline
			$T\Ecal=\Ecal\times\Ecal_c$, if $\Ecal$ is equipped&  $U\subset M$, $T_c\Ecal(U)=\Ecal(U)\times\Ecal_c(U)$\\
			with the Whitney topology;	here $\Ecal_c\doteq\Ci_c(M,\RR)$&\\
			\hline
			$\F_{\pol}$&smooth/smeared observables $\Sym(\Ecal_c^!)$\\
			\hline
			%Local functionals: $\Fcal_{\loc}$&?\\
			%\hline
			%Multilocal functionals: $\Fcal$&?\\
			%\hline
			\multicolumn{2}{|c|}{\textbf{Solutions to field equations: zero locus of a 1-form $dS$ on $\Ecal$}}\\
			\hline
			$dS\in \Gamma(T^*\Ecal)$, where $T^*\Ecal=\Ecal\times \Ecal_c'$&	$dS\in \Gamma(T_c\Ecal)$\\
			\hline
			Free field equation:&Free field equation:\\ $dS(\ph)=(\Box+m^2)\ph=0$	& 	 $dS(\ph)=(\Delta+m^2)\ph=0$\\	
			\hline
			Multilocal polyvector fields: $\mathfrak{PV}_\reg(U)$& $PV(U)$ \\
			\hline
			\multicolumn{2}{|c|}{\textbf{Classical observables}}\\
			\hline
			$\fv\Pfrak_\pol(\Ocal)=(\PV_\pol(\Ocal),\delta_{S})$,& $\Pcal(U)=PV(U)$ as vector spaces,\\
			where  $\delta_{S}\doteq -\iota_{dS}$ (insertion of the 1-form $dS$)& the differential is insertion of $dS$\\
			\hline
			Feynman propagator satisfies: & $G$ is a Green's function for $\Delta+m^2$\\
			$-(\Box+m^2)\circ G^{\mathrm{F}}=-G^{\mathrm{F}}\circ (\Box+m^2)=i\delta$  & $(\Delta+m^2)\circ G=\delta$\\	
			\hline
			\multicolumn{2}{|c|}{\textbf{Wick (normal) ordering operator}}\\
			\hline
			$\TT=e^{\frac{i\hbar}{2}\Dcal_{\mathrm{F}}}$, where $\Dcal_{\mathrm{F}}=\left\langle G^{\mathrm{F}},\frac{\delta^2}{\delta\ph^2}\right\rangle$&$W=e^{\hbar \partial_G}$, where $\partial_G$ is contraction \\ &   with the Green's function $G$\\
			
			\hline
			\multicolumn{2}{|c|}{\textbf{Quantum observables}}\\
			\hline
			$\fv\fA_\pol^q(\Ocal)\doteq (\PV_\pol(\Ocal)[[\hbar]],\hat{s},\Lap)$ where&$\Acal(U) = (PV(U)[[\hbar]],\d-i\hbar \triangle,\{-,-\})$\\
			\quad $\hat{s}=\delta_{S}-i\hbar\Lap$&\\
			$\fv\fA_\pol^q$ can be equipped with a graded commutative& factorization product\\
			 product $\cdot$&\\
			\hline
			There is a map&There is a cochain isomorphism\\ \quad $\TT:\fv\fA_\pol^q(\Ocal)\rightarrow 	\fv\Pfrak_\pol(\Ocal)[[\hbar]]$&\quad$W_U:\Pcal(U)[[\hbar]]\rightarrow \Acal(U)$\\
			that intertwines the differentials,&that deforms the factorization product\\
			and induces a new product on $\fv\Pfrak_\pol[[\hbar]]$\footnote{The right-hand side lives properly in the quantum world, as $\T$ is the time-ordered version of $\star$. On the left-hand side we have quantum observables modeled by classical objects. We can therefore think about the quantization in two ways: either have a simple product, but ``complicated'' observables (LHS), or have simple observables and a complicated product (RHS).}:&as follows\footnote{The product here is denoted by $\circledast$ instead of $\star_\hbar$ of \cite{Cos,CoGw} in order to avoid the collision of notation with the non-commutatuive star product appearing in the Lorenzian case.}:\\
			$F\T G=\TT(\TT^{-1}F\cdot\TT^{-1}G)$&$\al\circledast\beta=e^{-\hbar \partial_G}\left(e^{\hbar \partial_G}\al\cdot e^{\hbar \partial_G}\beta\right)$\\
			\hline
			$\Tcal_n(\Phi(f_1),\dots,\Phi(f_n))(0)\equiv \left<G_n,f_1\otimes\dots\otimes f_n)\right>$&\\
			$G^{(n)}$ is the vev of the time-ordered product&Euclidean Green's functions\\
			of $n$ fields, i.e. the $n$-point Green's function.& (Schwinger functions)\\
			\hline
		\end{longtable}
	\end{center}
\end{narrow}

\section{Outlook and next steps}

In this paper we treated non-renormalized scalar field, so the obvious next steps are to perform renormalization and to generalize to gauge theories. 
We also discuss the possibility of incorporating the Wick rotation into our framework.

\subsection{Interacting field theories}

Renormalization becomes relevant when we introduce interactions by means of time-ordered products.
Take the free quantum theory model $\fA_\reg$. Let $V\in\F_{\reg}$ be an interaction term. 
We deform the star product $\star$ to obtain a new product on $\F_\reg[[\hbar,\lambda]]$ as
\[
F_1\star_{\la V}F_2\doteq R_V^{-1}(R_V(F_1)\star R_V(F_2))\,.
\]
This \textit{interacting} product defines a new quantum model 
\[
\fA^\reg_{\la V}(\Ocal)\doteq (\F_\reg(\Ocal)[[\hbar,\lambda]],\star_{\la V}). 
\]
We now turn to a cochain-level version of this quantum model.

The corresponding deformation of the differential is
$\hat{s}=R_V^{-1}\circ\delta_{S}\circ R_V$.
Now assume that the formal S-matrix is invariant under $\delta_{S}$, i.e.,
\[
\delta_{S}\left(e_\TT^{\frac{i}{\hbar}V}\right)=0,
\]
which is a condition equivalent to the quantum master equation (QME):
\[
\la\delta_{S}V+\frac{1}{2}\{\la V,\la V\}_{\TT}-i\hbar \Lap(V)=0\,. 
\]
When the QME holds, explicit computation shows that
\[
\hat{s}=\delta_{S}+\{.,\la V\}_{\TT}-i\hbar \Lap\,.
\]
We define the interacting quantum observables as the cochain complex
\[
(\fv\mathfrak{PV}_\reg[[\hbar]],\hat{s})\,.
\]
The renormalization problem is then to extend the analysis just outlined from regular observables to non-linear (but local) observables. 

\begin{df}\label{defloc}
A \textbf{local functional} on scalar fields is a smooth functional such that for every field $\ph\in\E$, 
there exists a positive integer $k\in\NN$ and an $f$, a compactly supported function on the jet bundle, such that
\[
F(\ph)=\int_M f(j^k(\ph))d\mu_g\,,
\]
where $j_x^k(\ph)$ is the $k$th jet of $\ph$ at point $x$ and $d\mu_g(x)\doteq \sqrt{-g}d^dx$. The space of local functionals is denoted by $\F_{\loc}$.
\end{df}
%
%\begin{df}
%	The space of \textbf{multilocal functionals} $\F$ is the algebraic completion of $\F_{\loc}$ under the point-wise product $\cdot$ defined by
%	\[
%	(F\cdot G)(\ph)\doteq F(\ph)G(\ph)\,.
%	\]
%\end{df}

In Lorentzian signature, a mathematically rigorous framework for renormalization was provided by Epstein and Glaser \cite{EG}. In \cite{FR3} this framework was combined with the BV formalism, 
allowing one to construct physically useful dg quantum models. 

In light of the results of this paper, it is natural to ask whether one can produce a factorization algebra in Lorentzian setting.
Note that \emph{classical} observables form a factorization algebra even in the Lorentzian setting, with no extra work:
solutions to the equations of motion form a sheaf---of possibly singular and infinite-dimensional manifolds, but a sheaf nonetheless---and so functions on solutions forms a factorization algebra.
We hazard the following guess about the quantization of this situation.

\begin{conj}
Epstein-Glaser renormalization determines a factorization algebra deforming the classical observables. 
The restriction to $\Caus(\Mcal)$ determines the dg quantum model of~\cite{FR3}.
\end{conj}

\begin{rem}
We hope to address the precise relation of that renormalization framework to Costello's  \cite{Cos}  in our future work.
This direction of research is potentially divergent from the conjecture above.
\end{rem}

\subsection{Lifting Wick rotation to the algebraic level}
 
The evaluation of time-ordered products of functionals at the zero field gives back the Green functions (of the Lorentzian framework). 
For instance, the $n$-point Green function is given by
$$\left<G_n,f_1\otimes\dots\otimes f_n)\right>=\Tcal_n(\Ocal_{f_1},\dots,\Ocal_{f_n})(0)\,,$$
where $f_i\in\Dcal$ and $i=1,\dots,n$. 

On the other hand, in \cite{Cos,CoGw} the factorization algebras of QFT allow one to reconstruct Schwinger $n$-point functions. In this paper we have seen that the CG approach can be also applied to the Lorentzian case directly. However, it is instructive to see how the two are connected on the level of $n$-point functions on flat spacetime.

The relation between the Euclidean and the Lorentzian framework is usually established via analytic continuation of Schwinger $n$-point functions using the Osterwalder-Schrader axioms \cite{OS73}. The relation of Schwinger functions to time-ordered products has been discussed in \cite{EE79}.
%More precisely, given $S^{(n)}(x_1,\dots,x_n)$ we want to analytically continue this finction to a function on $\CC^n$ and then the Wightman $n$-point function $W^{(n)}$ is given by
%\be\label{Wickrot}
%W^{(n)}(x_1^0,\vr{x}_1,\dots, x_n^0,\vr{x}_n )=S^{(n)}(ix_1^0,\vr{x}_1,\dots, ix_n^0,\vr{x}_n)\,,
%\ee
%where $\vr{x}_k\doteq (x_k^1,x_k^2,x_k^3)$ are the ``spatial components'' of the position vector in Minkowski space.
%	
%In the pAQFT framework $W^{(n)}$s are recovered as the $n$-point functions, i.e. 
%\[
%\left<W^{(n)},f_1\otimes\dots \otimes f_n\right>=\Ocal_{f_1}\star\dots\star \Ocal_{f_n}\,.
%\]
%Hence the position space analytic continuation of Euclidean Schwinger functions recovers the star product on the Lorentzian side, rather than the time-ordered product.
%	
%To relate $S^{(n)}$ and $\Tcal_n$ one has to go to the momentum space. On Minkowski spacetime, where Fourier transform makes sense, the Fourier transform  $\hat{T}_n$ is the analytic continuation of  $\hat{S}_n$ in the sense that
%\[
%\hat{T}_n(p_1^0,\vr{p}_1,\dots, p_n^0,\vr{p}_n) )=\hat{S}_n(ip_1^0,\vr{p}_1,\dots, ip_n^0,\vr{p}_n).
%\]
%We can summarize this in the following diagram:
%\begin{center}
%\begin{tikzcd}
%W_n\arrow[rr,"a.c."]& &S_n \arrow[d,"Fourier"]\\
%T_n\arrow[r,"Fourier"]&\widehat{T}_n\arrow[r,"a.c."]&\widehat{S}_n
%\end{tikzcd}
%\end{center}
We expect that one should be able to formulate the Wick rotation on the level of factorisation algebras (or nets).  We want to address this issue in our future work. 

\subsection{Gauge theories}

It is in the context of gauge and gravity theories that the BV formalism demonstrates its full capacities and qualities,
and it would be natural to develop analogues of the results here in those contexts.

The case of abelian gauge theories---where are free theories, albeit cohomological in nature---can be treated by almost identical methods;
renormalization is not needed.
In \cite{CoGw} there is extensive discussion of the case of pure abelian Chern-Simons theory 
and of its factorization algebra. Its  AQFT counterpart has been constructed in \cite{DMS17}. 

The work of \cite{BSS14} is also quite relevant in this context.

More generally, the BV quantization of Yang-Mills theories and effective gravity has been performed in \cite{FR,FR3,BFRej13} (based on earlier results of \cite{H}), where the appropriate dg algebras are explicitly given and the need for dg models becomes truly manifest. In these cases, to tackle nonabelian gauge theories or to couple to matter fields requires renormalization,
and so
% we anticipate that 
the methods, along the lines discussed above for interacting scalar theories, are necessary. 

We expect that comparison results, analogous to the ones obtained in the present work, will be easy to prove, provided that the renormalization schemes are shown to be equivalent.

\section*{Acknowledgments}

This project grew out of interactions at the Mathematisches Forschungsinstitut Oberwolfach in May 2016,
and we are grateful for the comfortable and stimulating atmosphere provided by the Workshop ``Factorization Algebras and Functorial Field Theories'' and our fellow participants.
Subsequently we have benefited from the chance to interact in person at the Max Planck Institute for Mathematics, Bonn, and the Perimeter Institute.
The authors would like to thank both institutions for their remarkable hospitality and financial support. This research is also supported by the EPSRC grant \verb|EP/P021204/1|.

On a more personal level, we have benefited from the insights, experience, suggestions, feedback, and questions of many people.
Clearly, Klaus Fredenhagen and Kevin Costello play a crucial role in our perspectives on QFT.
In our work on this particular paper, we have enjoyed discussions with Pavel Safronov, Theo Johnson-Freyd, Marco Benini, Alex Schenkel and participants at the conference ``Modern Mathematics of Quantum Theory'' in York.

\section*{Appendix}

\def\hotimes{\widehat{\otimes}}
\def\coker{{\rm coker}\/}
\def\colim{{\rm colim}\,}

In this appendix we prove Theorem~\ref{painintheass}.
Its statement only involves nuclear spaces, 
but our proof will detour through the setting of bornological and convenient vector spaces. 
We introduce this machinery to mimic a standard approach to proving this kind of result with finite-dimensional vector spaces.
As this result is the only place where we need this brand of functional analysis,
we isolate it here in the appendix and do not undertake any exposition of it.
We recall some key facts and notations here,
following Appendix B of~\cite{CoGw} or \cite{Book},
but a much more systematic overview of this setting can be found in~\cite{Her18}, 
which also provides useful applications to QFT.
The reader who wishes to obtain familiarity with this setting should also examine the standard reference~\cite{Michor}.

\subsection{Recollections}

Following~\cite{CoGw}, let 
\begin{itemize}
\item $LCTVS$ denote the category whose objects are locally convex Hausdorff topological vector spaces and whose morphisms are continuous linear maps,
\item $BVS$ denote the category in which each object is the underlying bornological vector space of an object in $LCTVS$ and whose morphisms are the bornological linear maps, and
\item $CVS$ denote the full subcategory of $BVS$ consisting of the $c^\infty$-complete bornological spaces.
\end{itemize}
There are functors relating these categories.
There is a functor $born: LCTVS \to BVS$ sending a topological vector space to the underlying bornological vector space,
since a continuous linear map automatically preserves bounded sets and hence is bornological.
This functor $born$ has a left adjoint $inc_\beta: BVS \to LCTVS$ that equips a bornological vector space with the finest locally convex topology with the same bounded sets.
Similarly, the inclusion $inc_c: CVS \to BVS$ has a left adjoint $c^\infty: BVS \to CVS$,
which can be viewed as a kind of completion.
(See Theorems 2.15 and 4.29 of~\cite{Michor}. A more extensive treatment is Chapter 2 of \cite{FroKri},
but note they use ``preconvenient'' for what we call bornological.)

\begin{lemma}
A locally convex vector space is convenient if and only if it is Mackey complete.
If $V \in LCTVS$ is complete, then its underlying bornological space $born(V)$ is $c^\infty$-complete.
\end{lemma}

\begin{proof}
The first statement is Theorem 2.14 of~\cite{Michor}.
Lemma 2.2 of~\cite{Michor} shows that  implies Mackey-completeness,
and a complete topological vector space is sequentially complete.
\end{proof}

This result ensures  that the topological vector spaces arising from differential geometry, such as 
\begin{itemize}
\item smooth sections of a vector bundle,
\item compactly supported smooth sections of a vector bundle,
\item distributional sections of a vector bundle, and
\item compactly supported distributional sections of a vector bundle,
\end{itemize}
are {\em convenient} vector spaces, when one considers their underlying bornological vector space.

In the category $BVS$, the algebraic tensor product $V \otimes_{alg} W$ can be equipped with a natural bornology that co-represents bornological bilinear maps $f: V \times W \to U$;
we denote the associated bornological space by $V \otimes_\beta W$.
The category $CVS$ inherits a natural tensor product by completing the bornological tensor product:
for $V,W \in CVS$, we have
\[
V \hotimes_\beta W = c^\infty( inc_c(V) \otimes_\beta inc_c(W)).
\]
This tensor product has an appealing property when applied to our favorite examples.

\begin{lemma}[\cite{Her18}, Proposition 2.3.13]
Let $E \to X$ and $F \to Y$ be finite-rank vector bundles on smooth manifolds (not necessarily compact).
Let $E \boxtimes F \to X \times Y$ denote the vector bundle obtained by tensoring the pullback to $X \times Y$ of the bundle $E$ with the pullback to $X \times Y$ of the bundle $F$.
Then
\[
\Gamma_c(X, E) \hotimes_\beta \Gamma_c(Y,F) \cong \Gamma(X \times Y, E \boxtimes F)
\]
in $CVS$. 
In particular, $\Dcal(X) \hotimes_\beta \Dcal(Y)\cong \Dcal(X \times Y)$.
\end{lemma}

This result is a consequence of more foundational results.

\begin{prop}[\cite{Michor}, Proposition 5.8]
If $V,W \in LCTVS$ are metrizable, then 
\[
V \otimes_\pi W = V \otimes_\beta W,
\]
where $\otimes_\pi$ denotes the projective tensor product.
\end{prop}

As a corollary, we have the following crucial result. 

\begin{cor}
\label{cor:pivsbeta}
For $V,W$ Fr\'echet spaces, $V \hotimes_\pi W = V \hotimes_\beta W$.
\end{cor}

Here it is important to note that a Fr\'echet space is convenient, by Theorem 4.11 of~\cite{Michor}.

\begin{proof}
By Result 52.23 of \cite{Michor}, we know that in a metrizable locally convex vector space, 
the convergent sequences coincide with the Mackey-convergent sequences. 
Hence, the usual metric space completion agrees with the Mackey-completion and thus with the $c^\infty$-completion. 

Taking the metric space completion of $V \otimes_\pi W$ is thus equivalent to the completion to a convenient vector space. Because $V \otimes_\pi W = V\otimes_\beta W$, we obtain the result.
\end{proof}

This result has important consequences for LF spaces, i.e., those topological vector spaces that arise as a sequential colimit of Fr\'echet spaces.
For example, the compactly supported smooth functions $\Dcal(M)$ is typically topologized as an $LF$ space, as follows. 
Pick a sequence of compact subsets $M_1 \subset M_2 \subset \cdots$ such that $M = \bigcup_i M_i$ (an ``exhaustive'' sequence). Then $\Dcal$ is the sequential colimit of the Fr\'echet spaces $C^\infty(M_i)$, the smooth functions on $M$ whose support is contained in~$M_i$.
Note that for $N$ another manifold with an exhaustive sequence of compact subsets $N_i$,
the products $M_i \times N_i$ provide an exhaustive sequence for $M \times N$.
Hence
\[
\Dcal(M \times N) = \colim \Dcal( M_i \times N_i)
\]
but we know 
\[
\Dcal( M_i \times N_i) =  \Dcal( M_i) \hotimes_\pi \Dcal(N_i) = \Dcal( M_i) \hotimes_\beta \Dcal(N_i).
\]
Thus 
\[
\Dcal(M \times N) = \colim \Dcal( M_i) \hotimes_\beta \Dcal(N_i) = \Dcal( M) \hotimes_\beta \Dcal(N)
\]
since completion is a left adjoint and hence preserves colimits.
This argument justifies the lemma we opened with, when revised to encompass sections of vector bundles.

\subsection{The theorem and the proof strategy}

We now state a slightly different and more informative version of the theorem.

\begin{thm}
The commutative dg algebra $(PV(M), \d)$ of classical observables on a globally hyperbolic space $\Mcal=(M,g)$ determines a cochain complex of convenient vector spaces,
since in each cohomological degree the underlying bornological vector space is already convenient.
The differential is unchanged, as well.

This cochain complex in $CVS$ has cohomology 
\[
H^*(PV(M), \d) = 
\begin{cases} \Sym_{\widehat{\otimes}_\beta}(V'), & * = 0 \\ 0, & * \neq 0 \end{cases}
\]
where $V' = \Dcal/P(\Dcal)$ and where
\[
 \Sym_{\widehat{\otimes}_\beta}(V') = \bigoplus_{n=0}^\infty ((V')^{\widehat{\otimes}_\beta n})_{S_n},
\]
the symmetric algebra of $V'$ as a convenient vector space.
\end{thm}

Implicit in the statement of the theorem is that $V'$ is already convenient when viewed as a bornological space. 
For this claim, see lemma~\ref{vprimeisLF} below.

We now explain why this theorem implies Theorem~\ref{painintheass}.
The inclusion functor $inc_\beta \circ inc_c: CVS \to LCTVS$ preserves colimits
and hence cokernels.
As the cochain complex is concentrated in nonnegative degrees and the cohomology is concentrated in degree zero,
the only nontrivial cohomology is a cokernel.
Hence this computation in $CVS$ determines the cohomology as a cochain complex in $LCTVS$,
which implies Theorem~\ref{painintheass}.

We now outline the steps for proving this theorem,
which we pursue in the rest of the appendix.

Because the differential preserves the symmetric powers, 
the cochain complex decomposes as a direct sum of cochain complexes
\[
(PV(M), \d) = \left( \bigoplus_{n=0}^\infty PV_n(M), \d_n \right),
\]
where $PV_n(M)$ denotes the graded vector space built from distributions on $M^n$,
which we view as the polynomial observables that are homogeneous of degree~$n$.
To prove the theorem,
it thus suffices to prove that for each summand, 
the cohomology is concentrated in degree $0$ and is isomorphic to~$\Sym_{\otimes_\beta}^n(V')$.

The argument is by induction on $n$.
The base case $n=1$ is a fact of analysis:
we need to know that 
\[
0 \to \Dcal \xto{P} \Dcal \to V' \to 0
\]
is a short exact sequence.
The induction step, however, boils down to homological algebra,
notably a version of the acyclic assembly lemma.
This classic result says that the tensor product of two acyclic complexes is again acyclic,
and we will develop a version in our situation.

\subsection{The base case and the convenient structure of~$V'$}

Recall the following result.

\begin{prop}[BGP, Theorem 3.4.7]
Let $M$ be a globally hyperbolic manifold.
Let $P$ be a Green-hyperbolic operator, and let $G = G_+ - G_-$ denote the difference of its advanced and retarded Green's operator.
There is an exact sequence 
\[
0 \to \Dcal \xto{P} \Dcal \xto{G} \Ecal \xto{P} \Ecal
\]
of nuclear spaces.
\end{prop}

From it, we now deduce the base case.

\begin{lemma}\label{lemma:BGP}
Let $M$ be a globally hyperbolic manifold, and let $V$ denote the kernel of a Green-hyperbolic operator $P: \Dcal' \to \Dcal'$ on $M$. Then
\[
0 \to \Dcal \xto{P} \Dcal \to V' \to 0
\]
is a short exact sequence of nuclear spaces.
\end{lemma}

\begin{proof}
The exactness on the left side follows immediately from the preceding proposition.
We now turn to the right hand side.

As a first step, we want to see that there is a natural map $\Dcal/P(\Dcal) \to V'$.
Note that the image of $P$ on $\Dcal$ is closed by the closed graph theorem of Section 17 of Tr\`eves,
so this cokernel $\Dcal/P(\Dcal)$ is nuclear, by Proposition 50.1 of Tr\`eves. 
Since $V \subset \Dcal'$, there is a canonical restriction map $r: \Dcal\to V'$.
By the Hahn-Banach theorem, the restriction map is surjective.
Moreover, since $P$ annihilates $V$, the kernel of the restriction map $r$ contains the image $P(\Dcal) \subset \Dcal$ and so this restriction map factors through the cokernel $\Dcal/P(\Dcal)$. 
Hence we have a canonical surjective continuous linear map 
\[
\bar{r}: \Dcal/P\Dcal \to V'.
\]
To finish the lemma, we need to show $\bar{r}$ is an isomorphism.

It remains to show $\bar{r}$ is injective.
To see this, we need to show that if $r(f)= 0$, there exists $g \in \Dcal$ such that $f = Pg$.
But by [BGP, Theorem 3.4.7],
$f \in \im(P)$ if and only if $f \in \ker(G: \Dcal \to \Ecal)$,
so it suffices to show that $Gf = 0$.
Consider $G^*: \Ecal' \to \Dcal'$, the dual to $G: \Dcal \to \Ecal$,
and note that under the continuous inclusion $i: \Dcal \hookrightarrow \Ecal'$,
we find
\[
G^*i(h) = -Gh
\]
for any $h \in \Dcal$.
By [BGP, Theorem 3.4.7], $P(Gh) = 0$, and so for any $h \in \Dcal$, we find $P(G^* i(h)) = 0$ and hence $G^*i(h) \in V$. 
Thus $G^*i(\Dcal)$ is contained in $V$.
Hence for any $f$ such that $r(f) = 0$ and $h \in \Dcal$, we have
\[
0 = \langle r(f), G^* i(h) \rangle = \langle Gf, h \rangle,
\] 
where the first dual pairing is between $V$ and $V'$ and the second is between $\Dcal'$ and $\Dcal$,
which implies $Gf = 0$, as needed.
\end{proof}

We note that this exact sequence in $LCTVS$ provides, without change, an exact sequence in $BVS$.
We know, moreover, that $\Dcal(U)$ is convenient, so that the beginning of the sequence sits in $CVS$.
It is not immediately obvious, however, that $V'$ is convenient.

\begin{lemma}
\label{vprimeisLF}
The cokernel $V'$ is a nuclear LF space. 
As such, it is complete as a topological vector space and hence convenient.
In other words, 
\[
0 \to \Dcal \xto{P} \Dcal \to V' \to 0
\]
is a short exact sequence in $CVS$.
\end{lemma}

\begin{proof}
Recall that $\Dcal$ is a nuclear LF space,
as shown in the corollary just before Theorem 51.6 of~\cite{Tre67}.
Let us describe the structure explicitly.
Pick a countable sequence of compact subsets $M_1 \subset M_2 \subset \cdots$ such that $M = \bigcup_i M_i$. Then $\Dcal$ is the sequential colimit of the nuclear Fr\'echet spaces $\Dcal(M_i)$, the smooth functions on $M$ whose support is contained in $M_i$. A differential operator, such as $P$, preserves support, so that $P$ sends $\Dcal(M_i)$ to itself. Hence the action of $P$ on $\Dcal$ is the colimit of $P$ acting on this sequence. The image of $P$ on $\Dcal(M_i)$ is closed and so the cokernel $\Dcal(M_i)/P(\Dcal(M_i))$ is Fr\'echet. As colimits commute, the cokernel of $P$ on $\Dcal$ is the colimit of these cokernels. Hence this cokernel is also $LF$.
As an $LF$ space is complete as a locally convex topological vector space (Theorem 13.1 of \cite{Tre67}),
it is convenient, since completeness as a topological vector space is stronger than $c^\infty$ completeness. 
\end{proof}

\subsection{The induction step}

The essential fact we need to show is the following.

\begin{prop}\label{prop:nfold}
For any positive integer $n$, 
there is a canonical quasi-isomorphism
\[
(0 \to \Dcal \xto{P} \Dcal \to 0)^{\hotimes_\beta n} \xto{\simeq} (V')^{\hotimes_\beta n}
\]
between the $n$-fold completed projective tensor product of the complex for the operator $P$ and the $n$-fold completed projective tensor product of $V'$, the continuous linear functionals on distributional solutions $V$ to~$P$.
\end{prop}

This result is the crucial step toward the result we really want.

\begin{thm}
There is a canonical quasi-isomorphism
\[
\Sym^n_\beta(\Dcal[1] \xto{P} \Dcal) \xto{\simeq} \Sym^n_{\beta}(V')
\]
where $\Sym^n_\beta$ denotes the $S_n$-invariant subspace of the $n$-fold completed bornological tensor product, where $S_n$ acts by permuting the tensor factors.
\end{thm}

Note that the (continuous) action of $S_n$ decomposes the (completed) $n$-fold tensor product into a direct sum of factors, each of which is a cochain complex. 
By considering the invariant component, we obtain the theorem from the proposition.

The linchpin for proving the proposition is the following result, whose proof is modeled on the acyclic assembly lemma (see Lemma 2.7.3 of \cite{Weibel}) and is deferred to the end of the appendix.
Recall that a complex is {\em acyclic} if its cohomology is zero.

\begin{lemma}
\label{acyclic}
Let  $A^\bullet$ be an acyclic bounded complex in $CVS$ such that each space $A^k$ is a nuclear LF space.
Let $B^\bullet$ be another bounded complex in $CVS$, where each component $B^k$ is a nuclear LF space.
Then the total complex of the tensor product $A^\bullet \hotimes_\beta B^\bullet$ is an acyclic bounded complex in~$CVS$.
\end{lemma}

As a corollary of this lemma, we find a useful statement by induction on $n$, with Lemma~\ref{lemma:BGP} as the base case.

\begin{cor}\label{lem:induction}
For any positive integer $n$, the $n$-fold completed projective tensor product
\[
(0 \to \Dcal \xto{P} \Dcal \to V' \to 0)^{\hotimes_\beta n}
\]
is acyclic.
\end{cor}

The proof is by induction on $n$, with Lemma \ref{lemma:BGP} as the base case.

To show proposition~\ref{prop:nfold}, however,
we need one further observation about acyclic complexes.

For any acyclic complex $A^\bullet$ that is bounded above, such as
\[
\cdots \to A^{k-2} \xto{\d} A^{k-1} \xto{\d} A^k \to 0 \to 0 \to \cdots,
\]
consider the truncation that keeps only the components in degrees below~$k$:
\[
\tau^{< k} A^\bullet = \cdots \to A^{k-2} \xto{\d} A^{k-1} \to 0,
\]
and also consider the truncation that only keeps components in degrees above~$k-1$:
\[
\tau^{\geq k} A^\bullet = \cdots \to 0 \to A^k \to 0 \to \cdots
\]
consisting of $A^k$ in degree $k$, since $A^k$ is the last nonzero component.
There is a natural cochain map
\[
\tau^{< k} A^\bullet \xto{\phi_\d} (\tau^{\geq k} A^\bullet)[1]
\]
that in degree $k-1$ consists of $A^{k-1} \xto{d} A^k$ and that is the zero map in all other degrees.
This cochain map is a quasi-isomorphism by inspection.

In our situation, Lemma~\ref{lemma:BGP} gives us the bounded acyclic complex,
and so we see
\begin{equation}\label{step1}
(0 \to \Dcal \xto{P} \Dcal \to 0) \xto{\simeq} V'
\end{equation}
the desired quasi-isomorphism when $n=1$.

Moreover, by lemma~\ref{acyclic} and taking $B = V'$, 
we see that
\[
(0 \to \Dcal \xto{P} \Dcal \to V' \to 0) \hotimes_\beta V'
\]
is acyclic and hence that there is a quasi-isomorphism
\begin{equation}\label{step2}
(0 \to \Dcal \hotimes_\beta V' \xto{P \hotimes_\beta \id} \Dcal \hotimes_\beta V'\to 0)   \xto{\simeq} V' \hotimes_\beta V'.
\end{equation}
We have almost obtained proposition~\ref{prop:nfold} for $n=2$.

Consider now the double complex produced by the tensor product
\[
(0 \to \Dcal \xto{P} \Dcal \to V' \to 0) \hotimes_\beta (0 \to \Dcal \xto{P} \Dcal \to 0).
\]
Each ``row'' $(0 \to \Dcal \xto{P} \Dcal \to V' \to 0) \hotimes_\beta \Dcal$ is acyclic by lemma~\ref{acyclic}, 
so we can use our truncation trick to produce a quasi-isomorphism
\begin{equation}\label{step3}
(0 \to \Dcal \xto{P} \Dcal \to 0)^{\hotimes_\beta 2} \xto{\simeq} (0 \to \Dcal \xto{P} \Dcal \to 0) \hotimes_\beta V'.
\end{equation}
(This amount to truncating at the rightmost ``column,'' which is the one produced by tensoring with $V'$.)
Putting equations \eqref{step1}, \eqref{step2}, and \eqref{step3} together, we have the proposition for~$n=2$.

We now record this argument in general.

\begin{proof}[Proof of proposition~\ref{prop:nfold}]
We use induction on~$n$.

Lemma \ref{lemma:BGP} implies 
\[
(0 \to \Dcal[1] \xto{P} \Dcal \to 0) \xto{\simeq} V'
\]
is a quasi-isomorphism,
which is the base case~$n=1$.

\def\Cone{{\rm Cone}}

Denote the cochain complex $(0 \to \Dcal \xto{P} \Dcal \to 0)^{\hotimes_\beta n}$ by $C_n^\bullet$.
As the induction step, suppose we have a quasi-isomorphism
\begin{equation}
\label{indhyp}
\Phi_n: C_n^\bullet \xto{\simeq} (V')^{\hotimes_\beta n}.
\end{equation}
where the map is zero except in degree zero, 
where it goes from $\Dcal^{\hotimes_\beta n} \to (V')^{\hotimes_\beta n}$.
We wish to produce a quasi-isomorphism
\[
\Phi_{n+1}: C_{n+1}^\bullet \xto{\simeq} (V')^{\hotimes_\beta n+1}.
\]
Consider the cone of the map~\eqref{indhyp},
\[
\Cone(\Phi_n) = (C_n^\bullet[1] \xto{\Phi_n} (V')^{\hotimes_\beta n}),
\]
which is acyclic because $\Phi_n$ is a quasi-isomorphism.
This acyclicity assures us that $C_n^\bullet \hotimes_\beta V'$ is also acyclic,
by the arguments used for lemma~\ref{acyclic}.
Hence we can extract a quasi-isomorphism
\[
C^\bullet_n \hotimes_\beta V' \xto{\simeq} (V')^{\hotimes_\beta n+1}.
\]
It remains to exhibit a quasi-isomorphism
\[
C^\bullet_n \hotimes_\beta V' \xto{\simeq} C^\bullet_{n+1}.
\]
To see this,  note that $(0 \to \Dcal[1] \xto{P} \Dcal \to V' \to 0)$ is acyclic and so the tensor product
\[
(0 \to \Dcal[1] \xto{P} \Dcal \to V' \to 0) \hotimes_\beta (0 \to \Dcal[1] \xto{P} \Dcal \to 0)^{\hotimes_\beta n}
\]
is acyclic,
from which one can read off the desired quasi-isomorphism.
\end{proof}

\subsection{Proof of lemma~\ref{acyclic}}

We now prove the key fact about acyclic complexes in a series of small steps.
First, we examine how kernels and cokernels behave in this setting,
since cohomology is a quotient of a kernel.

\begin{lemma}
The completion functor $c^\infty: BVS \to CVS$ is a left adjoint and hence preserves colimits.
In particular, if $A: V \to W$ is a bornological linear map, then $c^\infty(W/A(V)) \cong c^\infty(W)/c^\infty(A(V))$.
That is, cokernels are preserved.
\end{lemma}

The case of kernels is more subtle, and so we prove the relevant fact only in the setting we need.

\begin{lemma}
Let $Q: A \to A'$ be a continuous linear operator between nuclear LF spaces. 
Let $B$ be another nuclear LF space.
Then there is a canonical isomorphism 
\[
(\ker Q) \hotimes_\beta B \cong \ker(Q \hotimes_\beta \id_B)
\]
of nuclear LF spaces, and hence of convenient vector spaces.
\end{lemma}

\begin{proof}
Fix a sequence of Fr\'echet spaces $A_1 \to A_2 \to \cdots $ such that $A = \colim A_k$,
and similarly sequences of Fr\'echet spaces $A'_1 \to A'_2 \to \cdots $ and $B_1 \to B_2 \to \cdots $
such that $A' = \colim A'_k$ and such that $B = \colim B_k$.
We note that $A \hotimes_\beta B$ is $\colim A_k \hotimes_\beta B_k$, since we can combine the colimits and then take the ``diagonal'' sequence therein.
Likewise,  $A' \hotimes_\beta B = \colim A'_k \hotimes_\beta B_k$

Let $K$ denote the null space of $Q$ on $A$, 
and let $K_k$ denote the null space of $Q$ restricted to $A_k$.
Let $L$ denote the null space of $Q \hotimes_\beta \id_B$ on $A \hotimes_\beta B$, 
and let $L_k$ denote the null space of $Q \hotimes_\beta \id_B$ restricted to $A_k \hotimes_\beta B$.
We wish to show that $L$ is isomorphic to $K \hotimes_\beta B$.

Note that $K \hotimes_\beta B$  is the colimit of $K_k \hotimes_\beta B_k$, since we can combine the colimits and then take the ``diagonal'' sequence therein. It is useful to work with this sequence.

For instance, consider the canonical map of short exact sequences
\[
\begin{array}{ccccccc}
0 & \to & K_k \otimes_\beta B_k & \to & A_k \otimes_\beta B_k & \xto{Q \otimes \id} & A'_k \otimes_\beta B_k\\
& & \downarrow &  &  \downarrow &  &  \downarrow\\
0 & \to & L_k & \to & A_k \hotimes_\beta B_k & \xto{Q \hotimes_\beta \id} & A'_k \hotimes_\beta B_k
\end{array}
\]
that arises from including the uncompleted tensor products into the completed tensor products. 
The two vertical maps on the right are dense inclusions, 
which implies the leftmost vertical map $f_k: K_k \otimes_\beta \Dcal(N_k) \to L_k$ is dense too.
Now the completion-inclusion adjunction $c^\infty \dashv inc_c$ tells us that we obtain a canonical map
\[
\widehat{f}_k: K_k \hotimes_\beta \Dcal(N_k) \to L_k,
\]
through which the map $f_k$ factors.
By the argument for corollary~\ref{cor:pivsbeta}, the convenient completion agrees with the completion as a topological vector space,
so this map $\widehat{f}_k$ is an isomorphism.

As we have exhibited an isomorphism for every index $k$, we have an isomorphism between the sequential colimits, as desired.
\end{proof}

We can now prove the simplest version of lemma~\ref{acyclic}.

\begin{lemma}\label{cor:acyclic}
Let  $A^\bullet$ be an acyclic bounded complex in $CVS$ such that each space $A^k$ is a nuclear LF space.
Let $B$ be  a nuclear LF space.
Then the complex $A^\bullet \hotimes_\beta B$ is an acyclic bounded complex in~$CVS$.
\end{lemma}

\begin{proof}
Let $d^n: A^n \to A^{n+1}$ denote the differential in $A^\bullet$.
For each $n$, we have an isomorphism ${\rm im}(d^n) \xto{\cong} \ker(d^{n+1})$ because $A^\bullet$ is acyclic.
By the preceding lemma, we know
\[
\ker(d^{n+1}) \hotimes_\beta B = \ker(d^{n+1} \hotimes_\beta \id_B),
\]
so we just need to show that 
\[
d^n \hotimes_\beta \id_B(A^n \hotimes_\beta B) = d^n(A^n) \hotimes_\beta B
\]
to verify that $A^\bullet \hotimes_\beta B$ is acyclic.
But this follows from that fact that images are cokernels and $-\hotimes_\beta B$ preserves cokernels.
\end{proof}

With this result in hand, we can prove our acyclic assembly lemma. 

\begin{proof}[Proof of lemma~\ref{acyclic}]
We manifestly obtain a bounded cochain complex of nuclear LF spaces (and hence in $CVS$),
since the total complex of a tensor product of bounded complexes is bounded.
Hence it remains only to prove the acyclicity.

Let $P$ denote the differential on $A$ and $Q$ the differential on $B$. 
Consider the $n$th degree component of the tensor product:
it is a direct sum
\[
\bigoplus_{i+j = n} A^i \hotimes_\beta B^j,
\]
and its differential is
\[
\d^n = \sum_{i+j = n} P^i \otimes \id_{B^j} + \id_{A^i} \otimes Q^j.
\]
Any degree $n$ element $x$ has the form
\[
\sum_{i+j = n} x_{ij}
\]
where $x_{ij} \in A^i \hotimes_\beta B^j$.

Note that for each degree $k$, the complex $A^\bullet \hotimes_\beta B^k$ is acyclic,
by lemma~\ref{cor:acyclic}.
Here the differential is just $P \hotimes_\beta \id_{B^k}$.

Choose an element $x$ in the kernel of the total differential.
We wish to show it equals $\d^{n-1}(y)$ for some element $y$ of total degree~$n-1$.
The key is to use a kind of staircase or zigzag argument.
(In other words, we now recapitulate the argument of the usual acyclic assembly lemma, 
but in this $CVS$ setting.)

Due the boundedness, there is some biggest index $K$ such that $A^k = 0$ for all $k > K$ but $A^K \neq 0$.
Hence, we know that the component
\[
x_{K(n-K)} \in A^{K} \hotimes_\beta B^{n-K},
\]
is annihilated,
\[
P^{K}\hotimes_\beta \id_{B^{n-K}} (x_{K(n-K)}) = 0,
\]
since $A^{K+1} = 0$.
By the assumption of acyclicity, there is some element $y_{(K-1)(n-K)} \in A^{K-1} \hotimes_\beta B^{n-K}$ such that
\[
P^{K-1}\hotimes_\beta \id_{B^{n-K}} (y_{K-1,n-K})=  x_{K(n-K)}.
\]
(From now on we will drop indices from the differentials $P$ and $Q$ but keep them on the elements.)
Because $x$ is in the kernel of the total differential,
we know 
\[
\id_A \hotimes_\beta Q(x_{K(n-K)}) = P\hotimes_\beta \id_{B}(x_{(K-1)(n-K+1)}),
\]
which implies that 
\begin{equation}\label{eq:xy}
P\hotimes_\beta \id_{B}(x_{(K-1)(n-K+1)}) = P \hotimes_\beta Q (y_{(K-1)(n-K)} ).
\end{equation}
This equation \eqref{eq:xy} implies that the difference $x_{(K-1)(n-K+1)} -  \id_A \hotimes_\beta Q (y_{K-1,n-K})$ is in the kernel of $P\hotimes_\beta \id_{B}$.
By acyclicity with respect to $P$, one can again produce a preimage $y_{(K-2)(n-K+1)}$.
By hypothesis, the sum $y_{(K-1)(n-K)} + y_{(K-2)(n-K+1)}$ has the property that its image under the total differential is
\[
x_{K(n-K)} + x_{(K-1)(n-K+1)} + \id \otimes Q(y_{(K-2)(n-K+1)}).
\]
Hence we have produced a partial trivialization of~$x$. 

By repeating the process just shown, we can work further up the staircase, 
introducing terms $y_{(j-1)k}$ until we produce $y$ such that $x = \d^{n-1}(y)$.
This process terminates because the double complex is bounded.
This shows the completed tensor product is acyclic as well.
\end{proof}

\end{document}